\documentclass[preprint,12pt,nopreprintline]{elsarticle}

\usepackage{amssymb}
\usepackage{amsmath}
\usepackage{amsthm}
\usepackage{mathtools}
\usepackage{mathrsfs}
\usepackage{tcolorbox}
\usepackage[cm]{fullpage}
\usepackage{hyperref}
\usepackage{orcidlink}
\newtheorem{theorem}{Theorem}

\begin{document}

\begin{frontmatter}

\title{Entropy and Holography through Adjunctions: A Bicategorical Perspective on Landauer's Principle}

\author[1]{Petr Vlachopulos\corref{cor1}\,\orcidlink{0009-0001-6410-7275}}
\ead{pvlachopulos@gmail.com}

\cortext[cor1]{Corresponding author}

\ead[url]{https://www.researchgate.net/profile/Petr-Vlachopulos?ev=hdr_xprf}

\affiliation[1]{organization={Department of Mathematics and Statistics, Masaryk University, Faculty of Science},
                addressline={Kotlarska 2}, 
                city={Brno},
                postcode={611 37}, 
                country={Czech Republic}}

\begin{abstract}

We develop a bicategorical formulation of entropy and Landauer's principle in which deterministic monotone maps between logical and thermodynamic state spaces are embedded into a broader theory of open, many-to-many feasibility interfaces. Entropy-ordered state spaces form the objects of a locally posetal bicategory, Boolean profunctors encode feasible implementations, and refinements define 2-morphisms. Although this bicategory admits genuinely non-representable interfaces, every exact Landauer adjunction lies in its representable sector and is generated by a unique monotone implementation map. Such an adjunction also induces an idempotent boundary monad, together with a dual idempotent bulk comonad. Our main result identifies the Eilenberg-Moore object of the boundary monad with a universal boundary-stable quotient determined by the bulk-mediated round-trip relation. This provides a canonical reconstruction principle: the quotient captures exactly those boundary distinctions that remain observable after implementation in the bulk and subsequent readout. In the representable case, two boundary states are identified precisely when they have the same image under the implementation map. Therefore, the result turns irreversible information loss into an explicit universal object and provides a compositional framework for studying reconstruction, coarse-graining, and bulk-boundary visibility in open thermodynamic systems. Finally, we outline a potential cost-enriched extension in which interfaces carry dissipation costs and composition minimizes accumulated dissipation over admissible intermediate realizations, opening a route from Boolean feasibility to quantitative thermodynamic implementation.

\end{abstract}

\begin{keyword}

Entropy \sep Thermodynamic Entropy \sep Bicategories \sep Categorification of Entropy \sep Landauer's Principle \sep Landauer Adjunction \sep Holographic Principle \sep Eilenberg-Moore Construction
\end{keyword}

\end{frontmatter}

\newcommand\restr[2]{{
  \left.\kern-\nulldelimiterspace 
  \littletaller
  \right|_{#2} 
  }}

\section{Introduction}

The operational content of the second law of thermodynamics can often be expressed in the form of a constraint on which transformations of thermodynamic states are feasible and on what must be dissipated when information is manipulated.

Landauer's principle made this interconnection explicit within the setup of irreversible computation, i.e. erasing information has inevitable thermodynamic consequences in the form of thermodynamic costs. Such costs are typically described as heat dissipation to the surrounding environment \cite{Landauer1961}.   
In Carathéodory-type formulations, under suitable regularity and integrability assumptions, the second law leads locally to entropy as a state function whose level sets describe the adiabatic leaves of the thermodynamic state space \cite{boyling1972axiomatic}, \cite{frankel2004geometry}. A natural requirement of Carathéodory's approach to the second law is a continuous structure on the phase space \cite{boyling1972axiomatic, lieb2004guide}. This leads us to the central axiomatic observation that entropy can be used as a comparison measure on equilibrium states by transferring an accessibility relation on phase space to the real line \cite{lieb1999physics}.  

In the original paper \cite{kycia2018landauer}, the author adopts this perspective and uses it as the entry point for the categorification process of thermodynamic entropy. Given a state space \(\Gamma\) and an entropy function \(S:\Gamma\longrightarrow\mathbb{R}\), the entropy values determine a total preorder
\begin{equation}
X\preceq_S Y\iff S\left(X\right)\leq S\left(Y\right).
\end{equation}
Distinct states with equal entropy may satisfy both \(X\preceq_SY\) and \(Y\preceq_SX\), so antisymmetry need not hold on \(\Gamma\) itself. Passing to the equivalence relation
\begin{equation}
X\sim_S Y\iff X\preceq_S Y\land Y\preceq_S X,
\end{equation}
produces the poset reflection \(\Gamma/_{\sim S}\). When the Comparison hypothesis \cite{lieb1999physics, lieb2004guide} allows the physical accessibility relation to be completely represented by entropy, this quotient order is the corresponding entropy order.

Additionally, $S$ satisfies additivity and extensivity \cite{kycia2018landauer}. Thus, the entropy function induces a total preorder on the phase space \(\Gamma\), and after quotienting by entropy equivalence, the resulting poset inherits the corresponding total order represented in \(\mathbb{R}\).

This is where the language of category theory comes into play. The conceptual shift from numerical values (collection of states) to order-theoretic structures transforms entropy into morphisms between specific categories, and reversibility/irreversibility becomes a property of order-preserving maps acting on the phase space of thermodynamic states. 

In the original approach from \cite{kycia2018landauer}, we can see that category theory is not invoked to re-prove Landauer's principle via abstraction alone, rather, it isolates the universal properties of the principle and offers a more conceptual and general treatment of thermodynamic entropy. In this spirit, the paper \cite{kycia2018landauer} frames the classical logic-thermodynamics table underlying Landauer's principle as evidence of an adjoint relationship. Therefore, the author formulates the bridge theorem (Theorem 4 in \cite{kycia2018landauer}), interconnecting logical and thermodynamic (ir)reversibility properties, resulting in the Galois connection.

The key mathematical structure in \cite{kycia2018landauer} is the Landauer connection, realized as minimal entropy-preserving coupling. Let two entropy systems $G_1:\left(\Gamma_1,\preceq_1\right)$ and $G_2:=\left(\Gamma_2,\preceq_2\right)$ be related by functors (monotone maps) $F:\Gamma_1\longrightarrow\Gamma_2$ and $G:\Gamma_2\longrightarrow\Gamma_1$ ($F\dashv G$), satisfying the Galois condition
\begin{equation}\label{eq.2}
F\left(c\right)\preceq_2d\iff c\preceq_1 G\left(d\right),
\end{equation}
for all $c\in G_1$ and $d\in G_2$.

In \cite{kycia2018landauer}, this is interpreted as the most general relation that preserves the entropy-induced ordering. From a physical standpoint, one immediately observes that the composite $GF$ acts as a closure operator on $\Gamma_1$. Hence, it sends a state to a "higher or equal" state in the entropy order, encoding the fact that a round-trip through the implementing system cannot increase recoverable information on the original side of the system. 

Considering the canonical example (Figure 1 in \cite{kycia2018landauer}), irreversible logical operations induce irreversible physical processes accompanied by heat emission, with the quantitative lower bound depending on physical specifics, while the categorical content corresponds to the existence of the adjunction itself. 

The contribution of the present paper is not the construction of a new abstract profunctor bicategory, since the bicategory depiction of the Boolean profunctors between posets is standard. Rather, the contribution is the identification of this standard bicategorical structure as the natural compositional semantics of entropy-constrained implementation.

More precisely, we make three contributions. First, we enlarge the ambient morphism theory from deterministic monotone maps to Boolean profunctors, thereby allowing genuinely open and non-representable many-to-many feasibility interfaces. Second, we identify exact Landauer adjunctions with the representable/Cauchy sector of the Boolean profunctor bicategory. Categorically, the underlying representability phenomenon is a standard consequence of the Cauchy completeness of posets in the Boolean distributor setup. Nonetheless, our contribution is to isolate its consequences for Landauer implementation and to show how it constrains the ostensibly open thermodynamic interface theory.    

Third, starting from the boundary round-trip interface \(T:=\Psi\circ\Phi\), we identify the preorder generated by \(T\) and construct its poset reflection \(B_T:=B/{\sim_T}\). Then, we prove that this explicitly constructed quotient realizes the Eilenberg-Moore object of the idempotent Landauer monad \(T\). While the abstract theory of idempotent monads, their splittings, and Eilenberg-Moore objects is classical (see \cite{Street1972,stubbe2005categorical}), the contribution here is the identification of the resulting universal object with a concrete Landauer-induced quotient of the boundary information. This provides a thermodynamic interpretation for the universal construction as \(B_T\) records precisely those distinctions of the boundary description that remain visible after bulk-mediated implementation and readout. 

\section{Main results}

While the approach in \cite{kycia2018landauer} already captures a deep universality in the categorification of thermodynamic entropy as an adjunction between entropy-ordered phase spaces, it remains, by design, a thin categorification process. To be clear, a thin category (posetal category) is a category in which, given a pair of objects $a, b$ and any two morphisms $f,g:a\longrightarrow b$, the morphisms are equalities $f=g$. Within this thin-categorical setup, morphisms are merely order relations, and the implementing maps are deterministic monotone functions.

However, a logical macrostate may admit many physical realizations, since a physical state may support multiple compatible logical descriptions, and interaction with an environment is not an afterthought but part of the operational semantics. Hence, the notion of implementation is intrinsically \textit{open} and \textit{many-to-many}. This leads us to our first conceptual extension. 

\vspace{0.2cm}
\hspace{-0.6cm}\underline{\textbf{Conceptual Extension I. - open entropy interfaces}}

\vspace{0.2cm}
We shall directly replace deterministic monotone maps by \textit{open couplings/interfaces} encoded as profunctorial relations and organize them into a bicategory in which:
\begin{itemize}
    \item objects are entropy systems (posets/thin categories)
    \item $1$ - morphisms are open interfaces encoding the feasibility of realizations of physical systems
    \item $2$ - morphisms are refinements relating open interfaces by inclusion
\end{itemize}

Within this setup, the Landauer connections from \cite{kycia2018landauer} embed as a representable special case, while the bicategorical language makes it natural to interpret the logical$\hspace{0.1cm}\longleftrightarrow\hspace{0.1cm}$ physical correspondence as a boundary $\longleftrightarrow$ bulk correspondence. Here, the boundary data (information-theoretic description) is coupled to bulk data (thermodynamic realization) by an adjunction of interfaces, whose induced boundary monad generalizes the closure operator in \cite{kycia2018landauer}.

From a technical point of view, concepts such as realizations or open interfaces will be precisely defined in the following paragraphs.

\vspace{0.3cm}
\hspace{-0.6cm}\underline{\textbf{Conceptual Extension II. - Landauer adjunctions and boundary stable reconstructions}}

\vspace{0.2cm}
The second conceptual extension consists of the holographic perspective in categorical form. In higher-categorical implementation of the holographic principle, the guiding theme is that the bulk structure is reconstructible from the boundary structure via a universal construction. In our setting, the corresponding universal object is not a geometric bulk spacetime but the bulk sector visible through a given coupling. Concretely, a bulk-boundary Landauer adjunction (on the boundary) induces a canonical monad/closure operator. This means that the holographic reconstruction problem transforms into the following question:

\vspace{0.2cm}
\textit{To what extent does the boundary, equipped with this induced monad, determine the bulk content accessible through the interface?}

\vspace{0.2cm}
This can be formalized via the Eilenberg-Moore type constructions for the induced monad which, in the poset case, reduces to an explicit fixed-point (closed state) reconstruction that parallels the discussion about the closure operator in \cite{kycia2018landauer}. Technically speaking, a pair of such interfaces \(\Phi\rightsquigarrow D\) and \(\Psi\rightsquigarrow B\) is required to form an adjunction in the locally posetal bicategory of open entropy systems. This recovers the usual Galois/Landauer connection in the representable case. A further result established below shows that exact adjunctions are, in fact, rigid, i.e. every exact adjunction of Boolean profunctors between posets is representable. Therefore, genuinely non-representable interfaces occur in the ambient open theory, but not as the pillars of an exact Landauer adjunction.

\vspace{0.3cm}
\hspace{-0.6cm}\underline{\textbf{Conceptual Extension III. - quantitative enrichment}}

\vspace{0.2cm}
Finally, we formulate a cost-enriched version in which feasibility is refined by dissipation cost, and composition selects an optimal intermediate bulk mediator. This indicates how Boolean feasibility theory can potentially be extended toward a quantitative theory of entropy production.

\vspace{0.2cm}
Let us now proceed to the following subsection, where we shall build the bicategorical picture of open entropy systems.

\subsection{\textbf{The bicategory of open entropy systems}}

Accordingly, the upcoming Theorems~\ref{thm.1} and~\ref{thm.2} below should be read as establishing the ambient semantics rather than as the principal novelty of the paper. Their role is to fix the variance conventions and to justify the physical interpretation of profunctorial composition as gluing thermodynamic implementations while hiding intermediate degrees of freedom. The new Landauer-specific content begins when adjunctions in this bicategory are interpreted as implementation/abstraction pairs and when the induced boundary monad is shown to select a universal stable sector. In essence, the categorical novelty of the present paper should not be understood as the invention of a new profunctor bicategory. Rather, the contribution is the systematic specialization of this established machinery to entropy-ordered state spaces, together with its interpretation as a language for open, many-to-many thermodynamic implementations of logical boundary data by physical bulk states.

First of all, let us recall useful facts from \cite{kycia2018landauer} in order to build the correct notion of an entropy system. In \cite{kycia2018landauer}, the author works with a state space $\Gamma$ (equilibrium of macrostates) optionally equipped with a scaling action by $\left(\mathbb{R}_{+},\cdot,1\right)$. Thereafter, the author uses entropy to transfer an accessibility/ordering relation on $\Gamma$ to $\mathbb{R}$. This leads to the definition of the entropy function.

\newpage
\newtheorem{definition}{Definition}
\begin{definition}\cite{kycia2018landauer}
\textit{Let $\Gamma$ be a state space equipped with:}
\begin{itemize}
    \item a composition operation $\left(X,Y\right)\in\Gamma_1\times\Gamma_2$ (compound system),
    \item a scaling action $\lambda X$ for $\lambda\in\mathbb{R}_+$
\end{itemize}
\textit{Then an entropy function is a map}
\begin{equation}\label{eq.3}
S:\Gamma\longrightarrow\mathbb{R},
\end{equation}
\textit{satisfying:}
\begin{itemize}
    \item\textbf{Monotonicity w.r.t. adiabatic accessibility:}
    \begin{equation}\label{eq.4}
    X\prec Y\iff S\left(X\right)\leq S\left(Y\right).
    \end{equation}
    \item\textbf{Additivity for compound systems:}
    \begin{equation}\label{eq.5}
    S\left(X,Y\right)=S\left(X\right)+S\left(Y\right).
    \end{equation}
    \item\textbf{Extensivity under scaling:}
    \begin{equation}\label{eq.6}
    S\left(\lambda X\right)=\lambda S\left(X\right),
    \end{equation}
    \textit{for $\lambda>0$.}
\end{itemize}
\end{definition}

This conception is then explicitly used in a reverse direction as a construction principle, i.e. assume that such $S$ exists on the relevant domain (Comparison hypothesis), and define the order via $S$. 

If one wants to generalize the previous consideration to the bicategorical setup, the underlying object must at least be a poset or a preorder, because the $1$-cells are order-compatible interfaces. In other words, we shall consider the following

\begin{definition}\cite{kycia2018landauer}
\textit{An entropy object (entropy system) is a pair $\left(\Gamma, S\right)$, where $S:\Gamma\longrightarrow\mathbb{R}$ is a function such that the relation}
\begin{equation}\label{eq.7}
X\preceq_S Y\iff S\left(X\right)\leq S\left(Y\right),
\end{equation}
is the intended entropy accessibility order.
\end{definition}

One can immediately notice that if the Comparison hypothesis holds (assumption for the main construction in \cite{kycia2018landauer}), then $\preceq_S$ is a total preorder on $\Gamma$. Without the Comparison hypothesis, one may restrict to a comparable domain or treat $\preceq_S$ as a preorder on the specific comparable subset. 

\begin{definition}[\textbf{Entropy system (objects)}]
\textit{An entropy-ordered system is a poset \(P=\left(P,\leq_P\right)\), interpreted as an accessibility order, optionally equipped with an entropy function}
\begin{equation}
S:P\longrightarrow\mathbb{R},
\end{equation}
\textit{that is monotone as}
\begin{equation}
p\leq_P p^{\prime}\implies S\left(p\right)\leq S\left(p^{\prime}\right).
\end{equation}
\end{definition}
When the Comparison hypothesis holds in the domain under consideration (see, \cite{lieb1999physics, lieb2004guide}), the order may be represented completely by entropy
\begin{equation}
\nonumber
p\leq_P p^{\prime}\iff S\left(p\right)\leq S\left(p^{\prime}\right),
\end{equation}
up to the usual identification of mutually accessible states. Nonetheless, in the general bicategorical framework below, we retain the underlying partial order and do not require every pair of states to be comparable. Throughout, all posets are assumed to be small or alternatively, contained in a fixed Grothendieck universe.  

In relation to previous considerations, we work with entropy systems only up to categorical equivalence by passing from  entropy preorders $\left(\Gamma, \preceq_S\right)$ to their poset reflection $\Gamma/{\sim_S}$, where $x\sim y\iff$ $x\preceq_S y$ and $y\preceq_S x$. Consequently, one can think of $P$ as a thin category. This is a classical result from category theory, see \cite{roman2017introduction}. This categorical treatment of entropy systems is close in spirit to Lawvere's formulation of state categories and semicontinuous entropy functions, where thermodynamic states and entropy-compatible processes are organized by categorical structure \cite{lawvere1984state}.

Now we proceed to the formalization of the concept of entropy interfaces (open couplings).

\begin{definition}[\textbf{Open couplings/interfaces ($1$-morphisms)}]
\textit{Let $P,Q$ be posets, then we define an entropy interface $\Phi:P\nrightarrow Q$ to be a Boolean profunctor}
\begin{equation}\label{eq.8}
\Phi:P^{op}\times Q\longrightarrow\mathbf{Bool},
\end{equation}
\textit{where $\mathbf{Bool}=\{\perp,\top\}$ represents a posetal category of Boolean-valued relations with \(\bot\leq\top\).}
\end{definition}

Additionally, we can say that because $P,Q$ are posets, giving $\Phi$ is equivalently providing a relation $\Phi\subseteq P\times Q$ that is monotone in both arguments:
\begin{itemize}
    \item\textbf{Contravariance in the source variable} $P$: For all $p,p^{\prime}\in P$, if $p\leq p^{\prime}$ and $\Phi\left(p^{\prime},q\right)=\top$, then $\Phi\left(p,q\right)=\top$
    \item\textbf{Covariance in the target variable} $Q$: For all $q,q^{\prime}\in Q$, if $q\leq q^{\prime}$ and $\Phi\left(p,q\right)=\top$, then $\Phi\left(p,q^{\prime}\right)=\top$ 
\end{itemize}

Equivalently, $\Phi$ is monotone as a map $P^{op}\times Q\longrightarrow\mathbf{Bool}$. Additionally, one can view $\mathbf{Bool}$ as a thin category with morphisms $\perp\longrightarrow\perp$, $\perp\longrightarrow\top$, $\top\longrightarrow\top$ but no morphism $\top\longrightarrow\perp$. Whenever we write an inequality in $\mathbf{Bool}$, one can read it as implication. From now on, we will denote the Boolean thin category as \(\mathbf{2}:=\mathbf{Bool}\).

Later, when we switch to the holographic bulk-to-boundary setup, meaning when \(P\) is specialized to a boundary information system and \(Q\) to a bulk thermodynamic system, then \(\Phi\left(p,q\right)\) will be read as the feasibility of realizing the boundary description \(p\) via the bulk state \(q\). Moreover, the two monotonicity conditions of $\Phi$ will have a direct thermodynamic meaning.  Considering the boundary monotonicity, we can say that if a bulk  state can realize a detailed boundary description, it can realize any coarser boundary description obtained by forgetting information about the given system. When it comes to the covariance in the bulk variable, it is an explicit modeling assumption, i.e. the chosen bulk order is required to be compatible with feasibility in the sense that moving upward in the bulk order preserves admissibility. The physical interpretation of this condition depends on the specific meaning assigned to the bulk order. Nevertheless, one point is clear, it should not be read as the unconditional claim that arbitrary thermal noise preserves implementability. 

The profunctor $\Phi$ (open couplings/interfaces) admits an alternative description, and to do this, we must define the following notions.

\begin{definition}[\textbf{Upward closure}]
\textit{Let $P$ be a poset, then a subset $U\subseteq P$ is upward closed if}
\begin{equation}\label{eq.9}
\forall x,y\in P \hspace{0.1cm}\left(x\in U\hspace{0.1cm}\land
\hspace{0.1cm} x\leq y\implies y\in U\right).
\end{equation}
\end{definition}

\begin{definition}
The set of all upward closed subsets of $P$ is defined as
\begin{equation}\label{eq.10}
\operatorname{Up}\left(P\right):=\{U\subseteq P\hspace{0.1cm}\vert\hspace{0.1cm}U\hspace{0.1cm }is\hspace{0.1cm}upward\hspace{0.1cm}closed\}.
\end{equation}
\end{definition}

Accordingly, we can partially order the set $\operatorname{Up}\left(P\right)$ by inclusion
\begin{equation}\label{eq.11}
U\leq_{\operatorname{Up}}V\iff U\subseteq V.
\end{equation}
Then $\left(\operatorname{Up}\left(P\right),\subseteq\right)$ is a poset. Similarly, one can define the set of all downward closed subsets of $P$ as $\operatorname{Down}\left(P\right):=\{L\subset P\hspace{0.1cm}\vert\hspace{0.1cm}L\hspace{0.1cm}\textit{is downward closed}\}$.

This leads us to the following order-theoretic presentation of the Boolean profunctor $\Phi$ (open couplings/interfaces).

The following result is the posetal specialization of the standard equivalence between Boolean-valued profunctors, order-compatible relations, and monotone maps into posets of upper or lower sets. We record it explicitly because, in the thermodynamic interpretation, the two variance directions have concrete meanings, i.e. the contravariance in the boundary variable expresses stability under coarse-graining of logical descriptions, while covariance in the bulk variable expresses stability under increasing entropy or decreasing thermodynamic control. 

\begin{theorem}\label{thm.1}
\textit{For posetal entropy systems $P,Q$, the following statements are equivalent:}
\begin{enumerate}
    \item An open coupling/interface from $P$ to $Q$ is a Boolean profunctor
    $\Phi: P^{op}\times Q\longrightarrow\mathbf{2}$
    \item A relation $\Phi\subseteq P\times Q$ (write $\Phi\left(p,q\right)$ for $\left(p,q\right)\in\Phi$) is downward closed in the variable $P$ and upward closed in the variable $Q$ 
    \item A monotone map $f_{\Phi}:P^{op}\longrightarrow\operatorname{Up}\left(Q\right)$ such that $f_{\Phi}\left(p\right):=\{q\in Q\hspace{0.1cm}\vert\hspace{0.1cm}\Phi\left(p,q\right)=\top\}$
    \item A monotone map $g_{\Phi}:Q\longrightarrow\operatorname{Down}\left(P\right)$ such that $g_{\Phi}\left(q\right):=\{p\in P\hspace{0.1cm}\vert\hspace{0.1cm}\Phi\left(p,q\right)=\top\}$
\end{enumerate}
\end{theorem}

\vspace{0.2cm}
Here, \(P^{op}\) encodes contravariance in the source variable. Consequently, for fixed \(q\), the fiber \(\{p\in P\hspace{0.1cm}\vert\hspace{0.1cm}\Phi\left(p,q\right)=\top\} \) is downward closed in \(P\), while for fixed \(p\), the fiber \(\{q\in Q\hspace{0.1cm}\vert\hspace{0.1cm}\Phi\left(p,q\right)=\top\}\) is upward closed in \(Q\).

From the categorical point of view, Theorem~\ref{thm.1} is a standard fact about \(\mathbf{2}\)-valued profunctors between posets, see \cite{fong2019invitation}. The equivalence is used here to move between three views of the same implementation interface. The profunctorial view is convenient for composition, the relational view expresses physical feasibility, and the fiber view assigns to each boundary description its upward-closed family of admissible bulk realizations. Contravariance in the boundary variable implements forgetting, i.e. a realization of a sharper description also realizes every coarser description. Covariance in the bulk variable expresses stability under entropy-compatible degradation. 

Moreover, everything in Theorem~\ref{thm.1} works identically for preorders. Within this setup, $\operatorname{Up}\left(P\right)$ is ordered by inclusion and remains a poset, even if $P$ is only a preorder.

Next, the use of $P^{op}$ encodes contravariance in the source variable. Thus, increasing the source variable reverses feasibility, while increasing the target variable preserves feasibility. More explicitly, if \(p\leq_Pp^{\prime}\), then \(\Phi\left(p^{\prime},q\right)\leq\Phi\left(p,q\right)\), and if \(q\leq_Qq^{\prime}\), then \(\Phi\left(p,q\right)\leq\Phi\left(p,q^{\prime}\right)\).

As in every bicategory, it is necessary to define the operation of composition. 

\begin{definition}[\textbf{Composition (gluing of open couplings/interfaces)}]\label{def.7}
\textit{Given two open couplings (morphisms) $\Phi: P\nrightarrow Q$ and $\Psi: Q\nrightarrow R$, the composition $\Psi\circ\Phi: P\nrightarrow R$ is defined as the following relation}
\begin{equation}\label{eq.26}
\left(\Psi\circ\Phi\right)\left(p,r\right)\iff\left(\exists q\in Q\right)\left(\Phi\left(p,q\right)=\top\land\Psi\left(q,r\right)=\top\right).
\end{equation}
\end{definition}

Categorically, one can usually define~(\ref{eq.26}) via the relation/profunctor formula
\begin{equation}\label{eq.27}
\left(\Psi\circ\Phi\right)\left(p,r\right):=\bigvee_{q\in Q}\left(\Phi\left(p,q\right)\land\Psi\left(q,r\right)\right),
\end{equation}
and since we are in $\operatorname{Bool}$, this boils down to~(\ref{eq.26}).

In other words, the composition connects the output of $\Phi$ to the input of $\Psi$ by existentially hiding the intermediate interface variable $q$.

From a physical perspective, we connect an implementation $\Phi$ (boundary-bulk feasibility) to another stage $\Psi$ (bulk - deeper bulk feasibility) by hiding the intermediate degrees of freedom. 

\begin{definition}[\textbf{Identity interfaces}]\label{def.8}
\textit{For a poset $P$, define the identity interface $\operatorname{Id}_P:P\nrightarrow P$ as the following hom-relation}
\begin{equation}\label{eq.28}
\operatorname{Id}_P\left(p,p^{\prime}\right)=
\begin{cases}
\top;\hspace{0.1cm}if\hspace{0.1cm}p\leq p^{\prime} \\
\perp;\hspace{0.1cm}otherwise.
\end{cases}
\end{equation}
\end{definition}

Equivalently, $\operatorname{Id}_P$ is the hom-relation in the thin category $P$ regarded as a monotone map $P^{op}\times P\longrightarrow\operatorname{Bool}$, given by the order of $P$. Now, we proceed to the definition of the $2$-morphism, or the so-called refinement. 

\begin{definition}[\textbf{Refinement ($2$-morphism)}]
\textit{Given two interfaces $\Phi,\Phi^{\prime}: P\nrightarrow Q$, we define the refinement ($2$-morphism)}
\begin{equation}\label{eq.29}
\alpha:\Phi\implies\Phi^{\prime},
\end{equation}
\textit{as a pointwise implication}
\begin{equation}\label{eq.30}
\Phi\subseteq\Phi^{\prime}\iff\left(\forall p\in P,\forall q\in Q \right)\left(\Phi\left(p,q\right)=\top\implies\Phi^{\prime}\left(p,q\right)=\top\right).
\end{equation}
\end{definition}

One can translate the inclusion $\Phi\subseteq\Phi^{\prime}$ as that $\Phi^{\prime}$ allows at least the couplings that are permitted by $\Phi$. In other words, $\Phi^{\prime}$ is a less constrained, more permissive interface.

The pointwise implication can also be viewed as an inclusion of relations $\Phi\left(p,q\right)\leq\Phi^{\prime}\left(p,q\right)$. In a categorical language, the refinement $\alpha$ can be considered a $2$-cell. 

Let us now define the composition of refinements ($2$-cells).
\begin{definition}
\textit{Given two refinements $\alpha:\Phi\implies\Phi^{\prime}$ and $\alpha^{\prime}:\Phi^{\prime}\implies\Phi''$, we can define two types of compositions:}
\begin{itemize}
    \item\textbf{Vertical composition:} $\Phi\implies\Phi''$ by the transitivity of implication
    \item\textbf{Horizontal composition:} if $\alpha:\Phi\implies\Phi^{\prime}$and $\beta:\Psi\implies\Psi^{\prime}$, then we define
    \begin{equation}\label{eq.31}
    \beta\star\alpha:\Psi\circ\Phi\implies\Psi^{\prime}\circ\Phi^{\prime},
    \end{equation}
    \textit{as the pointwise implication induced by the monotonicity of composition.}
\end{itemize}
\end{definition}
At this point, we have everything ready to summarize it into a bicategorical framework. 

Let $\mathbf{\operatorname{OEnt}}$ denote a bicategory of open entropy systems with:
\begin{itemize}
    \item objects as open entropy systems (posets)
    \item $1$-cells ($1$-morphisms) as open couplings/interfaces (Boolean profunctors) $P\nrightarrow Q$
    \item $2$-cells ($2$-morphisms) as refinements $\Phi\implies\Phi^{\prime}$
    \item composition defined as in Definition~\ref{def.7}.
    \item identities $\operatorname{Id}_P\left(p,p^{\prime}\right):=[p\leq p^{\prime}]$ as in Definition~\ref{def.8}
\end{itemize}

We now assemble these open couplings/interfaces into a bicategory. Categorically, this is the usual locally posetal bicategory of Boolean profunctors, restricted to those posets that arise as entropy-ordered state spaces. We nevertheless spell out the construction because its composition law has a direct physical meaning - composing interfaces corresponds to gluing open implementations and existentially hiding the intermediate thermodynamic degrees of freedom.

\begin{theorem}\label{thm.2}(\textbf{Bicategory of open entropy systems})
$\operatorname{OEnt}$ is a locally posetal bicategory. 
\end{theorem}
\begin{proof}[Proof of Theorem~\ref{thm.2}]

At the beginning of this proof, it is necessary to show that the identity profunctor and composites are well-defined. 

We must show that $\operatorname{Id}_P:P^{op}\times P\longrightarrow\mathbf{2}$ is monotone. Take $\left(p,p^{\prime}\right)\leq\left(q,q^{\prime}\right)$ in $P^{op}\times P$. By definition of product order, we get $p\leq_{P^{op}}q$, i.e. $q\leq_{P}p$ and $p^{\prime}\leq_{P}q^{\prime}$. Assume that $\operatorname{Id}_P\left(p,p^{\prime}\right)=\top$. Then $p\leq_P p^{\prime}$. Additionally, since $q\leq p\leq p^{\prime}\leq q^{\prime}$, we have $q\leq q^{\prime}$. Hence, $\operatorname{Id}_P\left(q,q^{\prime}\right)=\top$. Thus, $\operatorname{Id}_P$ is monotone and is well-defined as a $1$-cell. 

Now, let us proceed to the composite of $1$-cells. We must verify the contravariance in $P$ and the covariance in $R$ for $\Psi\circ\Phi$. 

Let \(p,p^{\prime}\in P\) with \(p\leq_P p^{\prime}\) and assume that \(\left(\Psi\circ \Phi\right)\left(p^{\prime},r\right)=\top\). Then there exists \(q\in Q\) such that \(\Phi\left(p^{\prime},q\right)=\top\) and \(\Psi\left(q,r\right)=\top\). Since \(\Phi\) is contravariant in \(P\), from
\(p\leq_P p^{\prime},\quad \Phi\left(p^{\prime},q\right)=\top\), we obtain \(\Phi\left(p,q\right)=\top\). Therefore, the same \(q\) witnesses \(\left(\Psi\circ\Phi\right)\left(p,r\right)=\top\), so \(p\leq_P p^{\prime}\implies \left(\Psi\circ\Phi\right)\left(p^{\prime},r\right)\leq\left(\Psi\circ\Phi\right)\left(p,r\right)\). Therefore, the contravariance in $P$ holds.

Regarding the covariance in $R$, let $r\leq r^{\prime}$ in $R$. Presume that $\left(\Psi\circ\Phi\right)\left(p,r\right)=\top$. Choose $q\in Q$ with $\Phi\left(p,q\right)=\top$ and $\Psi\left(q,r\right)=\top$. Since $\Psi$ is covariant in $R$ and $r\leq r^{\prime}$, then from $\Psi\left(q,r\right)=\top$, we obtain $\Psi\left(q,r^{\prime}\right)=\top$. 

Similarly, we can show that $\left(\Psi\circ\Phi\right)\left(p,r^{\prime}\right)=\top$. Hence, the covariance also holds in this case, proving that the $1$-cell $\Psi\circ\Phi$ is well-defined. 

Since our composition is well-defined, we can now show that $\operatorname{OEnt}$ is a strict $2$-category, hence a bicategory with identity associator/unitors. 

First of all, let us focus on the associativity of the composition. For $1$-cells $\Phi: P\nrightarrow Q$, $\Psi:Q\nrightarrow R$ and $\Sigma: R\nrightarrow S$, we must show that 
\begin{equation}\label{eq.32}
\Sigma\circ\left(\Psi\circ\Phi\right)=\left(\Sigma\circ\Psi\right)\circ\Phi,
\end{equation}
as a $1$-cell $P\nrightarrow S$.

Fix $p\in P$ and $s\in S$. We show the equivalence of truth values, so
\begin{equation}\label{eq.33}
\begin{aligned}
\left(\Sigma\circ\left(\Psi\circ\Phi\right)\right)\left(p,s\right)=\top\iff\left(\exists r\in R\right)\left(\left(\Psi\circ\Phi\right)\left(p,r\right)=\top\land\Sigma\left(r,s\right)=\top\right)\iff\\
\iff\left(\exists r\in R\right)\left(\exists q\in Q:\Phi\left(p,q\right)=\top\land \Psi\left(q,r\right)=\top\right)\land\Sigma\left(r,s\right)=\top.
\end{aligned}
\end{equation}
where the first equivalence stems from the definition of composition, and the second equivalence is given by the expansion of $\left(\Psi\circ\Phi\right)\left(p,r\right)$. 

Now we can move the existential quantifiers outward in order to expand~(\ref{eq.33}) as
\begin{equation}\label{eq.34}
\left(\exists q\in Q\right)\left(\exists r\in R\right)\left(\Phi\left(p,q\right)=\top\land\Psi\left(q,r\right)=\top\land\Sigma\left(r,s\right)=\top\right).
\end{equation}
Regrouping using the definition of $\left(\Sigma\circ\Psi\right)\left(q,s\right)$ yields
\begin{equation}\label{eq.35}
\begin{aligned}
\left(\exists q\in Q\right)\left(\Phi\left(p,q\right)=\top\land\left(\exists r\in R:\hspace{0.1cm}\Psi\left(q,r\right)=\top\land\Sigma\left(r,s\right)=\top\right)\right)\iff \\
\iff\exists q\in Q: \Phi\left(p,q\right)=\top\land\left(\Sigma\circ\Psi\right)\left(q,s\right)=\top.
\end{aligned}
\end{equation}
Hence by definition of composition, we obtain
\begin{equation}\label{eq.36}
\left(\left(\Sigma\circ\Psi\right)\circ\Phi\right)\left(p,s\right)=\top.
\end{equation}

One can immediately see that the two predicates coincide pointwise, implying that the $1$-cells are identical.

Now, let us focus on the left and right unit laws. Specifically, we have to show that for any $1$-cell $\Phi: P\nrightarrow Q$,
\begin{equation}\label{eq.37}
\operatorname{Id}_Q\circ\Phi=\Phi,\hspace{0.2cm}and \hspace{0.2cm}\Phi\circ\operatorname{Id}_P=\Phi.
\end{equation}

Fix $p\in P$ and $q^{\prime}\in Q$. Then
\begin{equation}\label{eq.38}
\begin{aligned}
\left(\operatorname{Id}_    Q\circ\Phi\right)\left(p,q\right)=\top\iff\exists q^{\prime}\in Q: \Phi\left(p,q^{\prime}\right)=\top\land q^{\prime}\leq_Q q.
\end{aligned}
\end{equation}
Now, suppose that \(\left(\operatorname{Id_Q\circ \Phi}\left(p,q\right)\right)=\top\). Then there exists \(q^{\prime}\in Q\) such that \(\Phi\left(p,q^{\prime}\right)=\top\) and \(q^{\prime}\leq_Q q\). Since \(\Phi\) is covariant in its \(Q\)-variable, it follows that \(\Phi\left(p,q\right)=\top\). Conversely, if \(\Phi\left(p,q\right)=\top\), choose \(q^{\prime}=q\). Since \(q\leq_Q q\), we have \(\operatorname{Id}_Q\left(q,q\right)=\top\), and therefore the existential condition defining \(\left(\operatorname{Id}_Q\circ\Phi\right)\left(p,q\right)\) is satisfied. Hence \(\left
(\operatorname{Id}_Q\circ\Phi\right)\left(p,q\right)=\Phi\left(p,q\right)\) for all \(\left(p,q\right)\). This proves the left unit law $\operatorname{Id}_Q\circ\Phi=\Phi$.

Proceeding with the right unit law, we should fix $p\in P$ and $q\in Q$. Then
\begin{equation}\label{eq.39}
\begin{aligned}
\left(\Phi\circ\operatorname{Id}_P\right)\left(p,q\right)=\top\iff\exists p^{\prime}\in P: \operatorname{Id}_P\left(p,p^{\prime}\right)=\top\land\Phi\left(p^{\prime},q\right)=\top\iff \\
\iff\exists p^{\prime}\in P:\left(p\leq p^{\prime}\right)\land\Phi\left(p^{\prime},q\right)=\top.
\end{aligned}
\end{equation}

If the implication $\implies$ holds, then pick such $p$. Since $p\leq p^{\prime}$ and $\Phi$ is contravariant in $P$, then we obtain $\Phi\left(p,q\right)=\top$.

Conversely, if the implication $\impliedby$ holds, then we have $\Phi\left(p,q\right)=\top$, so we can directly choose $p^{\prime}=p$. Then $p\leq p^{\prime}$, which gives us $\operatorname{Id}_P\left(p,p\right)=\top$. So the existential condition holds. Therefore, $\left(\Phi\circ\operatorname{Id}_P\right)\left(p,q\right)=\Phi\left(p,q\right)$. This proves the right unit law $\Phi\circ\operatorname{Id}_P=\Phi$.  

So far, we can conclude that $\operatorname{OEnt}$ is a bicategory, which is indeed a strict $2$-category at the level of $1$-cells. Hence, it remains to show that $\operatorname{OEnt}$ has a local posetal structure by proving that each hom is a poset.

At first, let us fix posets $P,Q$. Consider the class $\operatorname{OEnt}\left(P,Q\right)$ of $1$-cells $P\nrightarrow Q$ and define the order on $\operatorname{OEnt}\left(P,Q\right)$ by
\begin{equation}\label{eq.40}
\Phi\leq\Phi^{\prime}\iff\forall\left(p,q\right):\Phi\left(p,q\right)\leq\Phi^{\prime}\left(p,q\right),
\end{equation}
which holds pointwise in $\operatorname{Bool}$.

One can immediately see the reflexivity of~(\ref{eq.40}), because for each pair $\left(p,q\right)$, the inequality $\Phi\left(p,q\right)\leq\Phi\left(p,q\right)$ holds in $\operatorname{Bool}$.

Similarly, we can simply show the transitivity. For each pair $\left(p,q\right)$, we have
$\Phi\left(p,q\right)\leq\Phi^{\prime}\left(p,q\right)\leq\Phi''\left(p,q\right)$. Hence if $\Phi\leq\Phi^{\prime}$ and $\Phi^{\prime}\leq\Phi''$, then $\Phi\leq\Phi''$.

At last, it remains to prove the antisymmetry. Hence for each pair $\left(p,q\right)$, we have $\Phi\left(p,q\right)\leq\Phi^{\prime}\left(p,q\right)$ and $\Phi^{\prime}\left(p,q\right)\leq\Phi\left(p,q\right)$. Antisymmetry in $\operatorname{Bool}$ gives us $\Phi\left(p,q\right)=\Phi^{\prime}\left(p,q\right)$, which boils down to equality as predicates. 

Therefore, we have shown that $\operatorname{OEnt}\left(P,Q\right)$ is a poset, which automatically implies that between any two $1$-cells, there is at most one $2$-cell.

Furthermore, it is evident that the vertical composition in $\operatorname{OEnt}$ is exactly the transitivity. So given $2$-cells $\alpha:\Phi\implies\Phi^{\prime}$ and $\alpha^{\prime}:\Phi^{\prime}\implies\Phi''$, their vertical composite $\alpha^{\prime}\circ\alpha:\Phi\implies\Phi''$ exists and is given by pointwise implications $\Phi\leq\Phi''$, because $\Phi\leq\Phi^{\prime}\leq\Phi{''}$ and $2$-cells are just pointwise inequalities $\Phi\leq\Phi^{\prime}$.

Considering two pairs of profunctors $\Phi,\Phi^{\prime}:P\nrightarrow Q$ and $\Psi,\Psi^{\prime}:Q\nrightarrow R$, we must also show that the horizontal composition is monotone in each argument. In order to do so, we shall fix $p\in P$ and $r\in R$ and first assume that $\left(\Psi\circ\Phi\right)\left(p,r\right)=\top$. Then there exists $q\in Q$ with $\Phi\left(p,q\right)=\top$ and $\Psi\left(q,r\right)=\top$. From $\Phi\leq\Phi^{\prime}$, we have $\Phi\left(p,q\right)\implies\Phi^{\prime}\left(p,q\right)=\top$ and from $\Psi\leq\Psi^{\prime}$, we obtain $\Psi\left(q,r\right)=\top\implies\Psi^{\prime}\left(q,r\right)=\top$.

Similarly, we can prove that $\left(\Psi^{\prime}\circ\Phi^{\prime}\right)\left(p,r\right)=\top$. Thus $\left(\Psi\circ\Phi\right)\left(p,r\right)\leq\left(\Psi^{\prime}\circ\Phi^{\prime}\right)\left(p,r\right)$ in $\operatorname{Bool}$.  This proves that horizontal composition preserves implication.

Finally, it is necessary to show the validity of the interchange law. In a strict 2-category, the interchange law says
\begin{equation}\label{eq.41}
\left(\beta^{\prime}\circ\beta\right)\star\left(\alpha^{\prime}\circ\alpha\right)=\left(\beta^{\prime}\star\alpha^{\prime}\right)\circ\left(\beta\star\alpha\right),
\end{equation}
which holds whenever the composites make sense. 

Now, because each hom is a poset, there is at most one $2$-cell between any two fixed $1$-cells. Therefore, to show the equality of two $2$-cells, it is adequate to show that they have the same source and target, and that they both exist. 

Take $\alpha:\Phi\implies\Phi^{\prime}$, $\alpha^{\prime}:\Phi^{\prime}\implies\Phi''$ to be $2$-cells in $\operatorname{OEnt}\left(P,Q\right)$ and $\beta:\Psi\implies\Psi^{\prime}$, $\beta^{\prime}:\Psi^{\prime}\implies\Psi''$ to be $2$-cells in $\operatorname{OEnt}\left(Q,R\right)$. Then $\alpha^{\prime}\circ\alpha$ is the unique $2$-cell $\Phi\implies\Phi''$, i.e. $\Phi\leq\Phi''$ and $\beta^{\prime}\circ\beta$ is the unique $2$-cell $\Psi\implies\Psi''$, i.e. $\Psi\leq\Psi''$. 

Therefore, $\left(\beta^{\prime}\circ\beta\right)\star\left(\alpha^{\prime}\circ\alpha\right)$ is the unique $2$-cell
\begin{equation}\label{eq.42}
\Psi\circ\Phi\implies\Psi''\circ\Phi'',
\end{equation}
and it exists. 

On the other hand, $\beta\star\alpha$ is the unique $2$-cell $\Psi\circ\Phi\implies\Psi^{\prime}\circ\Phi^{\prime}$ and $\beta^{\prime}\star\alpha^{\prime}$ is the unique $2$-cell $\Psi^{\prime}\circ\Phi^{\prime}\implies\Psi''\circ\Phi''$. Hence, the vertical composite $\left(\beta^{\prime}\star\alpha^{\prime}\right)\circ\left(\beta\star\alpha\right)$ is also the unique $2$-cell
\begin{equation}\label{eq.43}
\Psi\circ\Phi\implies\Psi''\circ\Phi''.
\end{equation}
Therefore, both sides of~(\ref{eq.41}) are unique $2$-cells with identical source and target (they are equal).

This proves the locally posetal structure of $\operatorname{OEnt}$, which completes the proof of Theorem~\ref{thm.2}.
\end{proof}

Therefore, the bicategory \(\mathrm{OEnt}\) is precisely the Boolean profunctor bicategory on entropy posets, equipped with the thermodynamic semantics introduced above. The standard categorical structure becomes physically meaningful here because identity profunctors encode entropy accessibility, profunctor composition encodes bulk-mediated implementation through hidden intermediate states, and \(2\)-cells encode relaxation or refinement of feasible couplings. Also notice that the composition \(\left(\Psi\circ\Phi\right)\left(p,r\right)\) existentially hides the intermediate state \(q\). Physically, this is the compositional law for open implementations. Two implementation stages may be glued while the internal thermodynamic mediator is treated as unobserved. Associativity says that the effective feasibility relation of a multistage process is independent of how the stages are bracketed. Moreover, identity interfaces encode the native accessibility relation of the state space. 

\subsection{\textbf{Landauer's connection}}

Having established the ambient profunctorial bicategory, we now reinterpret Landauer's connection within it. The categorical notion of adjunction used below is the ordinary notion of adjunction in a locally posetal bicategory. What is specific to the present setting is its thermodynamic reading. The unit expresses the feasibility of implementing boundary information in the bulk and returning to the boundary, while the counit expresses the impossibility of obtaining a sharper bulk realization by abstraction followed by re-implementation. 

In \cite{kycia2018landauer}, the author starts from an entropy function $S:\Gamma\longrightarrow\mathbb{R}$ and defines an entropy order $x\preceq y\iff S\left(x\right)\leq S\left(y\right)$, thereby viewing $\left(\Gamma,\preceq\right)$ as a thin posetal category. Within that setting, a Landauer connection between two entropy systems corresponds to the Galois connection $F\dashv G$ between the corresponding posets, i.e. an adjunction between thin categories. Its core irreversibility content consists of the induced closure operator $GF$. 

In our case, the locally posetal bicategory $\operatorname{OEnt}$ retains precisely this order-theoretic background but replaces deterministic implementation maps by open interfaces, i.e. $1$-cells $P\nrightarrow Q$ as Boolean profunctors, encoding which bulk states can realize which boundary states, and $2$-cells are refinements given by pointwise implication. Therefore, within $\operatorname{OEnt}$, a Landauer adjunction should be simply an adjunction $\Phi\dashv\Psi$ between open interfaces, i.e. the pair of inequalities $\operatorname{Id}_P\leq\Psi\circ\Phi$ and $\Phi\circ\Psi\leq\operatorname{Id}_D$ in the hom-posets. This generalizes the Landauer connection from \cite{kycia2018landauer} in two decisive ways. First, it admits many-to-many, environment-mediated realizations, since a single boundary state may have a whole upward-closed family of feasible bulk realizations (and conversely), and second, it distinguishes the bulk-boundary coupling itself as a morphism, rather than collapsing it to a single monotone map. Furthermore, the original adjunction $F\dashv G$ from \cite{kycia2018landauer} is then recovered as the representable special case where the open interface $\Phi$ is generated by a monotone function F via $\Phi\left(p,q\right)\iff F\left(p\right)\leq q$. Generally speaking, the bicategory $\operatorname{OEnt}$ allows implementations that cannot be summarized by a single canonical bulk state per boundary state. Therefore, the induced boundary endo-interface $T:=\Psi\Phi$ is then the bicategorical analogue of the closure operator $GF$ from \cite{kycia2018landauer}, capturing the Landauer principle as a structural constraint rather than a property of deterministic maps alone. In summary, the bicategorical framework of \(\operatorname{OEnt}\) enlarges the ambient space of possible implementations, while the exact Landauer adjunctions form a distinguished representable subclass within it.  

Considering a general bicategory $\mathscr{B}$, let $f:A\longrightarrow B$ and $g:B\longrightarrow A$ be $1$-cells. An adjunction $f\dashv g$ consists of:
\begin{itemize}
    \item a unit $2$-cell
    \begin{equation}\label{eq.44}
    \eta:\operatorname{Id}_A\implies g\circ f,
    \end{equation}
    \item a counit $2$-cell
    \begin{equation}\label{eq.45}
    \varepsilon: f\circ g\implies\operatorname{Id}_B,
    \end{equation}
\end{itemize}
such that the triangle identities hold
\begin{equation}\label{eq.46}
\left(\varepsilon\circ f\right)\cdot\left(f\circ\eta\right)=\operatorname{Id}_f,\hspace{0.2cm} \left(g\circ\varepsilon\right)\cdot\left(\eta\circ g\right)=\operatorname{Id}_g.
\end{equation}
Here, the symbol $\cdot$ denotes the vertical composition of $2$-cells and $\circ$ denotes the horizontal composition.

\newpage
The following property of locally posetal bicategories was already used in the proof of Theorem~\ref{thm.2}. Let $\mathscr{B}$ be a locally posetal bicategory with $1$-cells $f:A\longrightarrow B$ and $g:B\longrightarrow A$. Then the following data are equivalent
\begin{enumerate}
    \item $f\dashv g$ in $\mathscr{B}$
    \item there exist 2-cells
    \begin{equation}\label{eq.47}
    \operatorname{Id}_A\leq g\circ f\hspace{0.1cm} and \hspace{0.1cm}f\circ g\leq\operatorname{Id}_B.
    \end{equation}
\end{enumerate}

In other words, in a locally posetal bicategory, to specify the adjunction $f\dashv g$, it is sufficient to specify the two inequalities in~(\ref{eq.47}).

This leads us to the following definition.
\begin{definition}[\textbf{Adjunction in $\operatorname{OEnt}$}]
\textit{Let P,Q be open entropy systems and $\Phi:P\nrightarrow Q$, $\Psi:\
Q\nrightarrow P$ be $1$-cells. An adjunction $\Phi\dashv\Psi$ in $\operatorname{OEnt}$ precisely means that the following two $2$-cells exist}
\begin{equation}\label{eq.48}
\operatorname{Id}_P\leq\Psi\circ\Phi, \hspace{0.2cm}\Phi\circ\Psi\leq\operatorname{Id}_Q.
\end{equation}
\end{definition}
Since the 2-cells are implications, the inequalities in~(\ref{eq.48}) can be considered to be pointwise conditions. Let us unpack them completely below. 

Let $\Phi:P\nrightarrow Q$ and $\Psi:\
Q\nrightarrow P$ be $1$-cells, then the unit inequality $\operatorname{Id}_P\leq\Psi\circ\Phi$ is equivalent to
\begin{equation}\label{eq.49}
\forall p,p^{\prime}\in P,\left(p\leq p^{\prime}\right)\implies\left(\Psi\circ\Phi\right)\left(p,p^{\prime}\right)=\top.
\end{equation}
Using Definition~\ref{def.7}. (composition), we can rewrite~(\ref{eq.49}) as
\begin{equation}\label{eq.50}
\forall p,p^{\prime}\in P,\left(p\leq p^{\prime}\right)\implies\left(\exists q\in Q\right)\left(\Phi\left(p,q\right)=\top\land\Psi\left(q,p^{\prime}\right)=\top\right).
\end{equation}

Next, the counit inequality $\Phi\circ\Psi\leq\operatorname{Id}_Q$ is equivalent to
\begin{equation}\label{eq.51}
\forall q,q^{\prime}\in Q,\quad\left(\Phi\circ\Psi\right)\left(q,q^{\prime}\right)=\top\implies\left(q\leq q^{\prime}\right).
\end{equation}
Unpacking the composition yields
\begin{equation}\label{eq.52}
\forall q,q^{\prime}\in Q,\left(\exists p\in P\right)\left(\Psi\left(q,p\right)=\top\land\Phi\left(p,q^{\prime}\right)=\top\right)\implies\left
(q\leq q^{\prime}\right).
\end{equation}

Hence, the adjunction $\Phi\dashv\Psi$ in $\operatorname{OEnt}$ exactly corresponds to the pair of conditions~(\ref{eq.50}) and~(\ref{eq.52}).

From now on, we impose the boundary/bulk interpretation, specializing \(P\) and \(Q\) to the Landauer situation. Let us proceed to the following definition.
\begin{definition}[\textbf{Landauer interfaces}]
\textit{Let $B:=\left(B,\preceq_B\right)$ be the boundary information system, ordered by logical refinement/information-entropy content, and let $D:=\left(D,\preceq_D\right)$ be the bulk thermodynamic system, ordered by entropy/accessibility. Then a Landauer interface is a $1$-cell}
\begin{equation}\label{eq.53}
\Phi:B\nrightarrow D,
\end{equation}
\textit{in $\operatorname{OEnt}$, i.e. a Boolean profunctor}
\begin{equation}\label{eq.54}
\Phi:B^{op}\times D\longrightarrow\operatorname{Bool}.
\end{equation} 
\end{definition}

Similarly, we can define the so-called Landauer abstraction interface from the bulk $D$ back to the boundary $B$ as a 1-cell
$\Psi:D\nrightarrow B$ with $\Psi:D^{op}\times B\longrightarrow\operatorname{Bool}$.

The way we define the Landauer interface actually provides a rule by which we can physically compare two  bulk/boundary states in terms of information/entropy accessibility. In particular, \(b\preceq_B b^{\prime}\) means that \(b^{\prime}\) is at least as informative as \(b\), whereas \(d\preceq_D d^{\prime}\) means that \(d^{\prime}\) is at least as entropic or thermodynamically coarser than \(d\).

The following definition is the point where the present framework departs from the deterministic order-theoretic Landauer connection. We do not require a function assigning to each boundary state a unique bulk realization, nor a function assigning to each bulk state a unique boundary abstraction. Instead, implementation and abstraction are represented by feasibility interfaces. The adjunction inequalities then express the compatibility of these two interfaces, i.e. every entropy-compatible boundary comparison is realizable through the bulk, while bulk abstraction followed by re-implementation cannot produce a sharper thermodynamic state.
Therefore, the Landauer (abstraction) interface induces the notion of the Landauer adjunction.

\begin{definition}[\textbf{Landauer adjunction}]
\textit{A Landauer adjunction is an adjunction}\label{def.13}
\begin{equation}\label{eq.55}
\Phi\dashv\Psi\hspace{0.2cm}in\hspace{0.2cm}\operatorname{OEnt},
\end{equation}
i.e. the pair of inequalities
\begin{equation}\label{eq.56}
\operatorname{Id}_B\leq\Psi\circ\Phi,\hspace{0.2cm}\Phi\circ\Psi\leq\operatorname{Id}_D.
\end{equation}
\end{definition}

Equivalently, one can describe the Landauer adjunction in terms of unit (boundary feasibility) and counit, i.e. via the relations~(\ref{eq.50}) and~(\ref{eq.52}). We get
\begin{itemize}
    \item \textbf{Unit (boundary feasibility):}
    \begin{equation}\label{eq.57}
    \forall b,b^{\prime}\in B, b\leq_B b^{\prime}\implies\exists d\in D: \Phi\left(b,d\right)=\top\land\Psi\left(d,b^{\prime}\right)=\top,   
    \end{equation}
    \item\textbf{Counit:}
    \begin{equation}\label{eq.58}
    \forall d,d^{\prime}\in D:\left(\exists b\in B:\Psi\left(d,b\right)=\top\land\Phi\left(b,d^{\prime}\right)=\top\right)\implies d\leq_Dd^{\prime}.
    \end{equation}
\end{itemize}

Before we proceed further, it is obligatory to establish the following terminology. 

\vspace{0.2cm}
A Landauer adjunction \(\Phi\dashv\Psi\), where \(\Phi:B\nrightarrow D\) and \(\Psi:D\nrightarrow B\) satisfy Definition~\ref{def.13} will also be called an \textbf{exact Landauer adjunction}. Here, the adjective \textbf{exact} means that \(\Phi\) and \(\Psi\) form an ordinary internal adjunction in the bicategory \(\operatorname{OEnt}\), in the standard bicategorical sense \cite[Definition~6.1.1]{johnson20212}. Thus, there exist a unit and a counit \(2\)-cell 
\begin{equation}
\eta:\operatorname{Id}_B\implies\Psi\circ\Phi,\quad\varepsilon:\Phi\circ\Psi\implies\operatorname{Id}_D.
\end{equation}
Since \(\operatorname{OEnt}\) is locally posetal via Theorem~\ref{thm.2}, these \(2\)-cells (whenever they exist) are unique. Moreover, the bicategorical triangle identities are automatic because each side of either triangle identity is a parallel \(2\)-cell between the same pair of \(1\)-cells, and a hom-poset contains at most one such \(2\)-cell. Consequently, in \(\operatorname{OEnt}\), the existence of the adjunction is equivalent simply to the pair of inequalities
\begin{equation}\label{eq.E}
\operatorname{Id}_B\leq\Psi\circ\Phi,\quad\Phi\circ\Psi\leq\operatorname{Id}_D.
\end{equation}

Since the \(1\)-cells of \(\operatorname{OEnt}\) are Boolean-valued profunctors, the condition~(\ref{eq.E}) is equivalent to the following two pointwise conditions
\begin{equation}\label{eq.E1}
\forall b,b^{\prime}\in B,\quad \left(b\leq_Bb^{\prime}\right)\implies\left(\exists d\in D,\hspace{0.1cm}\Phi\left(b,d\right)=\top\wedge  \Psi\left(d,b^{\prime}\right)=\top\right),
\end{equation}
and
\begin{equation}\label{eq.E2}
\forall d,d^{\prime}\in D,\quad \left(\exists b\in B,\Psi\left(d,b\right)=\top\wedge\Phi\left(b,d^{\prime}\right)=\top\right)\implies d\leq_Dd^{\prime}.
\end{equation}

As we mention above, the term \textbf{exact} is used here solely to distinguish this ordinary, unweakened bicategorical adjunction from the lax, defective, approximate, or cost-bounded variants considered later in the quantitative setting (see Section~\ref{sec.3}). In particular, exactness does not mean that the two inequalities in~(\ref{eq.E}) are equalities, i.e. it does not require the unit or counit to be invertible. Thus, it is not an adjoint-equivalence condition. Nor does the adjective \textbf{exact} refer here to a Beck-Chevalley or exact-square condition. Technically, it does not include representability of \(\Phi\) and \(\Psi\) as part of the definition. Representability, when it occurs further below, is instead a consequence of the representability-rigidity result for exact adjunctions in \(\operatorname{OEnt}\), see Remark~\ref{rem.1}. 

For readers seeking further details, we recommend \cite{stubbe2005categorical},  where the reduction of the ordinary internal adjunction data to order inequalities is treated as a standard fact in locally ordered distributor/quantaloid settings.

Now, we shall take a step back and re-analyze the Definition~\ref{def.13}. Technically, the relation~(\ref{eq.57}) describing the atomic feasibility statement \(\Phi\left(b,d\right)=\top\) means that the bulk state \(d\) is an admissible realization of the boundary description \(b\). In fact, taking into account the second part of the unit condition, it is a much stronger statement. It says that every boundary comparison \(b\leq_B b^{\prime}\) is mediated by some bulk state \(d\) through implementation followed by abstraction. On the other hand, \(\Psi\left(d,b^{\prime}\right)=\top\) means that \(b^{\prime}\) is an admissible boundary readout of \(d\). Hence, the relation~(\ref{eq.E1}) states that every comparison already present in the boundary order can be mediated by an implementation-readout round trip through some bulk states. Condition~(\ref{eq.E2}) states that any bulk comparison created by abstraction followed by re-implementation is already contained in the native bulk order. Therefore, the term \textbf{exact} refers to precise satisfaction of these bicategorical unit and counit conditions, with no additional defect or cost term. 

Since our considerations are situated in a locally posetal bicategory $\operatorname{OEnt}$, once the adjunction $\Phi\dashv\Psi$ holds, which will be proven shortly, we obtain canonical endo-$1$-cells. Therefore, let us first define the following objects.
\begin{definition}[\textbf
{Boundary Landauer monad/closure interface}]\label{def.14}
\textit{Let \(T:=\Psi\circ\Phi:B\nrightarrow B\). The unit of the Landauer adjunction induces the unique \(2\)-cell \(\eta_T:\operatorname{Id}_B\implies T\) and the counit induces, by whiskering, the multiplication \(\mu_T:T\circ T\implies T\). Equivalently, since \(\operatorname{OEnt}\) is locally posetal,}
\begin{equation}
\operatorname{Id}_B\leq T,\quad T\circ T\leq T.
\end{equation}
\textit{We call \(T\) together with \(\eta_T\) and \(\mu_T\), the boundary Landauer monad.}
\end{definition}

Explicitly, \(T\left(b,b^{\prime}\right)=\top\) means that \(b^{\prime}\) is related to \(b\) by one bulk-mediated implementation/readout round trip. In expanded form, we have
\begin{equation}
T\left(b,b^{\prime}\right)=\top\iff\exists d\in D, \Phi\left(b,d\right)=\top\hspace{0.1cm}and\hspace{0.1cm}\Psi\left(d,b^{\prime}\right)=\top.
\end{equation}
The unit inequality \(\operatorname{Id}_B\leq T\) should be understood relationally, i.e. every native boundary comparison is also admitted by the bulk-mediated round-trip relation. Since \(T\) is generally a profunctor rather than an endomap \(B\longrightarrow B\), this inequality does not by itself describe a pointwise increase or decrease of information. This means that \(T\) is consistently evaluated on two boundary states \(T\left(b,b^{\prime}\right)\in\mathbf{Bool}\). Only in a special, representable situation, we replace \(T\) with an ordinary monotone endomap \(t:B\longrightarrow B\). In such a case, one must state explicitly which profunctor associated with \(t\) is being used and verify the inequalities separately.  

Under the refinement convention on \(B\), the physically relevant loss of distinguishability is captured by the possibility that distinct boundary states become mutually related by \(T\). Therefore, the irreversibility is represented by the collapse of the distinguishable boundary description in the induced quotient of \(B\), rather than by a literal map sending every state to a coarser state.  

Now, we proceed to the definition of the Landauer comonad. 
\begin{definition}[\textbf{Bulk Landauer comonad/interior interface}]\label{def.15}
\textit{The bulk endo-1-cell}
\begin{equation}\label{eq.60}
U:=\Phi\circ\Psi:D\nrightarrow D,
\end{equation}
\textit{is called the bulk Landauer comonad/interior interface. Furthermore, the counit of the Landauer adjunction induces \(\varepsilon_U:U\implies\operatorname{Id}_D\), while the unit induces, by whiskering, \(\delta_U: U\implies U\circ U\).}
\end{definition}
Specifically, the counit inequality $U\leq\operatorname{Id}_D$ describes the fact that no bulk over-sharpening occurs. Technically, an over-sharpening would occur if, starting from a bulk state \(d\), abstracting through \(\Psi\) followed by re-implementation through \(\Phi\) produced a state \(d^{\prime}\) satisfying \(d^{\prime}<_D d\). Nevertheless, the counit condition~(\ref{eq.58}) excludes precisely this possibility, i.e. whenever \(U\left(d,d^{\prime}\right)=\top\), one must have \(d\leq_D d^{\prime}\). If the bulk order is induced by an entropy function \(S_D\), this implies
\begin{equation}\label{eq.61}
S_D\left(d\right)\leq S_D\left(d^{\prime}\right),
\end{equation}
if the ordering $\leq_D$ is induced by $S_D$. Therefore, the bulk round trip cannot produce a bulk state with lower entropy (more specific bulk state) in comparison to the initial one. 

Going back to Definitions~\ref{def.14} and~\ref{def.15}, it is worth mentioning that in $\operatorname{OEnt}$, the Landauer monad $T$ is idempotent in the bicategorical sense, so $T\circ T\leq T$. Also, dually, we have $U\leq U\circ U$.

The above-mentioned considerations lead us to the conclusion that the Landauer adjunction induces boundary closure and bulk interface. This fact can be summarized in the following theorem.
\begin{theorem}\label{thm.3}
\textit{If $\Phi\dashv\Psi$ in $\operatorname{OEnt}$, then}
\begin{enumerate}
    \item the composite $T:=\Psi\circ\Phi: B\nrightarrow B$ is an idempotent boundary Landauer monad/closure interface on \(B\) in the sense that
    \begin{equation}\label{eq.62}
    \operatorname{Id}_B\leq T,\hspace{0.2cm} T^2= T,
    \end{equation}
    \item the composite $U:=\Phi\circ\Psi:D\nrightarrow D$ is a dual bulk comonad/interior interface in the sense that
    \begin{equation}\label{eq.63}
    U\leq\operatorname{Id}_D,\hspace{0.2cm}U^2=U.
    \end{equation}
\end{enumerate}
In particular, \(T\) and \(U\) are endo-\(1\)-cells/endo-interfaces.
\end{theorem}

More precisely, Theorem~\ref{thm.3} is the open-interface version of the categorical Landauer principle. In the deterministic setting, the closure \(GF\) says that implementation followed by abstraction cannot recover more information than was initially available. In the present setting, the same irreversibility statement holds without assuming a single implementation map or a single abstraction map. The loss/coarsening of recoverable boundary information is now a property of the composite feasibility interface \(T=\Psi\circ\Phi\). 
Within the structure of the proof of Theorem~\ref{thm.3}, we will repeatedly use two elementary facts that hold in an arbitrary bicategory and become order-theoretic statements when considering the locally posetal bicategory.

\newpage
\newtheorem{prop}{Proposition}
\begin{prop}\cite{johnson20212}\label{prop.1}
\textit{Let $\mathscr{B}$ be a locally posetal bicategory. Fix objects $A,B,C$.}
\begin{enumerate}
    \item If $f,f^{\prime}:A\longrightarrow B$ and $f\leq f^{\prime}$, then for any $g:B\longrightarrow C$,
    \begin{equation}\label{eq.64}
    g\circ f\leq g\circ f^{\prime}.
    \end{equation}
    \item If $g,g^{\prime}:B\longrightarrow C$ and $g\leq g^{\prime}$, then for any $f:A\longrightarrow B$,
    \begin{equation}\label{eq.65}
    g\circ f\leq g^{\prime}\circ f.
    \end{equation}
\end{enumerate}
\end{prop}
Basically, Proposition~\ref{prop.1} states that whiskering is monotone in any bicategory. The upcoming proposition simply follows from the definition \(f\cong g\implies f\leq g, g\leq f\implies f=g\), via the antisymmetry of the hom-posets.   

\begin{prop}\label{prop.2}
\textit{In a locally posetal bicategory, any isomorphism between parallel $1$-cells forces an equality.}
\end{prop}

In other words, we have invertible $2$-cells (associator and unitor)
\begin{equation}\label{eq.56}
a_{h,f,g}:\left(h\circ g\right)\circ f\cong h\circ\left(g\circ f\right), \hspace{0.2cm} \lambda_f:\operatorname{Id}\circ f\cong f,\hspace{0.2cm}\rho_f:f\circ\operatorname{Id}\cong f.
\end{equation}
Since $\leq$ is antisymmetric and an isomorphism gives both $\leq$ directions, we may treat these canonical isomorphisms as equalities of $1$-cells.

Now, let us proceed to the proof of Theorem~\ref{thm.3}. 

\begin{proof}[Proof of Theorem~\ref{thm.3}]
First of all, the inequality $\operatorname{Id}_B\leq\Psi\circ\Phi=T$ holds by the unit of the adjunction. 

Regarding the idempotency $T\circ T\leq T$, start with the definition
\begin{equation}\label{eq.67}
T\circ T=\left(\Psi\circ\Phi\right)\circ\left(\Psi\circ\Phi\right).
\end{equation}
We can reassociate it using the result of Proposition~\ref{prop.2} as
\begin{equation}\label{eq.68}
\left(\Psi\circ\Phi\right)\circ\left(\Psi\circ\Phi\right)=\Psi\circ\left(\Phi\circ\Psi\right
)\circ\Phi.
\end{equation}
Applying the counit inequality \(\Phi\circ\Psi\leq\operatorname{Id}_D\) of the Landauer adjunction, and whiskering by \(\Psi\) on the left and \(\Phi\) and the right via Proposition~\ref{prop.1}, we obtain
\begin{equation}\label{eq.69}
T\circ T=\Psi\circ\left(\Phi\circ\Psi\right)\circ\Phi\leq\Psi\circ\operatorname{Id}_D\circ\Phi\cong\Psi\circ\Phi.
\end{equation}

Simplifying with unit laws provides us
\begin{equation}\label{eq.71}
\Psi\circ\operatorname{Id}_D\circ\Phi=\Psi\circ\Phi=T.
\end{equation}

Putting all of the above-mentioned steps together yields
\begin{equation}\label{eq.72}
T\circ T=\Psi\circ\left(\Phi\circ\Psi\right)\circ\Phi\leq\Psi\circ\operatorname{Id}_D\circ\Phi=T.
\end{equation}

This proves the claim 1. of Theorem~\ref{thm.3}. As for the second claim, the inequality $\Phi\circ\Psi=U\leq\operatorname{Id}_D$ again holds by the counit inequality.

Considering the idempotency $U\leq U\circ U$, we start from the adjunction unit inequality $\operatorname{Id}_B\leq\Psi\circ\Phi$. Whiskering it on the left by $\Phi$ and on the right by $\Psi$, we get
\begin{equation}\label{eq.73}
\Phi\circ\operatorname{Id}_B\circ\Psi\leq\Phi\circ\left(\Psi\circ\Phi\right)\circ\Psi.
\end{equation}
Simplifying the left-hand side of~(\ref{eq.73}) by unit laws gives us
\begin{equation}\label{eq.74}
\Phi\circ\operatorname{Id}_B\circ\Psi=\Phi\circ\Psi=U.
\end{equation}
Reassociation of the right-hand side of~(\ref{eq.73}) yields
\begin{equation}\label{eq.75}
\Phi\circ\left(\Psi\circ\Phi\right)\circ\Psi=\left(\Phi\circ\Psi\right)\circ\left(\Phi\circ\Psi\right)=U\circ U.
\end{equation}

Therefore, we obtain
\begin{equation}\label{eq.76}
U\leq U\circ U,
\end{equation}
which proves the second claim of Theorem~\ref{thm.3}.

In conclusion, from the adjunction $\Phi\dashv\Psi$, we have constructed endo-1-cells $T=\Psi\circ\Phi$ and $U=\Phi\circ\Psi$ satisfying the conditions~(\ref{eq.62}) and (\ref{eq.63}), completing the proof of Theorem~\ref{thm.3}.
\end{proof}

\hspace{0.1cm}
Interestingly, one can think about the physical interpretation of Theorem~\ref{thm.3}. The closure interface \(T\) corresponds to the boundary \(\rightarrow\) bulk \(\rightarrow\) boundary relation. The inequality \(\operatorname{Id}_B\leq T\) says that this round trip relation extends the original boundary comparison, while \(T\circ T\leq T\) says that iterating the round trip adds no new boundary comparisons. Thus, the irreversible loss represented by \(T\) is stable after one round trip. In representable cases, this reduces to the usual closure-operator picture, and \(T\) becomes inflationary in the entropy order. In other words, returning from the physical realization cannot increase information on the boundary. It can only preserve the information or lose it. This exactly matches the author's statement in \cite{kycia2018landauer}, that the closure operator maps a state to a higher or equal one with respect to the Landauer connection.

Dually, the interior operator $U$ is somewhat of a boundary-consistent approximation of the bulk, it can only move us down the interior.

\vspace{0.2cm}
The original paper \cite{kycia2018landauer} treats the Landauer connection via monotone functions $F$ and $G$, with $F\dashv G$ being the Galois connection.

In order to embed that into our locally posetal bicategory $\operatorname{OEnt}$, we shall define the so-called \textbf{representable interface} of a monotone map $F:P\longrightarrow Q$ 

\begin{definition}[\textbf{Representable interface}]\label{def.16}
\textit{Let $F:P\longrightarrow Q$, $G:Q\longrightarrow P$ be monotone maps of posets. Then we can define the representable interface as the 1-cell $F_*:P\nrightarrow Q$ such that}
\begin{equation}\label{eq.66}
F_*\left(p,q\right)=\top\iff F\left(p\right)\leq q, \quad\textit{companion of F}. 
\end{equation}
We shall call $F_*$ the companion. Similarly, we can define its so-called conjoint as the 1-cell $G^*:P\nrightarrow Q$ such that
\begin{equation}\label{eq.77}
G^*\left(p,q\right)=\top\iff p\leq G\left(q\right),\quad\textit{conjoint of G}.
\end{equation}
\end{definition}

At this point, everything boils down to the fact that the adjunction map is equivalent to the adjunction of representable interfaces. Since both $F_*$ and $G^*$ are $1$-cells $P\nrightarrow Q$, they are directly comparable inside the hom-poset $\operatorname{OEnt}\left(P,Q\right)$. This is summarized in the following theorem.

The next theorem records the standard representable case. In the language of profunctor bicategories, a monotone map determines its associated representable interface, and an adjunction of monotone maps is equivalently expressed by the coincidence of the corresponding companion and conjoint interfaces. We include the statement because it shows exactly how the deterministic Landauer connection studied previously is recovered as a special case of the present open-interface formalism.

\begin{theorem}\label{thm.4}
\textit{Let $P$, $Q$ be posets (open entropy systems) and let $F:P\longrightarrow Q$, $G:Q\longrightarrow P$ be monotone maps with their representable interfaces (companion and conjoint) as in~(\ref{eq.66}) and~(\ref{eq.77}). Then the following data are equivalent}
\begin{enumerate}
    \item $F\dashv G$ as an adjunction of posets (Galois connection), so
    \begin{equation}\label{eq.80}
     \forall p\in P,\forall q\in Q: F\left(p\right)\leq_Qq\iff p\leq_PG\left(q\right).  
    \end{equation}
    \item The two interfaces $F_*$, $G^*:P\nrightarrow Q$ are equal, i.e.
    \begin{equation}
    F_*=G^*\quad in\hspace{0.1cm}\operatorname{OEnt}\left(P,Q\right).
    \end{equation}
    Equivalently, since $\operatorname{OEnt}\left(P,Q\right)$ is a poset, we have
    \begin{equation}
    F_*\leq G^*,\quad G^*\leq F_*.
    \end{equation}
\end{enumerate}
\end{theorem}

\vspace{0.1cm}
\begin{proof}[Proof of Theorem~\ref{thm.4}]

First of all, we shall prove that both $F_*$ and $G^*$ are well-defined. Define 
\begin{equation}
F_*\left(p,q\right)=\top\iff F\left(p\right)\leq q.
\end{equation}
We must show that \(F_*\) is monotone as a map \(P^{op}\times Q\longrightarrow\operatorname{Bool}\). Concretely, this means proving the contravariance in the \(P\)-variable and the covariance in the \(Q\)-variable. 

First, zooming on the contravariance in \(P\), let us take \(p,p^{\prime}\in
 P\) with \(p\leq p^{\prime}\). Assume that \(F_*\left(p^{\prime},q\right)=\top\). Then by definition of \(F_*\), this means
 \begin{equation}
 F\left(p^{\prime}\right)\leq q.
 \end{equation}
Since \(F\) is monotone and \(p\leq p^{\prime}\), we also obtain
\begin{equation}
F\left(p\right)\leq F\left(p^{\prime}\right).
\end{equation}
Combining these inequalities together gives us
\begin{equation}
F\left(p\right)\leq F\left(p^{\prime}\right)\leq q \implies F\left(p\right)\leq q.    
\end{equation}
Hence by definition, we again obtain
\begin{equation}
F_*\left(p,q\right)=\top.
\end{equation}
Thus
\begin{equation}
p\leq p^{\prime}\implies F_*\left(p^{\prime},q\right)\leq F_*\left(p,q\right),    
\end{equation}
which exactly corresponds to the contravariance in the first variable \(P\).

Now we can proceed to proving the covariance in \(Q\). Let us take \(q,q^{\prime}\in Q\) with \(q\leq q^{\prime}\). Presume that \(F_*\left(p,q\right)=\top\). Then \(F\left(p\right)\leq q\). By the transitivity of the order on \(Q\), we get
\begin{equation}
F\left(p\right)\leq q\leq q^{\prime}\implies F_*\left
(p,q^{\prime}\right)=\top.
\end{equation}
Hence,
\begin{equation}
q\leq q^{\prime}\implies F_*\left
(p,q\right)\leq F_*\left(p,q^{\prime}\right),
\end{equation}
and this corresponds to the covariance in the second variable \(Q\). Thus \(F_*\) is indeed a well-defined \(1\)-cell \(P\nrightarrow Q\).

Proceeding to the conjoint \(G^*\), let us define
\begin{equation}
G^*\left(p,q\right)=\top\iff p\leq G\left(q\right).
\end{equation}
Again, we must prove that \(G^*\) is monotone as a map \(P^{op}\times Q\longrightarrow\operatorname{Bool}\). 

Let us focus on the contravariance of \(G^*\) in the variable \(P\). Take \(p,p^{\prime}\in P\) with \(p\leq p^{\prime}\). Assume that \(G^*\left
(p^{\prime},q\right)=\top\). Then by its definition, we have
\begin{equation}
p^{\prime}\leq G\left(q\right).
\end{equation}
Since \(p\leq p^{\prime}\), the transitivity law yields
\begin{equation}
p\leq p^{\prime}\leq G\left
(q\right)\implies p\leq G\left(q\right).
\end{equation}
Therefore,
\begin{equation}
G^*\left(p,q\right)=\top,
\end{equation}
which corresponds to \(G^*\) being contravariant in \(P\).

Considering the covariance in \(Q\), let us take \(q,q^{\prime}\in Q\) with \(q\leq q^{\prime}\). Presume that \(G^*\left(p,q\right)=\top\). Then we directly have \(p\leq G\left(q\right)\). Since \(G\) is monotone and \(q\leq q^{\prime}\), we get
\begin{equation}
G\left(q\right)\leq G\left
(q^{\prime}\right).
\end{equation}
Combining both inequalities gives us
\begin{equation}
p\leq G\left(q\right)\leq G\left(q^{\prime}\right)\implies p\leq G\left(q^{\prime}\right).
\end{equation}
Therefore,
\begin{equation}
G^*\left(p,q^{\prime}\right)=\top.
\end{equation}
Hence \(G^*\) is covariant in \(Q\). In other words, \(G^*\) is also well-defined \(1\)-cell \(P\nrightarrow Q\).

Now, we shall focus on proving that the Galois connection of maps implies the adjunction of representable interfaces. Technically, we prove that the Galois connection \(F\dashv G\) is equivalent to the equality of the companion \(F_*\) and the conjoint \(G^*\).

\vspace{0.1cm}
\hspace{-0.6cm}\underline{(1)$\implies$ (2)}: Assume that \(F\dashv G\) as posets. By definition, this means
\begin{equation}
\forall p\in P,\hspace{0.1cm}\forall q\in Q,\quad F\left(p\right)\leq q\iff p\leq G\left(q\right).
\end{equation}
We want to prove that \(F_*=G^*\) as \(1\)-cells \(P\nrightarrow Q\).

Since \(1\)-cells are simply \(\operatorname{Bool}\)-valued predicates on pair \(p,q\), to prove equality, it is enough to prove pointwise equality of truth values for every pair \(p,q\). Hence, let \(p\in P\) and \(q\in Q\) be arbitrary. By definition of \(F_*\), we have
\begin{equation}
F_*\left(p,q\right)=\top\iff F\left(p\right)\leq q.
\end{equation}
Using the Galois condition, we get
\begin{equation}
p\leq G\left(q\right)\iff G^*\left(p,q\right)=\top.
\end{equation}
The combination of these equivalences provides us
\begin{equation}
F_*\left(p,q\right)=\top\iff G^*\left(p,q\right)=\top.
\end{equation}
Since this holds for every \(p\in P\) and \(q\in Q\), the two predicates agree pointwise. Hence we have the equality,
\begin{equation}
F_*=G^*,
\end{equation}
as interfaces \(P\nrightarrow Q\). 

Conversely, let us focus on the opposite implication.

\hspace{-0.6cm}\underline{(2)$\implies$ (1)}: Now presume that \(F_*=G^*\) as \(1\)-cells \(P\nrightarrow Q\). We must prove the adjunction \(F\dashv G\) as posets, i.e. that 
\begin{equation}
\forall p\in P,\hspace{0.1cm}\forall q\in Q,\quad F\left(p\right)\leq q\iff p\leq G\left(q\right).    
\end{equation}
So let us pick \(p\in P\) and \(q\in Q\) arbitrarily. Because \(F_*=G^*\), their values at \(\left(p,q\right)\) are equal, i.e. \(F_*\left(p,q\right)=G^*\left
(p,q\right)\). Now we shall unpack both sides. By definition of \(F_*\), we have
\begin{equation}
F_*\left(p,q\right)=\top\iff F\left(p\right)\leq q.
\end{equation}
By definition of \(G^*\) we get
\begin{equation}
G^*\left(p,q\right)=\top\iff p\leq G\left(q\right). 
\end{equation}
Since both truth values are identical, we obtain
\begin{equation}
F\left(p\right)\leq q\iff p\leq G\left(q\right).
\end{equation}
Additionally, because \(p\) and \(q\) were arbitrary, this inequality holds for all \(p\in P\) and \(q\in Q\). Therefore, proving the adjunction \(F\dashv G\).

This completes the proof of Theorem~\ref{thm.4}.
\end{proof}

\vspace{0.1cm}

This completes the recovery of the classical Galois/Landauer connection inside \(\mathrm{OEnt}\). Categorically, the argument is the usual companion-conjoint characterization of a poset adjunction. However, from the physical perspective, it is important because it identifies deterministic implementation maps as the representable interfaces among all open feasibility interfaces. The passage from monotone maps to arbitrary profunctors is therefore exactly the passage from single-valued implementations to many-to-many, environment-dependent thermodynamic realizations. In other words, the Galois connection is recovered through the statement that the two canonical ways of presenting the same interface coincide (equality of two parallel interfaces): 
\begin{itemize}
    \item The interface generated by the left adjoint \(F\),
    \begin{equation}
        F_*\left(p,q\right)\iff F\left(p\right)\leq q,
    \end{equation}
    \item The interface generated by the right adjoint \(G\),
    \begin{equation}
    G^*\left(p,q\right)\iff p\leq G\left(q\right).
    \end{equation}
\end{itemize}
This Galois law is precisely the assertion that these two predicates define the identical relation between \(P\) and \(Q\). Equivalently, we have the exact bicategorical translation of the classical Galois adjunction from \cite{kycia2018landauer}.

Notice that there is an interesting consequence of Theorem~\ref{thm.4} that deserves individual emphasis. Namely that every monotone map gives rise to a canonical profunctor adjunction between its companion and conjoint. This is just another feature of the profunctorial setting. Nonetheless, we include it to emphasize that representable interfaces are not ad hoc additions to the theory, but are built into the ordinary proarrow/profunctor formalism \cite{fong2019invitation}.

\newtheorem{lemma}{Lemma}

\newtheorem{rem}{Remark}
\begin{rem}\label{rem.1}(\textbf{Representability-rigidity result})
Let \(F:P\longrightarrow Q\) be a monotone map. Recall the companion and conjoint maps as:
\begin{equation}
F_*:P\nrightarrow Q,\quad F_*\left(p,q\right)=\top\iff F\left(p\right)\leq q,
\end{equation}
\begin{equation}
F^*:Q\nrightarrow P,\quad F^*\left(q,p\right)=\top\iff q\leq F\left(p\right).
\end{equation}
Then every monotone map \(F\) determines a canonical profunctor adjunction
\begin{equation}
F_*\dashv F^*\quad in\hspace{0.1cm}\operatorname{OEnt}.
\end{equation}
\end{rem}

\hspace{0.2cm}
\begin{proof}[Proof of Remark~\ref{rem.1}]
We prove the following inequalities
\begin{equation}
\operatorname{Id}_P\leq F^*\circ F_*,\quad F_*\circ F^*\leq\operatorname{Id}_Q.
\end{equation}
For the unit, fix \(p,p^{\prime}\in P\) and assume that \(p\leq p^{\prime}\). We must prove that
\begin{equation}
\left(F^*\circ F_*\right)\left(p,p^{\prime}\right)=\top. 
\end{equation}
Via composition, this means producing some \(q\in Q\) such that 
\begin{equation}
F_*\left(p,q\right)=\top\quad and\quad F^*\left(q,p^{\prime}\right)=\top.
\end{equation}
Let us choose \(q:=F\left(p\right)\), then we get \(F_*\left(p,F\left(p\right)\right)=\top\) because \(F\left(p\right)\leq F\left(p\right)\) and \(F^*\left(F\left(p\right),p^{\prime}\right)=\top\) because \(F\left(p\right)\leq F\left(p^{\prime}\right)\) is given by the monotonicity of \(F\) and the inequality \(p\leq p^{\prime}\). Hence, the unit holds.

Considering the counit, let us fix \(q,q^{\prime}\in Q\) and assume that
\begin{equation}
\left(F_*\circ F^*\right)\left(q,q^{\prime}\right)=\top.
\end{equation}
Then there exists \(p\in P\) such that
\begin{equation}
F^*\left(q,p\right)=\top\quad and\quad F_*\left(p,q^{\prime}\right)=\top. 
\end{equation}
Unpacking both sides gives us
\begin{equation}
q\leq F\left(p\right),\quad F\left(p\right)\leq q^{\prime}\implies q\leq q^{\prime}. 
\end{equation}
Therefore
\begin{equation}
\operatorname{Id}_Q\left(q,q^{\prime}\right)=\top,
\end{equation}
so \(F_*\circ F^*\leq\operatorname{Id}_Q\), proving the proposed adjunction.
\end{proof}

In summary, this remark isolates the categorical fact that every monotone map generates a profunctor adjunction. On the other hand, Theorem~\ref{thm.4} states the sharper fact that a specific poset adjunction \(F\dashv G\) is recovered as the equality of the companion and the conjoint maps \(F_*=G^*\). 

Additionally, we shall observe the following. Let \(\Phi:P\nrightarrow Q\) and \(\Psi:Q\nrightarrow P\) be Boolean profunctors satisfying \(\operatorname{Id}_P\leq\Psi\circ\Phi\) and \(\Phi\circ\Psi\leq\operatorname{Id}_Q\). For each \(p\in P\), the unit at \(\left(p,p\right)\) supplies some \(q_p\in Q\) such that
\begin{equation}
\Phi\left(p,q_p\right)=\top,\quad \Psi\left(q_p,p\right)=\top.
\end{equation}
Now, suppose that both \(q_1,q_2\) have this property. Since
\begin{equation}
\Psi\left(q_1,p\right)=\top,\quad \Phi\left(p,q_2\right)=\top,
\end{equation}
the counit implies \(q_1\leq q_2\). Interchanging \(q_1\) and \(q_2\) gives \(q_2\leq q_1\). Since \(Q\) is a poset, we get \(q_1=q_2\). Thus, every \(p\in P\) determines a unique canonical element \(F\left(p\right):=q_p\). Furthermore, 
\begin{equation}
\Phi\left(p,q\right)=\top\iff F\left(p\right)\leq q.
\end{equation}
Indeed, the reverse implication follows from the covariance of \(\Phi\), while for the forward implication, we use \(\Psi\left(F\left(p\right),p\right)=\top\), \(\Phi\left(p,q\right)=\top\) and the counit. Similarly,
\begin{equation}
\Psi\left(q,p\right)=\top\iff q\leq F\left(p\right),
\end{equation}
Finally, this implies that \(F\) is monotone. Thus
\begin{equation}
\Phi=F_*,\quad \Psi=F^*.
\end{equation}
Therefore, exact adjunctions in \(\operatorname{OEnt}\) exhibit a rigidity phenomenon. Although the ambient bicategory \(\operatorname{OEnt}\) contains genuinely non-representable, many-to-many feasibility interfaces, every exact adjunction \(\Phi\dashv\Psi\) is necessarily represented by the companion and conjoint of a unique monotone map \(F:P\longrightarrow Q\), namely \(\Phi=F_*\), \(\Psi=F^*\). Hence, openness occurs at the level of arbitrary interfaces, whereas imposing an exact adjunction selects a distinguished representable subclass.  

Consequently, genuinely non-representable implementation/abstraction pairs cannot satisfy the exact adjunction axioms in the bicategory \(\operatorname{OEnt}\). In order to accommodate such pairs while retaining genuinely many-to-many behavior, one would have to weaken the exact adjunction condition. For example, this could be done by considering a lax, defective, or cost-bounded notion of adjunction. We do not focus on developing such a weakening in the present Boolean framework since we will devote more attention to this topic in Section~\ref{sec.3}. 

\vspace{0.2cm}
The construction of representable interfaces in Definition~\ref{def.16} and Theorem~\ref{thm.4} can be formulated analogously via the bulk/boundary entropy systems. We just replace $P$ with $B$ as the boundary system of logical/informational macrostates and $Q$ with $D$ as the bulk system of thermodynamic/physical macrostates as posets. 

Subsequently, it is possible to generalize the result of Theorem~\ref{thm.4} further and present the Landauer adjunction as a universal property. This claim is based on a classical result from the field of posetal bicategories and it explains why the right Landauer interface can be regarded as a best abstraction, while the left Landauer interface can be regarded as the corresponding best implementation interface.
\begin{theorem}\cite{johnson20212}\label{thm.5}
\textit{Let $\mathscr{B}$ be a locally posetal bicategory. Fix objects $A,B$ and let $f: A\longrightarrow B$, $g: B\longrightarrow A$ be $1$-cells in $\mathscr{B}$. Then the following data are equivalent:}
\begin{enumerate}
    \item $f\dashv g$
    \item  For every object $X$ and every pair of $1$-cells $x:X\longrightarrow A$, $y:X\longrightarrow B$ in $\mathscr{B}$, we have
    \begin{equation}\label{eq.111}
    f\circ x\leq y\iff x\leq g\circ y.
    \end{equation}
\end{enumerate}
\textit{Moreover, the equivalence~(\ref{eq.111}) is natural and monotone in $x$ and $y$.}
\end{theorem}

Now substituting $\mathscr{B}=\operatorname{OEnt}$ and interpreting $A$ as a boundary $(B)$ and $B$ as a bulk $(D)$ in Theorem~\ref{thm.5}, we can display the result of Theorem~\ref{thm.4} as the following equivalence
\begin{equation}\label{eq.112}
\operatorname{OEnt}\left(B,D\right)\left(\Phi,\chi\right)\cong\operatorname{OEnt}\left(B,B\right)\left(\operatorname{Id}_B,\Psi\circ\chi\right).
\end{equation}

Indeed, we just took \(X=B\), \(x=\operatorname{Id}_B:B\nrightarrow B\) and \(y=\chi:B\nrightarrow D\), so the adjunction universal property gives
\begin{equation}\label{eq.113}
\Phi\circ\operatorname{Id}_B\leq\chi\iff\operatorname{Id}_B\leq\Psi\circ\chi.
\end{equation}
Using the unit law $\Phi\circ\operatorname{Id}_B=\Phi$, we obtain
\begin{equation}\label{eq.114}
\Phi\leq\chi\iff\operatorname{Id}_B\leq\Psi\circ\chi,
\end{equation}
which exactly represents the universal property as desired. 

In fact, this is a precise categorical meaning of $\Psi$ being the best abstraction, as a right adjoint, and $\Phi$ being the best interface that makes the unit inequality hold. This universal property states that a Landauer adjunction does not merely provide two compatible interfaces. In particular, it makes one interface universal among those implementations whose abstraction satisfies the prescribed boundary constraint. This is the exact sense in which the right interface behaves as an optimal abstraction and the left interface behaves as its corresponding implementation mechanism.

\subsection{\textbf{Eilenberg-Moore object of an idempotent Landauer monad in \texorpdfstring{$\operatorname{OEnt}$}{OEnt}}}

In category theory, a monad $T:C\longrightarrow C$ packages a completion/closure process in the following way
\begin{itemize}
    \item $\eta_T:\operatorname{Id}\implies T$ states that every object embeds into its completion
    \item $\mu_T: T\circ T\implies T$ states that applying the completion twice yields no further change.
\end{itemize}

The Eilenberg-Moore category $C^T$ collects the objects that are already compatible with the completion, i.e. objects equipped with an action $T\left(X\right)\rightarrow X$ satisfying the coherence condition.

This is where the idempotence of the monad $T$ becomes important. When a monad is idempotent, then $\mu_T$ is invertible. Hence, the completion turns into a reflection. In fact, applying $T$ sends us into a full subcategory of $T$-complete objects, and those objects are precisely the fixed points of T. In that case, the Eilenberg-Moore category is equivalent to a reflective subcategory of $C$. This corresponds to the classical connection
\begin{equation}\label{eq.115}
idempotent\hspace{0.1cm}monads\longleftrightarrow reflective\hspace{0.1cm}subcategories\longleftrightarrow splittings\hspace{0.1cm} of\hspace{0.1cm}idempotents.
\end{equation}

Within a bicategory, a monad on an object $P$ is an endo-1-cell $T:P\longrightarrow P$ with 2-cells $\eta_T:\operatorname{Id}_P\implies T$ and $\mu_T:T\circ T\implies T$. The Eilenberg-Moore object $P^T$ is the bicategorical substitute for the Eilenberg-Moore category, so it is essentially an object equipped with a universal $T$-algebra that represents all of the $T$-algebras. 

As we have seen in all of our results above, in a locally posetal bicategory, the ideas of Eilenberg-Moore category and $T$-algebras become especially clear because every hom corresponds to a poset. Thus, all of the algebra axioms become inequalities, and idempotence becomes equality. Therefore, we shall write $f\leq g$ for the unique $2$-cell $f\implies g$, if it exists. 

Let us summarize these considerations in the following definitions.
\begin{definition}\cite{lack20092}
\textit{Let $\mathscr{B}$ be a locally posetal bicategory, and let \(T:P\longrightarrow P\) be a monad. Next, let us fix an object \(X\). A $T$-algebra with a domain $X$ consists of a \(1\)-cell $x:X\longrightarrow P$ together with an action (\(2\)-cell) $a_x:T\circ x\implies x$, i.e. $T\circ x\leq x$}
\end{definition}

Let us mention that in a locally posetal bicategory, once an action \(a:T\circ x\leq x\) exists, the two usual algebra diagrams become equal automatically because there is at most one \(2\)-cell between fixed parallel \(1\)-cells, and the associativity axiom is automatic. Additionally, whiskering the unit \(\eta_T\) by \(x\) gives us \(x=\operatorname{Id}_P\circ x\leq T\circ x\), so the unit condition is automatic as well.

Hence, we can directly identify the poset of $T$-algebras from $X$ with the set
\begin{equation}\label{eq.116}
\operatorname{Alg}_T\left(X\right):=\{x:X\longrightarrow P\hspace{0.1cm}\vert\hspace{0.1cm}T\circ x\leq x\},
\end{equation}
ordered by the hom-order of $\mathscr{B}\left(X,P\right)$.

Since $\operatorname{Id}_P\leq T$, every $x\in\operatorname{Alg}_T\left(X\right)$ satisfies
\begin{equation}\label{eq.117}
x=\operatorname{Id}_P\circ x\leq T\circ x\leq x,
\end{equation}
thus $x=T\circ x$ for every $x\in\operatorname{Alg}_T\left(X\right)$. This kind of fixed point phenomenon is the underlying reason why idempotent monads behave like reflections. 

Now, let us proceed to the definition of the Eilenberg-Moore object.
\begin{definition}\cite{lack20092}
An Eilenberg-Moore object for a monad $T:P\longrightarrow P$ corresponds to
\begin{enumerate}
    \item an object $P^T$,
    \item a $1$-cell $u:P^T\longrightarrow P$,
    \item a $2$-cell $\alpha: T\circ u\implies u$,
\end{enumerate}
such that for every object $X$, the postcomposition with $u$ induces an isomorphism of posets
\begin{equation}\label{eq.118}
\mathscr{B}\left(X,P^T\right)\cong\operatorname{Alg}_T\left(X\right),
\end{equation}
which is natural in $X$.
\end{definition}
Equivalently, one can say that the $1$-cell $u:P^T\longrightarrow P$, equipped with its \(T\)-algebra action \(\alpha: T\circ u\implies u\), represents $T$-algebras. At this point, we shall apply all of the above considerations to our setup of the bicategory $\operatorname{OEnt}$. From the point of view of our profunctorial notation, it is important to note that in ordinary category, a monad is an endofunctor. However, in a bicategory, a monad on an object \(B\) is an endo-\(1\)-cell \(T:B\nrightarrow B\) equipped with unit and multiplication \(2\)-cells. Since in \(OEnt\) the \(1\)-cells are profunctors/open interfaces, the Landauer monad \(T\) will be treated as an endo-interface \(B\nrightarrow B\).

We have the following lemma.

\begin{lemma}\label{lem.2}
If \(T:B\nrightarrow B\) satisfies \(\operatorname{Id}_B\leq T\) and \(T\circ T\leq T\), then \(T\circ T=T.\)
\end{lemma}

\vspace{0.1cm}
\begin{proof}[Proof of Lemma~\ref{lem.2}]
Since \(\operatorname{Id}_B\leq T\), then by monotonicity of the composition, we obtain
\begin{equation}
T=T\circ\operatorname{Id}_B\leq T\circ T.
\end{equation}
Together with \(T\circ T\leq T\), the antisymmetry in the hom-poset \(\operatorname{OEnt}\left(B,B\right)\) yields
\begin{equation}
T\circ T=T.
\end{equation}
\end{proof}

These are precisely the conditions that actually state that $T\circ T=T$, thus T is an idempotent monad. Furthermore, define a binary relation $\preceq_T$ on the underlying set of $B$ by
\begin{equation}\label{eq.120}
b\preceq_T b^{\prime}\iff T\left(b,b^{\prime}\right)=\top.
\end{equation}

Next, let $\sim_T$ represent the induced equivalence relation
\begin{equation}\label{eq.121}
b\sim_T b^{\prime}\iff\left(b\preceq_T b^{\prime}\right)\land\left(b^{\prime}\preceq_Tb\right).
\end{equation}
Furthermore, \(\operatorname{OEnt}\) is used only up to order-theoretic structure, so any poset obtained by the present construction is admitted as its object.  

We now arrive at the main structural result of the paper. In ordinary category theory, idempotent monads correspond to reflective subcategories, and their Eilenberg-Moore algebras may equivalently be characterized as fixed objects of the monad \cite{DiLiberti2020}. However, in a general bicategory, the existence of the Eilenberg-Moore object is an additional representability property and should not be assumed automatically \cite{Street1972}. What is specific to the present work is the origin and interpretation of the idempotent monad, i.e. here \(T=\Psi\circ\Phi\) is not an arbitrary closure operator, but the boundary round trip obtained by implementing boundary data in a thermodynamic bulk and then abstracting it back to the boundary. The resulting Eilenberg–Moore object therefore has a concrete thermodynamic meaning - it is the universal sector of boundary information that is stable under the chosen bulk-mediated implementation/readout process. 

Before we proceed to the following theorem, we need to clarify certain technicalities. Since \(T\) is generally a relation-valued endo-interface, we do not define the stable sector by literal fixed points of the form \(T\left(b\right)=d\). Instead, we define the preorder generated by \(T\) via the relation~(\ref{eq.120}). The stable boundary sector is then obtained from this preorder by identifying mutually \(T\)-related boundary states as in~(\ref{eq.121}). Thus, the object \(B_T:=B/\sim_T\) will be a relational/profunctorial analogy of the fixed-point subposet of a closure operator. 

The abstract categorical mechanisms underlying the following result are classical: idempotent monads, their splittings, and Eilenberg-Moore objects. These belong to the standard formal theory of monads in bicategories \cite{Street1972}, while adjoint distributors and their relation to Cauchy completeness and idempotent splitting are part of the well established theory of enriched distributors \cite{stubbe2005categorical}. What is specific to the present framework is the explicit identification of the idempotent \(T:=\Psi\circ\Phi:B\nrightarrow B\) arising from an exact Landauer adjunction with a thermodynamically meaningful round-trip preorder on boundary descriptions
\begin{equation}
b\preceq_T b^{\prime}\iff T\left(b,b^{\prime}\right)=\top,
\end{equation}
and the subsequent construction of the quotient (poset reflection) \(B_T:=B/{\sim_T}\), where the relation \(\sim_T\) is defined as \(b\sim_T b^{\prime}\iff b\preceq_T b^{\prime}\land b^{\prime}\preceq_T b\).

In the theorem below, we prove that this explicitly constructed quotient is not merely an order-theoretic reduction of \(B\), but that it realizes the Eilenberg-Moore object of the Landauer round-trip monad \(T\). Thus, the categorical universal property acquires a concrete Landauer interpretation as \(B_T\) classifies precisely the boundary distinctions that remain observable after implementation through the bulk and subsequent readout. In this sense, the object \(B_T\) represents the bulk-visible boundary sector determined by the chosen Landauer coupling. Notice that no identification of \(B_T\) with the full microscopic bulk \(D\) is asserted. 

\newpage
Now we introduce our main structural result. 

\begin{theorem}\label{thm.6}(\textbf{Boundary-stable sector and bulk visible reconstruction})
\textit{Let \(T:=\Psi\circ\Phi:B\nrightarrow B\) be the boundary Landauer monad induced by an exact Landauer adjunction \(\Phi\dashv\Psi\), and let \(B_T:=B/{\sim_T}\) be the poset reflection of the round-trip preorder induced by \(T\). Equivalently, \(B_T\) is the set of equivalence classes with the order \(\preceq_T\) defined as}
\begin{equation}\label{eq.122}
[b]\leq_T[b^{\prime}]\iff b\preceq_T b^{\prime}.
\end{equation}
Then \(B_T\) gives an explicit splitting and Eilenberg-Moore realization of the Landauer round-trip monad. More precisely:
\begin{enumerate}
    \item $\left(B_T,\leq_T\right)$ is a well-defined poset, i.e. object of the bicategory $\operatorname{OEnt}$
    \item There exist $1$-cells
    \begin{equation}\label{eq.123}
    r:B\nrightarrow B_T,\hspace{0.2cm}i:B_T\nrightarrow B,
    \end{equation}
    elementwise
    \begin{equation}\label{eq.124}
    r\left(b,[c]\right)=\top\iff b\preceq_T c,\hspace{0.2cm} i\left([c],b\right)=\top\iff c\preceq_T b.
    \end{equation}
    \item These \(1\)-cells satisfy \begin{equation}\label{eq.125}
    i\circ r=T, \hspace{0.2cm} r\circ i=\operatorname{Id}_{B_T},
    \end{equation}
    hence \(r\dashv i\) in \(\operatorname{OEnt}\),
    \item Let \(\alpha:T\circ i\implies i\) be the canonical \(2\)-cell induced by the equality \(T\circ i=i\). Then the triple \(\left(B_T,i,\alpha\right)\) is an Eilenberg-Moore object for the monad \(T\). Hence, for every object $X\in\operatorname{OEnt}$, the postcomposition with $i$ induces an isomorphism of posets
    \begin{equation}\label{eq.126}
    \operatorname{OEnt}\left(X,B_T\right)\cong\operatorname{Alg}_T\left(X\right),
    \end{equation}
    natural in \(X\), where $\operatorname{Alg}_T\left(X\right)=\{x:X\nrightarrow B\hspace{0.1cm}\vert\hspace{0.1cm}T\circ x\leq x\}\subseteq\operatorname{OEnt}\left(X,B\right)$.
\end{enumerate}
\end{theorem}

\begin{proof}[Proof of Theorem~\ref{thm.6}]

First of all, we have to show that $\preceq_T$ is a preorder and it extends the original order.

Focusing on reflexivity, let us take any $b\in B$ and since $\operatorname{Id}_B\leq T$ with $\operatorname{Id}_B\left(b,b\right)=\top$ (because $b\leq b$), we obtain $T\left(b,b\right)=\top$. Hence, $b\preceq_T b$.

Considering the transitivity property, we shall presume that $b\preceq_T c$ and $c\preceq_T$ d, i.e. $T\left(b,c\right)=\top$ and $T\left(c,d\right)=\top$. Then, according to the definition of composition, we acquire
\begin{equation}\label{eq.127}
\left(T\circ T\right)\left(b,d\right)=\top.
\end{equation}

One can equivalently rewrite~(\ref{eq.127}) as
\begin{equation}\label{eq.128}
\left(T\circ T\right)\left(b,d\right)=\top\iff\exists x:T\left(b,x\right)\land T\left(x,d\right).
\end{equation}

Since $T\circ T\leq T$, we get
\begin{equation}\label{eq.129}
\left(T\circ T\right)\left(b,d\right)=\top\implies T\left(b,d\right)=\top.
\end{equation}
Thus $b\leq_T d$. This implies that $\preceq_T$ is a preorder. 

Consequently, if $b\leq b^{\prime}$, then $\operatorname{Id}_B\left(b,b^{\prime}\right)=\top$, and because $\operatorname{Id}_B\leq T$, we obtain $T\left(b,b^{\prime}\right)=\top$, i.e. $b\preceq_T b^{\prime}$. In other words, this extends the original preorder. 

In the next step, we will show that $\leq_T$ is well-defined on $B_T$ and makes $B_T$ a poset. Basically, we must show that if $b\sim_T b_1$ and $b^{\prime}\sim_T b_1^{\prime}$, then
\begin{equation}\label{eq.130}
b\preceq_T b^{\prime}\iff b_1\preceq_T b_1^{\prime}.
\end{equation}

Let us presume that $b\sim_T b_1$ and $b^{\prime}\sim_T b_1^{\prime}$. Then we get
\begin{equation}\label{eq.131}
b\preceq_T b_1,\hspace{0.2cm}b_1\preceq_T b, \hspace{0.5cm}b^{\prime}\preceq_T b_1^{\prime},\hspace{0.2cm} b_1^{\prime}\preceq_T b^{\prime}.
\end{equation}

Now, if $b\preceq_T b^{\prime}$, then by transitivity with $b_1\preceq_T b$ we get $b_1\preceq_T b^{\prime}$. Then, considering $b^{\prime}\preceq_T b_1^{\prime}$, we obtain $b_1\preceq_T b_1^{\prime}$.

Conversely, if we presume that $b_1\preceq_T b_1^{\prime}$, then with $b\preceq_T b_1$, we have $b\preceq_T b_1^{\prime}$. Subsequently, with $b_1^{\prime}\preceq_T b^{\prime}$, we acquire $b\preceq_T b^{\prime}$.

Thus, the issue of well-definedness depends only on equivalence classes, so $[b]\leq_T [b^{\prime}]$ is well-defined. 

Considering the poset axioms, the first that comes to hand is reflexivity. Well, we have $[b]\leq_T[b]$ because $b\preceq_T b$. When coming to transitivity, if $[b]\leq_T[c]$ and $[c]\leq_T[d]$, then $b\preceq_T c$ and $c\preceq_T d$. Thus $b\preceq_T d$, so $[b]\leq_T[d]$. Taking into account the antisymmetry, if $[b]\leq_T[b^{\prime}]$ and $[b^{\prime}]\leq_T[b]$, then $b\sim_T b^{\prime}$. Hence $[b]=[b^{\prime}]$ via definition of the quotient. Therefore, $B_T$ is a poset. 

In the following step, we shall define the functors $r$ and $i$ and then prove that they are $1$-cells in $\operatorname{OEnt}$.

Recall that a $1$-cell (profunctor) $P\nrightarrow Q$ corresponds to a functor $P^{op}\times Q\longrightarrow\operatorname{Bool}$, i.e. it is monotone contravariantly in $P$ and covariantly in $Q$. Let us now focus on the well-definedness of $r$ and $i$ with respect to representatives. 

We define
\begin{equation}\label{eq.132}
r\left(b,[c]\right)=\top\iff T\left(b,c\right)=\top,\hspace{0.2cm}i\left([c],b\right)=\top\iff T\left(c,b\right)=\top.
\end{equation}
Suppose that $c\sim_T c^{\prime}$. Then $c\preceq_T c^{\prime}$ and $c^{\prime}\preceq_T c$. Hence, for $r: T\left(b,c\right)=\top$ and $T\left(c,c^{\prime}\right)=\top,$ implies $T\left(b,c^{\prime}\right)=\top$, which comes from the transitivity of $\preceq_T$. Identically, we have $T\left(b,c^{\prime}\right)=\top\implies T\left(b,c\right)=\top$. Therefore, $T\left(b,c\right)=\top\iff T\left(b,c^{\prime}\right)=\top$. Thus, we can consider $r$ to be well-defined. 

For $i$, we have $T\left(c,b\right)=\top$ and $T\left(c^{\prime},c\right)=\top\implies T\left(c^{\prime},b\right)=\top$. Similarly, we can perform the same procedure in the other direction. Thus, $i$ is also well-defined.

As for whether $r$ and $i$ are profunctors, we must show that they are monotone in each argument. For \(r\), this means proving contravariance in \(B\) and covariance in \(B_T\), plus the other way around for \(i\). Therefore, focusing on the contravariance in \(B\), we shall assume that \(b\leq b^{\prime}\) in \(B\) and \(r\left(b^{\prime},[c]\right)=\top\). Then we have \(b^{\prime}\preceq_T c\). Since the original order extends to \(\preceq_T\), we get
\begin{equation}
b\leq b^{\prime}\implies b\preceq_T b^{\prime}. 
\end{equation}
By transitivity of \(\preceq_T\), we obtain
\begin{equation}
b\preceq_Tb^{\prime}\preceq_T c,
\end{equation}
thus \(b\preceq_T c\), so \(r\left(b,[c]\right)=\top\). 

Proceeding to the covariance in \(B_T\), we presume that \([c]\leq_T [d]\) and \(r\left(b,[c]\right)=\top\). Then \(c\preceq_T d\) and \(b\preceq_T c\). Via transitivity, we get
\begin{equation}
b\preceq_T d, 
\end{equation}
so \(r\left(b,[d]\right)=\top\). Thus \(r\) is a valid \(1\)-cell.

Now we shall perform a similar procedure with \(i\). In order to show that \(i:B_T^{op}\times B\longrightarrow\operatorname{Bool}\) is monotone, we must go through the contravariance in \(B_T\) and covariance in \(B\). Focusing on the first one, we assume that \([c]\leq_T [d]\) and \(i\left([d],b\right)=\top\). Then \(c\preceq_Td\) and \(d\preceq_T b\). Then by transitivity, we have
\begin{equation}
c\preceq_T b,
\end{equation}
so \(i\left([c],b\right)=\top.\)

Proceeding to the covariance in \(B\), we assume that \(b\leq b^{\prime}\) and \(i\left([c],b\right)=\top\). Then \(c\preceq_T b\). Since the original order extends to \(\preceq_T\), we have \(b\preceq_T b^{\prime}\). Therefore,
\begin{equation}
c\preceq_T b\preceq_T b^{\prime},
\end{equation}
hence \(i\left([c],b^{\prime}\right)=\top\). Thus \(i\) is also a valid \(1\)-cell.

Following the structure of Theorem~\ref{thm.6}, we must compute $i\circ r$ and $r\circ i$ and prove the required equalities. First of all, let us fix $b,b^{\prime}\in B$. Then by definition of composition in $\operatorname{OEnt}$, we get
\begin{equation}\label{eq.133}
\left(i\circ r\right)\left(b,b^{\prime}\right)=\top\iff\exists [c]\in B_T: r\left(b,[c]\right)=\top\land i\left([c],b^{\prime}\right)=\top.
\end{equation}
Unpacking $r$ and $i$ provides us the following equivalent expression
\begin{equation}\label{eq.134}
\exists c\in
 B: T\left(b,c\right)=\top\land T\left(c,b^{\prime}\right)=\top.
\end{equation}

However, this exactly corresponds to the definition of $\left(T\circ T\right)\left(b,b^{\prime}\right)=\top$. Hence
\begin{equation}\label{eq.135}
i\circ r=T\circ T,
\end{equation}
as a pointwise equality. 

Recalling the results of Theorem~\ref{thm.3}, the monad $T$ is idempotent, i.e. $T\circ T=T$. Thus
\begin{equation}\label{eq.136}
i\circ r=T. 
\end{equation}

Proceeding to the second equality, we shall fix $[b],[b^{\prime}]\in B_T$. Then by composition, we have
\begin{equation}\label{eq.137}
\left(r\circ i\right)\left([b],[b^{\prime}]\right)=\top\iff\exists c\in B: i\left([b],c\right)=\top\land r\left(c,[b^{\prime}]\right)=\top. 
\end{equation}
Unpacking $r$ and $i$ provides us the following equivalent expression
\begin{equation}\label{eq.138}
\exists c\in B: T\left(b,c\right)=\top\land T\left(c,b^{\prime}\right)=\top.
\end{equation}

Again, we can see that this corresponds to the definition of $\left(T\circ T\right)\left(b,b^{\prime}\right)=\top$. Then, from the idempotency of T, it is equivalent to state $T\left(b,b^{\prime}\right)=\top$, i.e. $b\preceq_T b^{\prime}$. So
\begin{equation}\label{eq.139}
[b]\leq_T[b^{\prime}].
\end{equation}
Applying the definition of identity in our bicategory $\operatorname{OEnt}$, we have
\begin{equation}\label{eq.140}
\operatorname{Id}_{B_T}\left([b],[b^{\prime}]\right)=\top\iff [b]\leq_T [b^{\prime}].
\end{equation}
This directly implies that $\left(r\circ i\right)\left([b],[b^{\prime}]\right)=\operatorname{Id}_{B_T}\left([b],[b^{\prime}]\right)$ for all $[b],[b^{\prime}]$. Hence, we obtain the desired equality
\begin{equation}\label{eq.141}
r\circ i=\operatorname{Id}_{B_T}.
\end{equation}

The next step can be considered quite trivial. We will prove the adjunction $r\dashv i$. In a locally posetal bicategory, the adjunction $r\dashv i$ is equivalent to the following inequalities
\begin{equation}\label{eq.142}
\operatorname{Id_B}\leq i\circ r,\hspace{0.2cm} r\circ i\leq\operatorname{Id}_{B_T}.
\end{equation}
Since $i\circ r=T$ and $\operatorname{Id}_B\leq T$ holds by assumption, we obtain $\operatorname{Id}_B\leq i\circ r$. Moreover, since $r\circ i=\operatorname{Id}_{B_T}$, we get $r\circ i\leq\operatorname{Id}_{B_T}$ (via reflexivity). Therefore, we obtain the adjunction $r\dashv i$.

Within the next crucial step, we have to show that $\left(B_T,i\right)$ is the anticipated Eilenberg-Moore object for the monad $T$. First of all, we shall give $i$ its canonical $T$-algebra structure. 

We need a $2$-cell $\alpha:T\circ i\implies i$, i.e. $T\circ i\leq i$. Well, since $T=i\circ r$, we have
\begin{equation}\label{eq.143}
T\circ i=\left(i\circ r\right)\circ i=i\circ\left(r\circ i\right)=i\circ\operatorname{Id}_{B_T}=i,
\end{equation}
where the second and third equalities follow from the associativity, and unit laws applied to $r$ and $i$. Therefore, $T\circ i=i$, so $T\circ i\leq i.$ This supplies the required $T$-algebra structure for $i$.

Next, we define the comparison maps between hom-posets and algebras. Hence, fix any object $X$ of $\operatorname{OEnt}$. Let us define
\begin{equation}\label{eq.144}
F_X:\operatorname{OEnt}\left(X,B_T\right)\longrightarrow\operatorname{Alg}_T\left(X\right),
\end{equation}\label{eq.145}
as the following composition
\begin{equation}
F_X\left(k\right):=i\circ k. 
\end{equation}
Here we must check that $i\circ k\in\operatorname{Alg}_T\left(X\right)$, i.e. that $T\circ\left(i\circ k\right)\leq i\circ k$. 
Well, it is possible to write
\begin{equation}\label{eq.146}
T\circ\left(i\circ k\right)=\left(i\circ r\right)\circ i\circ k=i\circ\left(r\circ i\right)\circ k=i\circ\operatorname{Id}_{B_T}\circ k=i\circ k.
\end{equation}
Thus $T\circ\left(i\circ k\right)=i\circ k\leq i\circ k$, so we have $F_X\left(k\right)$ as a $T$-algebra.

Additionally, $F_X$ is monotone, because if we presume that $k\leq k^{\prime}$, then $i\circ k\leq i\circ k^{\prime}$ via the monotonicity of composition.

Next, we shall define the map
\begin{equation}\label{eq.147}
G_X:\operatorname{Alg}_T\left(X\right)\longrightarrow\operatorname{OEnt}\left(X,B_T\right),
\end{equation}
as
\begin{equation}\label{eq.148}
G_X\left(x\right):=r\circ x.
\end{equation}

This is automatically considered monotone due to the monotonicity of the composition operation. 
At this point, it only remains to work efficiently with the compositions $F_X\circ G_X$ and $G_X\circ F_X$, in order to obtain the universality isomorphism between $\operatorname{OEnt}$ and $\operatorname{Alg}_T\left(X\right)$.

At first, let us take $x\in\operatorname{Alg}_T\left(X\right)$. Then $T\circ x\leq x$. Because $\operatorname{Id}_B\leq T$, then whiskering on the right-hand side of the above inequality by x yields
\begin{equation}\label{eq.149}
x=\operatorname{Id}_B \circ x\leq T\circ x.
\end{equation}
This gives us $x\leq T\circ x\leq x$. Therefore, by antisymmetry in the hom-poset $\operatorname{OEnt}\left(X,B\right)$, we obtain
\begin{equation}\label{eq.150}
x=T\circ x. 
\end{equation}

Now we can compute the required composition
\begin{equation}\label{eq.151}
\left(F_X\circ G_X\right)\left(x\right)=F_X\left(r\circ x\right)=i\circ\left(r\circ x\right)=\left
(i\circ r\right)\circ x=T\circ x=x,
\end{equation}
where the last equality comes from~(\ref{eq.150}). Therefore, $F_X\circ G_X=\operatorname{Id}.$

Focusing on the opposite composition, we take a profunctor $k:X\nrightarrow B_T$ and compute
\begin{equation}\label{eq.152}
\left(G_X\circ F_X\right)\left(k\right)=G_X\left(i\circ k\right)=r\circ\left(i\circ k\right)=\left(r\circ i\right)\circ k=\operatorname{Id}_{B_T}\circ k=k, 
\end{equation}
using the fact that $r\circ i=\operatorname{Id}_{B_T}$ and the unit law. Hence, we obtain the required composition $G_X\circ F_X=\operatorname{Id}$.  
Finally, we can summarize all of the above-mentioned considerations and conclude that since $F_X$ and $G_X$ are mutually inverse, they define an isomorphism of posets
\begin{equation}\label{eq.153}
\operatorname{OEnt}\left(X,B_T\right)\cong\operatorname{Alg}_T\left(X\right),
\end{equation}
gfor every object $X$ of $\operatorname{OEnt}$.
This shows us that we have constructed a family of isomorphisms for each \(X\), but not the full representing universal property. Hence, it remains to prove the naturality of the isomorphism in \(X\). 

Let \(h:Y\longrightarrow X\) be any \(1\)-cell. Then for \(k:X\longrightarrow B_T\), we have the composition
\begin{equation}
F_Y\left(k\circ h\right)=i\circ\left(k\circ h\right)=\left(i\circ k\right)\circ h=F_X\left(k\right)\circ h.
\end{equation}
Similarly, let \(x\in\operatorname{Alg}_T\left(X\right)\), so that \(T\circ x\leq x\). For any \(1\)-cell \(h:Y\longrightarrow X\), the monotonicity of whiskering yields
\begin{equation}
T\circ\left(x\circ h\right)=\left(T\circ x\right)\circ h\leq x\circ h.
\end{equation}
Hence, \(x\circ h\in\operatorname{Alg}_T\left(Y\right)\), and therefore
\begin{equation}
G_Y\left(x\circ h\right)=r\circ\left(x\circ h\right)=\left(r\circ x\right)\circ h=G_X\left(x\right)\circ h.
\end{equation}
Therefore, the family \(F_X\) is natural in \(X\), with the inverse family \(G_X\).

This exactly corresponds to the desired Eilenberg-Moore universal property in our locally posetal bicategory $\operatorname{OEnt}$, thereby proving that $\left(B_T,i,\alpha\right)$ is an Eilenberg-Moore object for the monad $T$.

This completes the proof of Theorem~\ref{thm.6}.
\end{proof}

\vspace{0.2cm}

It is important to mention that the object \(B_T\) is not obtained by selecting a subset of individually stable boundary states. Rather, it is the poset reflection of the preorder generated by the round-trip relation \(T\). Therefore, it precisely records the distinctions between boundary states that remain visible after saturation under the bulk-mediated interface. More explicitly, two boundary states \(b,b^{\prime}\in B\) determine the same element of \(B_T\) exactly when
\begin{equation}
T\left(b,b^{\prime}\right)=\top\land T\left(b^{\prime},b\right)=\top.
\end{equation}
Equivalently, 
\begin{equation}\label{eq.164}
[b]=[b^{\prime}]\iff b\sim_T b^{\prime}.
\end{equation}
Hence, we should interpret the object \(B_T\) as the universal quotient of boundary information by round-trip indistinguishability.

Moreover, Theorem~\ref{thm.6} does not claim that the originally given bulk object \(D\) is equivalent to \(B_T\). Rather, it uses the boundary data and the round-trip interface \(T=\Psi\circ\Phi\) to construct the part of the theory that is visible through the chosen bulk-boundary coupling. This implies that \(B_T\) can be considered the bulk-visible boundary sector, not necessarily the full microscopic or thermodynamic bulk.  

Now, by recalling the results of Theorem~\ref{thm.4} in relation to Theorem~\ref{thm.6}, one can expand their joint physical interpretation. Considering monotone maps \(F:B\longrightarrow D\), \(G:D\longrightarrow B\), then \(F\) can actually be understood as a deterministic implementation map, i.e. it assigns to each logical state \(b\) a canonical physical realization \(F\left(b\right)\). On the other hand, the map \(G\) corresponds to a deterministic abstraction map - it sends a physical macrostate \(d\) to the logical description \(G\left(d\right)\) that can be read off from it. 

Hence, Theorem~\ref{thm.4} shows that the logical-physical Galois connection \(F\left(b\right)\leq_D d\iff b\leq_B G\left(d\right)\) is precisely the demonstration of two-way testing of feasibility coincidence. Let us be more specific here. For a pair \(\left(b,d\right)\), the truth value \(F_*\left(b,d\right)=\top\) equivalently means that \(F\left(b\right)\leq_D d\). Under the entropy order, the bulk state \(d\) is at least as entropic, or at least as coarse, as the canonical implementation \(F\left(b\right)\). Thus, \(d\) is a bulk state that is thermodynamically sufficient to realize \(b\). It may be noisier, less sharply controlled, or more dissipative than the canonical realization, but nonetheless, it still lies above it in the allowed thermodynamic order. 

By contrast,
\begin{equation}
G^*\left(b,d\right)=\top\iff b\leq_B G\left(d\right).
\end{equation}
This basically says that the logical state \(b\) is no more informative than what the bulk state \(d\) can support when abstracted back to the boundary. Equivalently, the inequality \(b\leq_B G\left(d\right)\) means that the abstraction \(G\left(d\right)\) of the bulk state \(d\) contains at least the boundary information required by \(b\). Equivalently, \(b\) does not demand more logical refinement than can be recovered from \(d\) through the abstraction \(G\). 

Finally, we can summarize the results of Theorem~\ref{thm.4} as the following two-sided Landauer bridge composed of two equivalences:
\begin{tcolorbox}
\begin{center}
\textit{A bulk state \(d\) is feasible for realizing a boundary state \(b\) \(\iff\) \(d\) lies above the canonical implementation \(F\left(b\right)\) \(\iff\) \(b\) lies below the logical abstraction \(G\left(d\right)\).}
\end{center}    
\end{tcolorbox}

At this stage, it is feasible to clearly intertwine the physical interpretation of both Theorems~\ref{thm.4} and~\ref{thm.6}. As we have seen above, Theorem~\ref{thm.4} gives us the deterministic boundary-bulk dictionary, i.e. it describes how a boundary state and a bulk state match. By contrast, Theorem~\ref{thm.6} is about which boundary content survives bulk mediation and therefore becomes physically meaningful as a reconstructed visible sector of the boundary. At first glance, both theorems seem to differ in their nature. Let us then break down the physical meaning of Theorem~\ref{thm.6}.

The Eilenberg-Moore construction points out that if the Landauer adjunction induces the boundary monad \(T=\Psi\circ\Phi\), then Theorem~\ref{thm.6} constructs the Eilenberg-Moore object \(B_T\) and an adjunction \(r\dashv i\), such that \(i\circ r=T\). Specifically, the object \(B_T\) can be identified as the closed/stable sector of boundary data under the bulk-boundary round trip. Technically, one can say that this sector corresponds to the visible bulk, identified up to equivalence with the closed sector of boundary data (stable under the round trip). 

Physically, the composition \(T=\Psi\circ\Phi\) means:
\begin{enumerate}
    \item Start with a boundary state
    \item Send it through the bulk implementation
    \item Bring it back to the boundary
\end{enumerate}

The inequality \(\operatorname{Id}_B\leq T\) means that every comparison already present in the native boundary order is also present in the bulk-mediated round-trip relation. Since \(T\nrightarrow B\) is in general a profunctor rather than an ordinary endomap \(B\longrightarrow B\), this inequality should not by itself be interpreted as a pointwise sharpening or coarsening  of boundary states. Instead, the relevant stability information is contained in the preorder induced by \(T\), i.e. two boundary descriptions become indistinguishable under the round trip precisely when they are mutually \(T\)-related. Therefore, the quotient \(B_T\) retains exactly the distinctions that remain visible after saturation under bulk-mediated implementation and readout. In other words, \(B_T\) represents the boundary data that is stable under this process. We can distinguish such stability/instability in the following way:
\begin{itemize}
    \item In a general profunctorial setting, stability is not defined by selecting individual fixed boundary states. Instead, the object \(B_T\) identifies boundary descriptions that are mutually related by the round-trip preorder, i.e. states \(b,b^{\prime}\in B\) satisfying \(b\sim_T b^{\prime}\), 
    \item Only in the special case where the round-trip structure is represented by an ordinary idempotent endomap \(t:B\longrightarrow B\), does a literal fixed-point notion become available. In this case, a boundary state is fixed precisely when \(t\left(b\right)=b\).
\end{itemize}

This is similarly reflected through the maps \(i\) and \(r\) defined in~(\ref{eq.123})-(\ref{eq.125}):
\begin{itemize}
    \item \(i:B_T\nrightarrow B\) is the canonical interface from the stable quotient back to boundary descriptions,
    \item \(r:B\nrightarrow B_T\) is the quotient/reconstruction interface that sends a boundary description to the \(T\)-classes reachable from it. 
\end{itemize}
The factorization \(T=i\circ r\) shows that every bulk-mediated boundary comparison factors through the stable quotient \(B_T\). Thus, \(B_T\) precisely captures the distinctions retained by the idempotent round-trip relation. It does not merely tell us whether one \(b\) matches one \(d\), it tells us which boundary content is invariant under the whole bulk-mediated process. All of the above considerations about Theorem~\ref{thm.6} can be summarized in the following way:

\begin{tcolorbox}
\begin{center}
Theorem~\ref{thm.6} represents the global stability/reconstruction law \(\iff\) it identifies which boundary descriptions survive the entire boundary–bulk–boundary round trip. 
\end{center}
\end{tcolorbox}

Analogously, we can summarize the physical interpretation of Theorem~\ref{thm.4} as:
\begin{tcolorbox}
\begin{center}
Theorem~\ref{thm.4} is the local matching law \(\iff\) in the representable case it says when a boundary state \(b\) and a bulk state \(d\) are compatible.
\end{center}
\end{tcolorbox}
This basically tells us when one boundary description \(b\) and one bulk macrostate \(d\) are physically compatible.

Notice that the locality/globality of both theorems is itself the main reflection of their differences. Theorem~\ref{thm.4} is local in the sense that it is a statement about one pair \(b,d\) at a time. Basically, it answers when this bulk state realizes such a boundary state and when this physical macrostate supports such a logical description. On the contrary, Theorem~\ref{thm.6} is global, i.e. it is about the whole closure process induced by the interface pair. It answers which boundary descriptions are stable under the implementation round trip and what is the canonical stable sector that is determined by the given physics.  

We shall also add the following corollary, which stems directly from Theorem~\ref{thm.6}.

\newtheorem{cor}{Corollary}
\begin{cor}
Let \(\Phi\dashv\Psi\) be an exact Landauer adjunction. By the representability-rigidity result (Remark~\ref{rem.1}), there exists a unique monotone map \(F:B\longrightarrow D\) with \(\Phi=F_*\), \(\Psi=F^*\). Hence, \(T=F^*\circ F_*\) and
\begin{equation}\label{eq.172}
T\left(b,b^{\prime}\right)=\top\iff F\left(b\right)\leq_D F\left(b^{\prime}\right).
\end{equation}
Therefore,
\begin{equation}
b\sim_T b^{\prime}\iff F\left(b\right)=F\left(b^{\prime}\right),
\end{equation}
and the stable quotient is canonically isomorphic to the image poset
\begin{equation}
B_T\cong\operatorname{Im}\left
(F\right).
\end{equation}
\end{cor}
This turns the abstract quotient into a concrete, physically interpretable  object. 

In the exact, representable case, this provides a particularly precise interpretation of irreversibility in terms of the loss of distinguishability. Indeed, by~(\ref{eq.121}) and~(\ref{eq.172}) and using the definition of \(\sim_T\) , we have
\begin{equation}
b\sim_T b^{\prime}\iff F\left(b\right)\leq_D F\left(b^{\prime}\right)\land F\left(b^{\prime}\right)\leq_D F\left(b\right).  
\end{equation}
Since \(D\) is a poset, antisymmetry yields
\begin{equation}
b\sim_T b^{\prime}\implies F\left(b\right)=F\left(b^{\prime}\right).
\end{equation}
Hence, distinct boundary descriptions may become indistinguishable under the bulk-mediated implementation whenever they determine the same bulk realization. Therefore, the Landauer round trip need not be interpreted as a deterministic motion of a boundary state toward a coarser state. Rather, the implementation may fail to preserve distinctions originally present on the boundary, and the quotient \(B_T\) records exactly the distinctions that remain visible through the bulk-mediated coupling. 

Before we proceed to the final section of the presented paper, we shall shift our focus to actively implementing this bicategorical framework to finite examples. In particular, we will deal with two toy examples: first, a non-representable open entropy interface, and second, a finite Landauer adjunction combined with a closure endo-interface with a stable sector. 

\vspace{0.2cm}
\subsection{\textbf{Finite toy examples}}

\vspace{0.2cm}
In the following finite examples, we carefully distinguish between two notions. First, a closure interface is an endo-profunctor \(T:B\nrightarrow B\) satisfying \(\operatorname{Id}_B\leq T\) and \(T\circ T\leq T\). Second, an ordinary closure map is a monotone function \(t:B\longrightarrow B\) satisfying \(b\leq t\left(b\right)\) and \(t\left(t\left(b\right)\right)=t\left(b\right)\). These are not the same pieces of data. A closure map may induce profunctors in different ways, and one must specify the chosen representation. Therefore, in the examples below, whenever we use a map \(t\), we treat it as an ordinary order-theoretic closure map. On the contrary, whenever we use \(T\), we compute it directly as a relation \(T\left(b,b^{\prime}\right)\) by profunctor composition.

\vspace{0.3cm}
\hspace{-0.7cm}\textbf{\underline{Ex.1 - Non-representable open entropy interfaces}}

\vspace{0.1cm}
In the first example, we demonstrate why arbitrary profunctors/interfaces are practically useful within our bicategorical framework. It is not yet an adjunction, but its purpose is to show that a boundary state may have a feasible bulk realization set that cannot be generated by a single canonical monotone implementation map. 

Let the boundary entropy poset be the two-element chain
\begin{equation}
B=\{b_0<b_1\}.
\end{equation}

This can be interpreted in the following way:
\begin{itemize}
    \item \(b_0=\) coarser logical description
    \item \(b_1=\) sharper logical description
\end{itemize}
Next, let the bulk entropy poset be given by
\begin{equation}
D=\{0,a,c,1\},
\end{equation}
with order relations \(0<a<1\), \(0<c<1\) and \(a,c\) being incomparable. Diagrammatically, we can draw the following diagram

\begin{figure}[h]
\centering
\begin{tikzpicture}[x=1cm, y=1.2cm, baseline=(current bounding box.center)]
  
  \node (one)  at (0, 2)   {$1$};
  \node (a)    at (-0.4, 1) {$a$};
  \node (c)    at (0.4, 1)  {$c$};
  \node (zero) at (0, 0)   {$0$};

  \draw (zero) -- (a) -- (one);
  \draw (zero) -- (c) -- (one);
\end{tikzpicture}
\caption{The diamond-shaped bulk entropy poset}
\label{fig:bulk-diamond}
\end{figure}
where \(0\) represents a highly controlled low-entropy bulk state, \(a\) and \(c\) are two distinct thermodynamic implementation regimes, \(1\) is a coarse high-entropy bulk state, and \(a,c\) are incomparable because neither implementation regime is thermodynamically reducible to the other. 

Now, we shall define the interface. Let \(\Phi:B\nrightarrow D\) be given by the following fibers
\begin{equation}
\Phi\left(b_0,-\right)=\{a,c,1\}\quad \Phi\left
(b_1,-\right)=\{1\}.
\end{equation}
Equivalently, we can depict the fibers via the truth table
\[
\begin{array}{c|cccc}
\Phi & 0 & a & c & 1 \\
\hline
b_0  & \bot & \top & \top & \top \\
b_1  & \bot & \bot & \bot & \top
\end{array}
\]

Physically, the table tells us that the coarser boundary state \(b_0\) can be realized either by the bulk regime \(a\) and \(c\), or by the coarser regime \(1\). Considering the sharper boundary state \(b_1\), it is feasible only at the maximally coarse bulk state \(1\).

The fact that \(\Phi\) is formally an open interface can be straightforwardly checked through the bulk covariance and the boundary contravariance. For each fixed \(b\in B\), the evaluation \(\Phi\left(b,-\right)\) must be upward closed in \(D\). For \(b_0\), we have \(\Phi\left(b_0,-\right)=\{a,c,1\}\). If \(a\in\Phi\left(b_0,-\right)\), then \(a\leq 1\), and \(1\in\Phi\left(b_0,1\right)\). Also, if \(c\in\Phi\left(b_0,-\right)\), then \(c\leq 1\), and \(1\in\Phi\left(b_0,-\right)\). Therefore, the element \(1\) has no strictly larger element, so \(\Phi\left(b_0,-\right)\) is upward closed for \(b_0\). For \(b_1\), we have \(\Phi\left(b_1,-\right)=\{1\}\). This is upward closed because \(1\) is maximal. In conclusion, the bulk covariance holds.

Let us continue with the boundary contravariance. Since \(b_0\leq b_1\), we need 
\begin{equation}
\Phi\left(b_1,d\right)=\top\implies\Phi\left(b_0,d\right)=\top.
\end{equation}
The only \(d\) such that \(\Phi\left(b_1,d\right)=\top\) is \(d=1\). Since \(\Phi\left(b_0,1\right)=\top\), the condition automatically holds. This implies that \(\Phi\) is a valid Boolean profunctor/open entropy interface.

At this point, we want to show that the open interface \(\Phi\) is actually non-representable. 

Using the opposite perspective, a representable companion interface takes the form
\begin{equation}
F_*\left(b,d\right)=\top\iff F\left(b\right)\leq_D d,
\end{equation}
for some monotone map \(F:B\longrightarrow D\). Hence, for every \(b\in B\), the fiber 
\begin{equation}
F_*\left(b,-\right)=\{d\in D\hspace{0.1cm}\vert\hspace{0.1cm}F\left(b\right)\leq d\},
\end{equation}
must be a principal upset, i.e. \(F_*\left(b,-\right)=\hspace{0.1cm}\uparrow F\left(b\right)\). Thus, if \(\Phi\) were representable, then the fiber \(\Phi\left(b_0,-\right)=\{a,c,1\}\) would have to be principal. However, the principal upsets in \(D\) are
\begin{equation}
\uparrow 0=\{0,a,c,1\},\quad \uparrow a=\{a,1\}, \quad \uparrow c=\{c,1\}, \quad \uparrow 1=\{1\}.
\end{equation}
Obviously, none of these equals \(\{a,c,1\}\). Therefore, no element \(d_0\in D\) satisfies \(\uparrow d_0=\{a,c,1\}\). Hence, there is no monotone map \(F:B\longrightarrow D\) such that \(\Phi=F_*\). This directly implies that \(\Phi\) is non-representable. 

From a physical point of view, this can be considered the simplest finite example that demonstrates the practicality of profunctorial openness. The boundary state \(b_0\) has two incomparable minimal feasible bulk realizations \(a\) and \(c\). There is not a single canonical implementation \(F\left(b_0\right)\), because choosing \(a\) loses the feasible branch through \(c\), and choosing \(c\) loses the feasible branch through \(a\). Moreover, choosing \(0\) would include too much, because \(\uparrow 0=\{0,a,c,1\}\), while \(0\) is not feasible. 

Therefore, the interface cannot be compressed into a deterministic monotone implementation map. This is precisely the phenomenon we desire to capture when we say that logical boundary states and thermodynamic bulk states should be related not by a single map but by a feasibility interface describing all possible (feasible) realizations. Additionally, recalling Remark~\ref{rem.1}, this implies that \(\Phi\) cannot occur as the left adjoint of an exact Landauer adjunction in \(\operatorname{OEnt}\). In summary, the role of this example was to demonstrate that the ambient bicategory \(\operatorname{OEnt}\) contains genuinely open feasibility interfaces beyond the exact-adjoint subclass.   

\vspace{0.3cm}
\hspace{-0.7cm}\textbf{\underline{Ex.2 - Finite Landauer adjunction and its boundary-stable reconstruction}}

\vspace{0.1cm}
Here, we provide a complete finite example of the Landauer adjunction \(\Phi\dashv\Psi\) and compute the boundary closure \(T=\Psi\circ\Phi\) and its boundary-stable sector. We shall clarify that this example is representable. By the representability-rigidity result (Remark~\ref{rem.1}), every exact Landauer adjunction in \(\operatorname{OEnt}\) is representable. Accordingly, this example is constructed from a monotone implementation map \(F:B\longrightarrow D\) and its companion-conjoint adjunction. Nonetheless, it is very useful for us because it appropriately demonstrates the reconstruction theorem (Theorem~\ref{thm.6}) in the finite case. 

Let 
\begin{equation}
B=\{b_0<b_1<b_2\},
\end{equation}
be a three-element boundary chain, such that
\begin{itemize}
    \item \(b_0=\) the coarsest, least informative boundary description,
    \item \(b_1=\) an intermediate, partially refined boundary description,
    \item \(b_2=\) the sharpest, most informative boundary description.
\end{itemize}

Next, let 
\begin{equation}
D=\{d_0<d_1\},
\end{equation}
be a two-element bulk chain, such that
\begin{itemize}
    \item \(d_0=\) controlled low-entropy bulk regime,
    \item \(d_1=\) coarser high-entropy bulk regime.
\end{itemize}

Now, we can define the deterministic maps \(F\) and \(G\). Let
\begin{equation}
F:B\longrightarrow D,
\end{equation}
be given by
\begin{equation}
F\left(b_0\right)=d_0,\quad F\left(b_1\right)=d_1, \quad F\left(b_2\right)=d_1.
\end{equation}
Next, we define
\begin{equation}
G:D\longrightarrow B,
\end{equation}
as
\begin{equation}
G\left(d_0\right)=b_0,\quad G\left(d_1\right)=b_2.
\end{equation}

At first, we want to check the monotonicity. Since \(b_0<b_1<b_2\), we have \(F\left(b_0\right)=d_0\leq d_1=F\left(b_1\right)\) and \(F\left(b_1\right)=d_1\leq d_1=F\left(b_2\right)\). So \(F\) is monotone. Similarly for the map \(G\).

Let us continue with checking the Galois connection \(F\dashv G\). We claim that
\begin{equation}
F\left(b\right)\leq_D d\iff b\leq_B G\left(d\right).
\end{equation}
This is verifiable case by case. 

First, taking \(d=d_0\), we have \(G\left(d_0\right)=b_0\). Then 
\begin{equation}
b\leq G\left(d_0\right)\iff b\leq b_0\iff b=b_0.
\end{equation}
On the other hand,
\begin{equation}
F\left(b\right)\leq d_0,
\end{equation}
holds exactly when \(F\left(b\right)=d_0\), which occurs precisely for \(b=b_0\). Therefore,
\begin{equation}
F\left(b\right)\leq d_0\iff b\leq G\left(d_0\right).
\end{equation}

In the case when \(d=d_1\), we get \(G\left(d_1\right)=b_2\). Then for every \(b\in B\), we have \(b\leq b_2\). On the other side, since \(d_1\) is top in \(D\), for every \(b\in B\), we obtain
\begin{equation}
F\left(b\right)\leq d_1.
\end{equation}
Therefore, 
\begin{equation}
F\left(b\right)\leq d_1\iff b\leq G\left(d_1\right),
\end{equation}
for all \(b\in B\). We have proved \(F\dashv G\).

Now we focus on the associated interfaces. Let us define the companion interface \(\Phi=F_*:B\nrightarrow D\) by
\begin{equation}
\Phi\left(b,d\right)=\top\iff F\left(b\right)\leq d.
\end{equation}
The corresponding table is
\[
\begin{array}{c|cc}
\Phi & d_0 & d_1 \\
\hline
b_0  & \top & \top \\
b_1  & \bot & \top \\
b_2  & \bot & \top
\end{array}
\]

Next, let us define the conjoint/right interface \(\Psi\) of \(F\),
\begin{equation}
\Psi:=F^*:D\nrightarrow B,
\end{equation}
by 
\begin{equation}
\Psi\left(d,b\right)=\top\iff d\leq_D F\left(b\right).
\end{equation}
Using the companion-conjoint adjunction, we get
\begin{equation}
F_*\dashv F^*,
\end{equation}
so \(\Phi:=F_*\) and \(\Psi:=F^*\) form an adjunction in the bicategory \(\operatorname{OEnt}\).

Since \(F\dashv G\), Theorem~\ref{thm.4} additionally identifies the forward interface \(F_*\) with the conjoint \(G^*\) as
\begin{equation}
F_*=G^*:B\nrightarrow D.
\end{equation}
This equality concerns the forward interface and should not be confused with the right adjoint \(F^*:D\nrightarrow B\).

In conclusion, the right adjoint of \(\Phi=F_*\) in \(\operatorname{OEnt}\) is \(\Psi:D\nrightarrow B\) defined by \(\Psi\left(d,b\right)=\top\iff d\leq _DF\left(b\right)\), so \(F\left(b_0\right)=d_0\) and \(F\left(b_1\right)=F\left(b_2\right)=d_1\). We obtain the following table
\[
\begin{array}{c|ccc}
\Psi & b_0 & b_1 & b_2 \\
\hline
d_0  & \top & \top  & \top \\
d_1  & \bot & \top  & \top\\
\end{array}
\]
This corresponds to the Landauer adjunction in the bicategory \(\operatorname{OEnt}\), recovering the deterministic Galois connection as a representable special case, exactly as we have demonstrated in Theorem~\ref{thm.4}.

As we have stated earlier, a Landauer adjunction induces a boundary closure monad/endo-interface on \(B\), satisfying \(\operatorname{Id}_B\leq T\), \(T\circ T\leq T\). Let us now verify these relations directly. At first, we focus on the inflationarity condition. The identity relation corresponds to
\begin{equation}
\operatorname{Id}_B\left(b,b^{\prime}\right)=\top\iff b\leq b^{\prime}.
\end{equation}
Hence, its table is 
\[
\begin{array}{c|ccc}
\operatorname{Id}_B & b_0 & b_1 & b_2 \\
\hline
b_0  & \top & \top  & \top \\
b_1  & \bot & \top  & \top\\
b_2  & \bot & \bot  & \top\\
\end{array}
\]
This table simply records the order relation of the chain \(B\).

Now, we define the closure endo-interface \(T:B\nrightarrow B\) by the following table
\[
\begin{array}{c|ccc}
T:=\Psi\circ\Phi & b_0 & b_1 & b_2 \\
\hline
b_0  & \top & \top  & \top \\
b_1  & \bot & \top  & \top\\
b_2  & \bot & \top  & \top\\
\end{array}
\]
Equivalently, 
\begin{equation}\label{eq.178}
T\left(b_0,-\right)=\{b_0,b_1,b_2\},\quad T\left(b_1,-\right)=\{b_1,b_2\},\quad T\left(b_2,-\right)=\{b_1,b_2\}.
\end{equation}
One can see that from \(b_0\), the interface relates to three boundary states, so it is meaningless to simply write \(T\left(b_0\right)\) unless some representability assumption is added. 

Now we need to check whether \(T\) is a valid interface. That means checking if \(T:B^{op}\times B\longrightarrow\operatorname{Bool}\) is contravariant in the first variable and covariant in the second variable. At first, we focus on the covariance in the target variable. For each fixed argument \(b\), the set \(T\left(b,-\right)=\{b^{\prime}\in B\hspace{0.1cm}\vert\hspace{0.1cm}T\left(b,b^{\prime}\right)=\top \}\) must be upward closed in \(B\). Indeed, this is evident because the sets \(\{b_0,b_1,b_2\}\) and \(\{b_1,b_2\}\) are upward closed subsets of \(B\). Hence, each row \(T\left(b,-\right)\) is upward closed. 

Contravariance in the first variable says that for all \(b\leq b^{\prime}\), we have \(T\left(b^{\prime},c\right)=\top\implies T\left(b,c\right)=\top\). Equivalently, whenever \(b\leq b^{\prime}\), we must have \(T\left(b^{\prime},-\right)\subseteq T\left(b,-\right)\). Now we check this directly. Since \(b_0\leq b_1\), we need \(T\left(b_1,-\right)\subseteq T\left(b_0,-\right)\). Indeed, we have \(T\left(b_1,-\right)=\{b_1,b_2\}\) and \(T\left(b_0,-\right)=\{b_0, b_1, b_2\}\), so the inclusion is satisfied. Identically, one can check that the inclusions \(T\left(b_2,-\right)\subseteq T\left(b_1,-\right)\) and \(T\left(b_2,-\right)\subseteq T\left(b_0,-\right)\) are also satisfied. Thus \(T\) is a valid endo-interface.

For checking the unit inequality \(\operatorname{Id}_B\leq T\), we directly compare the tables for \(\operatorname{Id}_B\) and \(T\). We can see that every \(\top\)-entry of \(\operatorname{Id}_B\) is also a \(\top\)-entry of \(T\). Therefore, \(\operatorname{Id}_B\leq T\). Now we can elaborate more on the composition \(T\circ T:B\nrightarrow B\). Via the profunctor composition, we have 
\begin{equation}
\left(T\circ T\right)\left(b,b^{\prime}\right)=\top\iff\exists c\in B,\hspace{0.1cm} T\left(b,c\right)=\top\hspace{0.1cm}and\hspace{0.1cm}T\left(c,b^{\prime}\right)=\top.
\end{equation}
In other words,
\begin{equation}
\left(T\circ T\right)\left(b,b^{\prime}\right)=\bigvee_{c\in B}\left(T\left(b,c\right)\land T\left(c,b^{\prime}\right)\right).
\end{equation}
Let us compute each row individually. For the first row \(b=b_0\), we have \(T\left(b_0,-\right)=\{b_0,b_1,b_2\}\). So the intermediate state \(c\) may be any of \(b_0,b_1,b_2\). If we choose \(c=b_0\), then \(T\left(b_0,b^{\prime}\right)=\top\) for every \(b^{\prime}\in\{b_0,b_1,b_2\}\). Therefore, \(\left(T\circ T\right)\left(b_0,b^{\prime}\right)=\top\) for every \(b^{\prime}\in B\), which implies \(\left(T\circ T\right)\left(b_0,-\right)=\{b_0,b_1,b_2\}=T\left(b_0,-\right)\). In a similar fashion, one can proceed with the rows \(b=b_1\) and \(b=b_2\). Eventually, we would obtain \(\left(T\circ T\right)\left(b_1,-\right)=\{b_1,b_2\}=T\left(b_1,-\right)\) and \(\left(T\circ T\right)\left(b_2,-\right)=\{b_1,b_2\}=T\left(b_2,-\right)\). Putting all three rows together, we obtain the composite table
\[
\begin{array}{c|ccc}
T\circ T & b_0 & b_1 & b_2 \\
\hline
b_0  & \top & \top  & \top \\
b_1  & \bot & \top  & \top\\
b_2  & \bot & \top  & \top\\
\end{array}
\]
Hence \(T\circ T=T\). In particular, \(T\circ T\leq T\). Thus \(T\) is an idempotent closure interface. The fact that \(\preceq_T\) is a preorder is quite evident from the table above. We have
\begin{equation}
b_0\preceq_T b_0,\quad b_0\preceq_T b_1,\quad b_0\preceq_T b_2,\quad b_1\preceq_T b_1,\quad
b_1\preceq_T b_2,\quad
b_2\preceq_T b_2,
\end{equation}
and also \(b_2\preceq_T b_1\), \(b_2\preceq_T b_2\).
So the only new relation beyond the original chain order is \(b_2\preceq_T b_1\). This actually means that the round-trip interface identifies \(b_1\) and \(b_2\) up to \(T\)-equivalence. Because \(\operatorname{Id}_B\leq T\), the relation \(\preceq_T\) is reflexive and contains the original order of \(B\). Additionally, since \(T\circ T=T\), the relation \(\preceq_T\) is also transitive.

Now, we focus on the \(T\)-equivalence relation by computing the equivalence classes. First, \(b_0\) is equivalent only to itself, because
\begin{equation}
T\left(b_0,b_1\right)=\top,\quad T\left(b_1,b_0\right)=\bot,
\end{equation}
Also,
\begin{equation}
T\left(b_0,b_2\right)=\top,\quad T\left(b_2,b_0\right)=\bot.
\end{equation}
Therefore, \([b_0]=\{b_0\}\). Similarly, one can check that \(b_1\) and \(b_2\) are also equivalent. Hence, \([b_1]=[b_2]=\{b_1,b_2\}\). So the quotient set corresponds to 
\begin{equation}
B_T=B/_{\sim T}=\{[b_0],[b_1]\},
\end{equation}
where \([b_1]=[b_2]\). Notice that we have two equivalence classes \(A:=[b_0]\) and \(C:=[b_1]=[b_2]\). Given the order on \(B_T\), we have \(A\leq_T C\), because \(b_0\preceq_T b_1\). Since \(b_1\npreceq_T b_0\) and \(b_2\npreceq_T b_0\), there is no reverse inequality \(C\leq_T A\). This implies that \(B_T\) is the two-element chain \(A<C\), corresponding to the boundary-stable sector
\begin{equation}
B_T=\{A<C\}.
\end{equation}
We can interpret it in a way that the interface cannot distinguish \(b_1\) and \(b_2\) after the round trip. The elements \(b_1,b_2\) become one stable boundary class, so \(B_T\) is precisely the bulk-visible boundary sector.

At the final stage of this example, we focus on establishing the reconstruction interfaces \(r\) and \(i\). In Theorem~\ref{thm.6}, we construct two interfaces \(r:B\nrightarrow B_T\) and \(i:B_T\nrightarrow B\) such that \(r\left(b,[c]\right)=\top\iff b\preceq_T c\) and \(i\left([c],b\right)=\top\iff c\preceq_T b\). Let us recall that \(B_T=\{A<C\}\), where \(A=[b_0]\) and \(C=[b_1]=[b_2]\). Working with the first interface, we have \(r\left(b_0,A\right)=\top\), because \(T\left(b_0,b_0\right)=\top\). Also, \(r\left(b_0,C\right)=\top\), because \(T\left(b_0,b_1\right)=\top\). Next, \(r\left(b_1,A\right)=\bot\), because \(T\left(b_1,b_0\right)=\bot\). However, \(r\left(b_1,C\right)=\top\), because \(T\left(b_1,b_1\right)=\top\). We can complete this procedure in a similar way with the element \(b_2\). Eventually, we obtain the table for \(r\),
\[
\begin{array}{c|cc}
r & A & C  \\
\hline
b_0  & \top & \top \\
b_1  & \bot & \top \\
b_2  & \bot & \top \\
\end{array}
\]
This is a valid interface \(r:B\nrightarrow B_T\). Indeed, each row is upward closed in \(B_T\), and the rows are contravariant in the source variable \(B\). 

Next, we continue with the interface \(i\). For \(A=[b_0]\), we choose the representative \(b_0\). Then \(T\left(b_0,-\right)=\{b_0,b_1,b_2\}\). This tells us that 
\begin{equation}
i\left(A,b_0\right)=i\left(A,b_1\right)=i\left(A,b_2\right)=\top.
\end{equation}
Similarly, for \(C=[b_1]=[b_2]\), we choose the representative \(b_1\). Then \(T\left(b_1,-\right)=\{b_1,b_2\}\). Therefore,
\begin{equation}
i\left(C,b_0\right)=\bot,\quad i\left(C,b_1\right)=i\left(C,b_2\right)=\top.
\end{equation}
So the table for the interface \(i\) is
\[
\begin{array}{c|ccc}
i & b_0 & b_1 & b_2 \\
\hline
A  & \top & \top & \top \\
C  & \bot & \top & \top
\end{array}
\]
Again, this is a valid interface \(i:B_T\nrightarrow B\). Notice that the definition is independent of the choice of the representative. For example, choosing \(b_2\) instead of \(b_1\) for the class \(C\) provides the same row, because \(T\left(b_1,-\right)=T\left(b_2,-\right)=\{b_1,b_2\}\).

Now, we want to verify that \(i\circ r=T\). This can be done individually row by row. From the table of \(r\), we know that for \(b=b_0\),
\begin{equation}
r\left(b_0,A\right)=\top,\quad r\left(b_0,C\right)=\top.
\end{equation}
Hence, both classes \(A\) and \(C\) may mediate. The row of \(i\) at \(A\) is \(\{b_0,b_1,b_2\}\). Thus, already through \(A\), we get every target \(b^{\prime}\). This means that
\begin{equation}
\left(i\circ r\right)\left(b_0,-\right)=\{b_0,b_1,b_2\}=T\left(b_0,-\right).
\end{equation}
Equivalently, we can proceed row by row to ultimately obtain \(i\circ r=T.\)

For completeness, we should also verify the adjunction \(r\dashv i\). Equivalently, we have \(\operatorname{Id}_B\leq i\circ r\) and \(r\circ i\leq\operatorname{Id}_{B_T}\). Since \(i\circ r=T\), the first inequality is evident. Also, we have already checked that \(\operatorname{Id}_B\leq T\). Thus, we shall compute \(r\circ i:B_T\nrightarrow B_T\). The composite itself satisfies \(\left(r\circ i\right)\left(X,Y\right)=\top\) if and only if there exists \(b\in B\) such that \(i\left(X,b\right)=\top\) and \(r\left(b,Y\right)=\top\), with \(X,Y\in B_T\) being arbitrary. 

Now, we compute the rows. For \(X=A\), we get \(i\left(A,-\right)=\{b_0,b_1,b_2\}\). Since \(r\left(b_0,A\right)=\top\) and \(r\left(b_0,C\right)=\top\), we have
\begin{equation}
\left(r\circ i\right)\left(A,A\right)=\top,\quad \left(r\circ i\right)\left(A,C\right)=\top.
\end{equation}
This implies that
\begin{equation}
\left(r\circ i\right)\left(A,-\right)=\{A,C\}.
\end{equation}
Next, for \(X=C\), we obtain \(i\left(C,-\right)=\{b_1,b_2\}\). For both \(b_1\) and \(b_2\), the \(r\)-row is \(\{C\}\). Therefore,
\begin{equation}
\left(r\circ i\right)\left(C,-\right)=\{C\}.
\end{equation}
Hence, the table for the composition \(r\circ i\) is
\[
\begin{array}{c|cc}
r\circ i & A & C \\
\hline
A  & \top & \top \\
C  & \bot & \top
\end{array}
\]
One can immediately see that this is exactly the identity interface on the two-element chain \(A<C\). Therefore, \(r\circ i=\operatorname{Id}_{B_T}\). In particular, \(r\circ i\leq\operatorname{Id}_{B_T}\), so we have the adjunction \(r\dashv i\). Together with the composition \(i\circ r=T\), this realizes \(B_T\) as the boundary-stable sector selected by the closure interface \(T\).

This example precisely shows that the closure interface \(T\) does not send each state to a single state. Instead, it gives a round-trip relation. Such a relation identifies \(b_1\) and \(b_2\), because \(T\left(b_1,b_2\right)=\top\) and \(T\left(b_2,b_1\right)=\top\). Thus, \(b_1\) and \(b_2\) become indistinguishable from the perspective of the round-trip interface. On the other hand, \(b_0\) remains distinct. This is why the stable sector corresponds to the quotient \(B_T=\{[b_0],[b_1]=[b_2]\}\), so the stable sector is the two-element chain \([b_0]<[b_1]=[b_2]\).

\section{Possible extension - quantitative enrichment and open problems}\label{sec.3}

In this section, we outline a broader extension of the preceding results, paradigms, and consider how the bicategorical entropy-cost framework could be developed even further. The generalization focuses on expanding the already used categorical structure in a way that would more closely correspond to the true physical nature of thermodynamic systems and their entropy. In particular, we outline a quantitative extension in which the same compositional interface semantics are enriched by dissipation costs.

A natural first step should be to replace our locally posetal bicategory $\operatorname{OEnt}$ with a certain kind of "process" bicategory and also replace $\operatorname{Bool}$ with a quantale $\left([0,\infty],+,\leq\right)$, so the hom sets carry some kind of entropy production costs. This would lead us to a bicategory of costed open interfaces, where the Landauer adjunction yields a quantitative boundary monad.

In other words, our original open interface should carry a minimal entropy (dissipation) cost, and the Landauer adjunction should then become a quantitative adjunction whose unit inequality should recover precise numerical lower bounds, rather than just order statements. 

One can summarize the punchline of the possible extension as follows:
\begin{itemize}
    \item $\operatorname{OEnt}$ can be viewed as the $\operatorname{Bool}$-enriched profunctor bicategory on thin categories (local posets)
    \item Replacing $\operatorname{Bool}$ with a quantale $\mathscr{V}$ turns the feasibility relation into a graded feasibility relation, i.e. a relation with a cost function, and the relation composition transforms into an infimal convolution
    \item A Landauer adjunction $\Phi\dashv\Psi$ becomes an enriched bulk-boundary adjunction, whose unit/counit inequalities correspond to quantitative Landauer constraints, i.e. the boundary monad $T=\Psi\circ\Phi$ would become a costed closure operator
\end{itemize}

In the following paragraphs, we will attempt to rigorously formulate our considerations and outline the exact procedure for how we could implement the above-mentioned points.

\subsection{\textbf{Quantales as the enrichment base for the costed feasibility relation}}

Let us start with the definition of a quantale.
\begin{definition}\cite{rosenthal1990quantales}
A (unital) quantale is a tuple
\begin{equation}\label{eq.156}
\left(\mathscr{V},\leq,\otimes,I,\bigvee\right),
\end{equation}
such that
\begin{enumerate}
    \item $\left(\mathscr{V},\leq\right)$ is a complete lattice 
    \item $\left(\mathscr{V},\otimes,I\right)$ is a monoid
    \item $\otimes$ distributes over joins in each variable, so
    \begin{equation}\label{eq.157}
    a\otimes\left(\bigvee_i b_i\right)=\bigvee_i\left(a\otimes b_i\right),\hspace{0.2cm}\left(\bigvee_i a_i\right)\otimes b=\bigvee_i\left(a_i\otimes b\right).
    \end{equation}
\end{enumerate}
\end{definition}

In this setup, we are able to represent the elements of $\mathscr{V}$ as weights/costs, the operation $\otimes$ as a sequential composition of costs, and $\bigvee$ as an aggregation of alternative realizations, i.e. how to choose the best composition according to $\leq$.

As an example, we can directly mention $\operatorname{Bool}=\{\perp,\leq,\top\}$ with $\otimes=\land$, $I=\top$ and joins given by ordinary suprema in the two-element lattice. This precisely corresponds to our feasibility relational setup that we have used in $\operatorname{OEnt}$.

Another very useful example with regard to the above considerations is the Lawvere quantale \cite{lawvere1973metric}. In order to effectively conduct quantitative thermodynamics, the canonical choice is the Lawvere quantale
\begin{equation}\label{eq.158}
\mathscr{V}_{cost}:=\{[0,\infty], \leq_{\mathscr{V}}, \otimes, I, \bigvee \},
\end{equation}
where \(a\leq_{\mathscr{V}}b\iff a\geq b\) numerically, \(a\otimes b:=a+b\), \(I:=0\) and \(\bigvee S:=\inf S\) for every \(S\subset [0,\infty]\).

Hence, with respect to the quantale order, we have:
\begin{itemize}
    \item  $a\geq b$ numerically, meaning that $a$ is no cheaper than $b$
    \item $\bigvee$ can be understood as the best alternative, i.e. the ordinary $\operatorname{inf}$
    \item $\otimes$ is the sequential composition cost, i.e. the operation $+$, which adds sequential costs,
    \item unit $I=0$ is the zero-cost identity
\end{itemize}

This is precisely what we want, if we desire to encode the functor $\Phi\left(.,.\right)$ as a minimum entropy-production cost. Then, among all intermediate realizations, we shall take the infimum, meaning that the composition is supposed to choose the least costly intermediate realization.

At this point, we can pass from our original locally posetal bicategory $\operatorname{OEnt}$ to a quantale-enriched version. We can straightforwardly define the $\mathscr{V}$-relations as costed interfaces.

\begin{definition}
Let $X$, $Y$ be sets (or posets, based on the order-compatibility). A $\mathscr{V}$-relation
\begin{equation}\label{eq.159}
R:X\nrightarrow Y,
\end{equation}
is a function
\begin{equation}\label{eq.160}
R:X\times Y\longrightarrow\mathscr{V}.    
\end{equation}
\end{definition}

One can think of $R\left(x,y\right)$ as a cost that is spent on relating $x$ to $y$. In $\operatorname{Bool}$, the relation $R\left(x,y\right)\in\{\perp,\top \}$ can be understood as infeasible/feasible. On the other hand, in $\mathscr{V}_{cost}$, we can encode $R\left(x,y\right)\in [0,\infty]$ as the minimal cost needed to relate \(x\) to \(y\), with $\infty$ representing the impossible/infinite cost

When it comes to the composition operation, we can formulate the following definition.
\begin{definition}
Given $R:X\nrightarrow Y$ and $S:Y\nrightarrow Z$, we can define their composition 
\begin{equation}\label{eq.161}
S\circ R:X\nrightarrow Z,
\end{equation}
via the standard quantale matrix multiplication
\begin{equation}\label{eq.162}
\left(S\circ R\right)\left(x,z\right):=\bigvee_{y\in Y}\left(R\left(x,y\right)\otimes S\left(y,z\right)\right),
\end{equation}
and the identity \(\mathscr{V}\)-relation on \(X\) as
\begin{equation}\label{eq.165}
\operatorname{Id}_X\left(x,x^{\prime}\right)=
\begin{cases}
I;\hspace{0.1cm}x=x^{\prime} \\
\perp_{\mathscr{V}}\hspace{0.1cm}x\neq x^{\prime},
\end{cases}
\end{equation}
where $\perp_{\mathscr{V}}$ denotes the least element of $\mathscr{V}$. 
\end{definition}

Explicitly, for $\textbf{Bool}$, the composition becomes
\begin{equation}\label{eq.163}
\left(S\circ R\right)\left(x,z\right)=\top\iff\exists y\in Y: R\left(x,y\right)=\top\land S\left(y,z\right)=\top,
\end{equation}
which is a standard relational composition.

In the case of $\mathscr{V}_{cost}$, we have 
\begin{equation}\label{eq.217}
\left(S\circ R\right)\left(x,z\right)=\inf_{y\in Y}\left(R\left(x,y\right)+S\left(y,z\right)\right),
\end{equation}
which corresponds to infimal convolution, i.e. it provides the cheapest intermediate realization, where the infimum of the total cost is taken over all intermediate mediators.

Considering the identity in the ordered version \(\mathscr{V}-\operatorname{OEnt}\), where objects are posets and \(1\)-cells are ordered-compatible \(\mathscr{V}\)-interfaces \(P^\times Q\longrightarrow\mathscr{V}\), the identity must be the \(\mathscr{V}\)-valued hom relation of the poset
\begin{equation}
\operatorname{Id}_P\left(p,p^{\prime}\right)=
\begin{cases}
I,\quad p\leq_Pp^{\prime}, \\
\perp_{\mathscr{V}},\quad otherwise.
\end{cases}
\end{equation}
Thus, for $[0,\infty]$ with the ordering $\geq$, we have $\perp_{\mathscr{V}}=\infty$. 

From a bird's-eye perspective, it is quite obvious that the entire bicategorical machinery that led us to Theorem~\ref{thm.2} should be generalizable to the current $\mathscr{V}$-setup. Well, by fixing a quantale $\mathscr{V}$, it is possible to define a bicategory $\operatorname{\mathscr{V}-Rel}$ as
\begin{itemize}
    \item objects - sets $X$ (or even small categories)
    \item $1$-cells - $\mathscr{V}$-relations $R: X\times Y\longrightarrow\mathscr{V}$,
    \item $2$-cells - pointwise order $R\implies R^{\prime}\iff\forall x,y: R\left(x,y\right)\leq R^{\prime}\left(x,y\right)$,
    \item composition - as in~(\ref{eq.162}),
    \item identities - as in~(\ref{eq.165}).
\end{itemize}

This leads us to the following proposition
\begin{prop}\label{prop.3}
For every unital quantale \(\mathscr{V}\), the quantale-valued relation construction \(\mathscr{V}-\operatorname{Rel}\) is a locally posetal bicategory.    
\end{prop}
This is essentially a standard fact, as quantaloid-enriched category theory has long treated enriched distributors and the corresponding locally ordered structures \cite{stubbe2005categorical}. 

Considering the behavior of the entropy orders, we should be able to add them directly via extrapolation from $\operatorname{OEnt}$. Specifically, in $\operatorname{OEnt}$, the objects were posets and $1$-cells were required to be monotone. A similar idea can be applied to $\mathscr{V}$-relations.  This brings us to the following definition
\begin{definition}
Let $P$, $Q$ be posets. A $\mathscr{V}$-open interface
\begin{equation}\label{eq.166}
\Phi:P\nrightarrow Q,
\end{equation}
is a function 
\begin{equation}
\Phi:P^{op}\times Q\longrightarrow\mathscr{V}.
\end{equation}
Equivalently, \(\Phi\) is a map \(\Phi:P\times Q\longrightarrow\mathscr{V}\) satisfying
\begin{enumerate}
    \item Contravariance in \(P\):
    \begin{equation}
    p\leq_P p^{\prime}\implies \Phi\left(p^{\prime},q\right)\leq_{\mathscr{V}}\Phi\left(p,q\right),
    \end{equation}
    \item Covariance in \(Q\):
    \begin{equation}
    q\leq_Qq^{\prime}\implies \Phi\left(p,q\right)\leq_{\mathscr{V}}\Phi\left(p,q^{\prime}\right).
    \end{equation}
\end{enumerate}
\end{definition}

\vspace{0.1cm}
Again, we can encapsulate this into a bicategory $\operatorname{\mathscr{V}-OEnt}$ such that
\begin{itemize}
    \item objects are entropy posets, i.e. (open) entropy systems
    \item $1$-cells are $\mathscr{V}$-open interfaces, i.e. order-compatible $\mathscr{V}$-profunctors \(P\nrightarrow Q\),
    \item $2$-cells are pointwise inequalities
    \item composition is given as in~(\ref{eq.162}), i.e. inherited from \(\mathscr{V}-\)Rel
    \item identities are the order-based hom-interfaces encoded in $\mathscr{V}$ that are given by
    \begin{equation}
    \operatorname{Id}_P\left(p,p^{\prime}\right)=
    \begin{cases}
    I,\quad p\leq_P p^{\prime} \\[1ex]
    \perp_{\mathscr{V}},\quad otherwise
    \end{cases}
    \end{equation}
\end{itemize}
It is crucial to note that for \(\mathscr{V}_{cost}\), the order is reversed, i.e. \(\alpha\leq_{\mathscr{V}}\beta\) means \(\alpha\geq\beta\) numerically. Hence, the variance conditions should be read numerically as cost monotonicity statements.  

This actually leads us to a new conjecture.
\newtheorem{conj}{Conjecture}
\begin{conj}\label{conj.1}
\textit{Let \(\mathscr{V}\) be a unital quantale. The entropy posets and order-compatible \(\mathscr{V}\)-valued profunctors \(P^{op}\times Q\longrightarrow\mathscr{V}\), ordered pointwise, form a locally posetal bicategory \(\mathscr{V}-\operatorname{OEnt}\).}    
\end{conj}
Nonetheless, this should be very easy to prove by going through all the bicategorical axioms. Moreover, one can immediately notice that our choice of the quantale $\mathscr{V}$-setup is meaningful since, in Conjecture~\ref{conj.1} we can recover the original locally posetal bicategory $\operatorname{OEnt}$ by choosing $\mathscr{V}=\operatorname{Bool}$.

Just as it is possible to physically interpret our original setup in $\operatorname{OEnt}$, we can do the same in the case of the $\mathscr{V}$-setup, since that is precisely the original purpose of this generalization. In other words, this formalizes our earlier statement that the Boolean theory represents something like a feasibility shadow of the weighted quantale theory. 

In the Boolean case, the relation $\Phi\left(b,d\right)=\top$ literally means that $b$ can be realized by $d$. Let us now specialize to the cost quantale \(\mathscr{V}_{cost}=\{[0,\infty],\geq,+,0,\inf\}\). Let \(B\) be a boundary entropy poset and \(D\) a bulk entropy poset. Then a costed interface \(\Phi:B\nrightarrow D\) assigns to each pair \(b,d\) a value \(\Phi\left(b,d\right)\in [0,\infty]\) that could naturally be interpreted in the following way:
\begin{tcolorbox}
\begin{center}
\(\Phi\left(b,d\right)=\)minimal dissipation cost of realizing the boundary state \(b\) by the bulk state \(d\), subject to the chosen admissible operations.
\end{center}
\end{tcolorbox}

Analyzing the concrete values of $\Phi\left
(b,d\right)$, we have
\begin{itemize}
    \item $\Phi\left(b,d\right)=0$ denotes an ideally reversible realization,
    \item $\Phi\left(b,d\right)\in\left(0,\infty\right)$ means realizable but necessarily dissipative,
    \item $\Phi\left(b,d\right
    )=\infty$ - impossible (forbidden) to realize under admissible operations.
\end{itemize}

Furthermore, if \(\Psi:D\nrightarrow B\) is a costed abstraction interface, then the composite \(T_{cost}:=\Psi\circ\Phi:B\nrightarrow B\) is given by
\begin{equation}
T_{cost}\left(b,b^{\prime}\right)=\inf_{d\in D}\left(\Phi\left(b,d\right)+\Psi\left(d,b^{\prime}\right)\right).
\end{equation}
Therefore, \(T_{cost}\left(b,b^{\prime}\right)\) corresponds to the minimal bulk-mediated round trip cost of starting from \(b\), using the bulk as an intermediate realization, and returning to the boundary as \(b^{\prime}\). Likewise, 
\begin{equation}
U_{cost}:=\Phi\circ\Psi:D\nrightarrow D,
\end{equation}
would correspond to the minimal bulk to boundary to bulk round trip cost. 

At first sight, one might try to define a quantitative Landauer adjunction simply by the enriched analogy of the Boolean inequalities \(\operatorname{Id}_B\leq_{\mathscr{V}}\Psi\circ\Phi\), \(\Phi\circ\Psi\leq_{\mathscr{V}}\operatorname{Id}_D\). However, in the Lawvere quantale, this must be interpreted with significant care. Since \(a\leq_{\mathscr{V}}b\iff a\geq b\) numerically, the first inequality \(\operatorname{Id}_B\leq_{\mathscr{V}}\Psi\circ\Phi\) must be interpreted as
\begin{equation}
0\geq\left(\Psi\circ\Phi\right)\left(b,b^{\prime}\right),
\end{equation}
on an ordered pair \(b\leq_B b^{\prime}\). Since all costs are nonnegative, this naturally forces
\begin{equation}
\left(\Psi\circ\Phi\right)\left(b,b^{\prime}\right)=0.
\end{equation}
Therefore, a strict enriched adjunction in \(\mathscr{V}_{cost}\) expresses an idealized zero-defect compatibility condition.

Accordingly, the physical interpretation of such a quantitative extension should mainly separate two things:
\begin{enumerate}
    \item The structural feasibility layer, encoded by the Landauer adjunction in \(\operatorname{OEnt}\) (in the Boolean sense)
    \item The cost layer, encoded by costed interfaces and cost lower-bound profiles.
\end{enumerate}

Accordingly, the physical interpretation of the proposed extension could be embedded into something like a conservative quantitative Landauer structure. Thus, we formulate the so-called cost-decorated Landauer data as a formalization of our cost extension within quantales. 

\begin{definition}\label{def.23}(\textbf{Cost-decorated Landauer data})
A cost-decorated Landauer structure consists of:
\begin{enumerate}
    \item A Landauer (Boolean) adjunction
    \begin{equation}
    \Phi_0:B\nrightarrow D,\quad\Psi_0:D\nrightarrow B,
    \end{equation}
    such that \(\Phi_0\dashv\Psi_0\) in \(\operatorname{OEnt}\),
    \item Costed interfaces
    \begin{equation}
    \widehat{\Phi}:B\nrightarrow D,\quad \widehat{\Psi}:D\nrightarrow B,
    \end{equation}
    in the enriched \(\mathscr{V}_{cost}-\operatorname{OEnt}\),
    \item A support condition
    \begin{equation}
    \widehat{\Phi}\left(b,d\right)<\infty\iff\Phi_0\left(b,d\right)=\top,\quad\widehat{\Psi}\left(d,b\right)<\infty\iff\Psi_0\left(d,b\right)=\top.
    \end{equation}
\end{enumerate}
\end{definition}
This definition basically tells us that the Boolean interfaces record which realizations are admissible, while the costed interfaces record how expensive the admissible realizations are. 

Notice that the importance of Definition~\ref{def.23} is reflected in the fact that it provides a separation of feasibility from dissipation. Thus, it avoids trying to force both roles into the same strict enriched adjunction. 

Therefore, the costed round trip \(T_{cost}\) can be understood as the  quantitative analogy of the Boolean Landauer closure \(T=\Psi\circ\Phi\). In the Boolean case, \(T\left(b,b^{\prime}\right)\) records whether \(b^{\prime}\) is reachable from \(b\) through a feasible bulk round trip. In the costed case, \(T_{cost}\left(b,b^{\prime}\right)\) records the minimal dissipation cost of such a round trip. This naturally leads us to the costed adaptation of the Landauer monad.

\begin{definition}(\textbf{Costed boundary closure interface)}\label{def.24}
Given cost-decorated Landauer data, the costed boundary closure interface can be defined as
\begin{equation}
T_{cost}:=\widehat{\Psi}\circ\widehat{\Phi}:B\nrightarrow B.
\end{equation}
\end{definition}

This implies that 
\begin{equation}
T_{cost}\left(b,b^{\prime}\right)=\inf_{d\in D}\left(\widehat{\Phi}\left
(b,d\right)+\widehat{\Psi}\left(d,b^{\prime}\right)\right),
\end{equation}
can be perceived as the minimal dissipation cost of the boundary transition from \(b\) to \(b^{\prime}\) mediated by the bulk.

Specifically, this should be interpreted as a costed closure interface, not as a monad in the strict Boolean sense, unless we impose some additional hypotheses. From a physical point of view, the objects introduced in Definitions~\ref{def.23} and~\ref{def.24} enable us to keep track of how much entropy production is forced, instead of only tracking how the entropy increases or decreases.

This naturally induces an issue of a lower bound on the cost of the physical implementation. One may encode this via a suitable boundary interface.

\begin{definition}(\textbf{Landauer lower-bound profile})\label{def.25}
A Landauer lower-bound profile is a costed boundary interface \(\Lambda_B:B\nrightarrow B\) such that, for every allowed boundary degradation \(b\leq_Bb^{\prime}\), we have
\begin{equation}
\Lambda_B\left(b,b^{\prime}\right)\in [0,\infty].
\end{equation}
\end{definition}

The direct interpretation of \(\Lambda_B\left(b,b^{\prime}\right)\) is the lower bound for physically implementing the transition from \(b\) to \(b^{\prime}\). 

Evidently, this gives us a tool for reasonably working with the minimal dissipation cost \(T_{cost}\) because the lower-bound profile \(\Lambda_B\) provides its natural limitation
\begin{equation}\label{eq.196}
T_{cost}\leq_{\mathscr{V}}\Lambda_B.
\end{equation}
Such an inequality represents the desired quantitative requirement for the minimal dissipation cost. However, it is important to recall that \(\leq_{\mathscr{V}}\) is a reverse numerical ordering in the quantale sense. Hence, the relation~(\(\ref{eq.196}\)) is equivalent to
\begin{equation}
\Lambda_B\left(b,b^{\prime}\right)\leq T_{cost}\left(b,b^{\prime}\right),
\end{equation}
in a numerical sense. Thus, \(\Lambda_B\) can be considered to be a genuine lower bound for the minimal round trip cost. 

Recalling Definition~\ref{def.23} of the cost-decorated Landauer data, there is still a legitimate role for its strict enriched adjunction variant, because we can analogously define a strict \(\mathscr{V}_{cost}\)-Landauer adjunction as a pair of costed interfaces \(\widehat{\Phi}:B\nrightarrow D\)  and \(\widehat{\Psi}:D\nrightarrow B\) satisfying \(\operatorname{Id}_B\leq_{\mathscr{V}}\widehat{\Psi}\circ\widehat{\Phi}\) and \(\widehat{\Phi}\circ\widehat{\Psi}\leq_{\mathscr{V}}\operatorname{Id}_D\). Here, the unit \(\operatorname{Id}_B\leq_{\mathscr{V}}\widehat{\Psi}\widehat{\Phi}\) imply that for \(b\leq_B b^{\prime}\), we have \(0\geq_{\mathbb{R}}\left(\widehat{\Psi}\widehat{\Phi}\right)\left(b,b^{\prime}\right)=0\). This should be interpreted in a way that the strict unit forces zero round-trip cost on a native boundary comparisons. However, the counit \(\widehat{\Phi}\widehat{\Psi}\leq_{\mathscr{V}}\operatorname{Id}_D\) does something different. For \(d\leq_{\mathscr{V}}d^{\prime}\), the right-hand side has value \(0\), so the numerical inequality is simply \(\left(\widehat{\Phi}\widehat{\Psi}\right)\left(d,d^{\prime}\right)\geq 0\), which is automatic. Thus, it does not force zero cost. Additionally, whenever \(d\nleq_D d^{\prime}\), we get \(\operatorname{Id}_D\left(d,d^{\prime}\right)=\infty\), so the counit forces \(\left(\widehat{\Phi}\widehat{\Psi}\right)\left(d,d^{\prime}\right)= \infty\). In other words, the strict counit forbids finite-cost bulk round trips that violate the native bulk order, so it does not force every allowed bulk round trip to have zero cost. In conclusion, this is just an alternative way to perceive the enriched Landauer adjunction.

As part of the proposed extensions and future research aims, it would be appropriate to summarize and formalize our previous considerations into two structured directions, in addition to the already proposed Conjecture~\ref{conj.1}.

\vspace{0.2cm}
\hspace{-0.7cm}\textbf{\underline{Open problem 1 - Enriched Eilenberg-Moore theory}}

\vspace{0.1cm}
In the Boolean case, the boundary endo-interface \(T=\Psi\circ\Phi\) admits an Eilenberg-Moore object describing the boundary data stable under the bulk-mediated round trip. One may ask for an analogous structure in the quantale-valued setup. 

The aim would be to develop an Eilenberg-Moore theory for suitable idempotent or split costed endo-interfaces \(T:B\nrightarrow B\) in the \(\mathscr{V}-\operatorname{OEnt}\) setup and to identify conditions under which there exists an object representing \(\mathscr{V}\)-algebras
\begin{equation}
x:X\nrightarrow B\quad with\quad T\circ x\leq_{\mathscr{V}}x.
\end{equation}

One can make this aim more technical and compact in the following way. Let \(\mathscr{V}\) be a unital quantale and \(B\) be an object of \(\mathscr{V}-\operatorname{OEnt}\). A monad on \(B\) consists of an endo-\(1\)-cell \(T:B\nrightarrow B\) together with \(2\)-cells \(\operatorname{\eta}_T:\operatorname{Id}_B\implies T\),\(\mu_T:T\circ T\implies T\), satisfying the usual monad axioms. Since the hom-categories are posets, these structural cells may equivalently be written as \(\operatorname{Id}_B\leq_{\mathscr{V}} T, \quad T\circ T\leq_{\mathscr{V}}T\). We call \(T\) idempotent when \(\mu_T\) is invertible. In particular, for the idempotent case, we ask for conditions ensuring the equality \(T\circ T=T\). Then, for every object \(X\), we can define the poset of \(T\)-algebras out of \(X\) by
\begin{equation}
\operatorname{Alg_T}\left(X\right):=\{x:X\nrightarrow B\hspace{0.1cm}\vert\hspace{0.1cm}T\circ x\leq_{\mathscr{V}} x\},
\end{equation}
ordered by the hom-order in \(\mathscr{V}-\operatorname{OEnt}\). Open problem \(1\) basically asks the following question:
\begin{tcolorbox}
\begin{center}
Is it possible to determine sufficient conditions under which there exists an object \(B_T\), a \(1\)-cell \(i:B_T\nrightarrow B\) and a \(2\)-cell \(\alpha:T\circ i\implies i\) such that, for every object \(X\), postcomposition with \(i\) induces a natural isomorphism of posets
\begin{equation}
\nonumber
\mathscr{V}-\operatorname{OEnt}\left(X,B_T\right)\cong\operatorname{Alg_T}\left(X\right)\hspace{0.3cm}?
\end{equation}
\end{center}
\end{tcolorbox}
Equivalently speaking, can we determine when the idempotent \(\mathscr{V}\)-open monad \(T\) admits an Eilenberg-Moore object in the bicategory \(\mathscr{V}-\operatorname{OEnt}\)?

One can find several indications that the problem should admit a meaningful solution. For example, even in the \(\mathscr{V}\)-valued setup, the bicategory \(\mathscr{V}-\operatorname{OEnt}\) is expected to be locally posetal. Thus, the algebra condition \(T\circ x\leq_{\mathscr{V}} x\) is still an order-theoretic condition in a hom-poset. This is the reason why the bicategorical formulation of Eilenberg-Moore objects remains meaningful. 

However, there are also several reasons why this problem is quite nontrivial. Unlike the Boolean case, a \(\mathscr{V}\)-valued idempotent interface does not immediately determine an ordinary quotient order on the underlying set of boundary states. The main difficulty is therefore to identify the correct enriched replacement for the construction mechanism of the Boolean fixed point/stable sector. 

\vspace{0.3cm}
\hspace{-0.7cm}\textbf{\underline{Open problem 2 - Enriched Cauchy completion phenomenon}}\label{2}

\vspace{0.1cm}
The main aim of this open problem is to relate the Eilenberg-Moore object to the splitting \(\mathscr{V}\)-idempotents and hence to a \(\mathscr{V}\)-enriched Cauchy completion. In order to do so, one would need an enriched theory that simultaneously handles:
\begin{itemize}
    \item The quantale-valued matrix bicategory,
    \item The order-compatible variance \(P^{op}\times Q\longrightarrow\mathscr{V}\),
    \item The appropriate notion of idempotent splitting in such order/enriched setup.
\end{itemize}

In order to make these points more precise, one should consider a unital quantale \(\mathscr{V}\) and assume that \(B\) is some object of \(\mathscr{V}-\operatorname{OEnt}\). Next, suppose that \(T:B\nrightarrow B\) is an idempotent \(\mathscr{V}\)-open endo-interface, in the sense that \(T\circ T=T\) in the hom-poset \(\mathscr{V}-\operatorname{OEnt}\left
(B,B\right)\). Then we would aim to develop a \(\mathscr{V}\)-enriched categorical framework in which:
\begin{enumerate}
    \item The object \(B\) is represented by a \(\mathscr{V}\)-category \(\mathscr{B}\) compatible with the entropy order on \(B\),
    \item The endo-interface \(T:B\nrightarrow B\) is regarded as a \(\mathscr{V}\)-module \(T:\mathscr{B}\nrightarrow\mathscr{B}\),
    \item Characterize when \(T\) is Cauchy, i.e. when it admits a right adjoint as a \(\mathscr{V}\)-module,
    \item The splitting of \(T\) in the bicategory of \(\mathscr{V}\)-modules is represented by an object or full \(\mathscr{V}\)-subcategory of the Cauchy completion \(\widehat{B}_C\) of \(\mathscr{B}\),
    \item The corresponding splitting object represents precisely the \(T\)-stable \(\mathscr{V}\)-modules
    \begin{equation}
    \nonumber
    x:X\nrightarrow B\quad \textit{such that}\quad T\circ x\leq x.
    \end{equation}
\end{enumerate}

In particular, determine the conditions under which the Eilenberg-Moore object for an idempotent costed Landauer monad \(T\) exists and is equivalent to the splitting of \(T\) inside the enriched Cauchy completion of the boundary \(\mathscr{V}\)-category. Actually, there is a broader, interesting reason for considering such a splitting. Since the costed interfaces are \(\mathscr{V}\)-profunctorial rather than functorial, the relevant enriched notion is not merely the Cauchy completion of a \(\mathscr{V}\)-category in isolation, but the splitting of idempotent \(\mathscr{V}\)-modules in the bicategory of \(\mathscr{V}\)-modules.

In summary, the open problem \(2\) is really asking:
\begin{tcolorbox}\label{open.problem 2}
\begin{center}
Can the enriched analogue of \(B_T\) be realized as the object representing the split of the idempotent costed Landauer module \(T\) inside \(\widehat{B}_C\)?   
\end{center}
\end{tcolorbox}

\hspace{-0.7cm}\textbf{\underline{Open problem 3 - Thermodynamic splitting of costed Landauer modules}}\label{3}

\vspace{0.2cm}
The Boolean theory developed above associates a Landauer adjunction \(\Phi_0:B\nrightarrow D\), \(\Psi_0:D\nrightarrow B\) with an idempotent boundary interface \(T_0:=\Psi_0\circ\Phi_0:B\nrightarrow B\). This interface determines the preorder
\begin{equation}
b\preceq_{T_0}b^{\prime}\iff T_0\left(b,b^{\prime}\right)=\top,
\end{equation}
and consequently, the stable quotient
\begin{equation}
B_{T_0}= B/_{\sim T_0}.
\end{equation}

In Theorem~\ref{thm.6}, we prove that this quotient is not merely an order-theoretic construction. It represents all boundary interfaces that are stable under the Landauer round trip
\begin{equation}
\operatorname{OEnt}\left
(X,B_{T_0}\right)\cong\operatorname{Alg}_{T_0}\left(X\right),
\end{equation}
naturally in \(X\). Therefore, in the Boolean case, the round trip stability is represented by a universal object. 

The cost-enriched theory leads us to the analogous, but substantially more difficult question. Let \(\left(\mathscr{V},\leq_{\mathscr{V}},\otimes,I\right)\) be the unital quantale, and let \(B\) be a \(\mathscr{V}\)-enriched boundary system. Next, a costed boundary endo-interface is a \(\mathscr{V}\)-module \(T:B\nrightarrow B\). For the cost quantale \(\mathscr{V}_{cost}=\left([0,\infty],\geq,+,0\right)\) , the value \(T\left(b,b^{\prime}\right)\) is interpreted as the minimal dissipation cost of producing the boundary description \(b^{\prime}\) from \(b\) through an admissible bulk-mediated round trip. If \(T=\Psi\circ\Phi\), then composition is given by infimal convolution
\begin{equation}
T\left(b,b^{\prime}\right)=\inf_{d\in D}\left(\Phi\left
(b,d\right)+\Psi\left(d,b^{\prime}\right)\right).
\end{equation}

The Boolean interface records whether such a round trip exists, and the costed interface records the least dissipation among all feasible bulk mediators. 

Considering the Boolean case, an idempotent interface determines an ordinary preorder and hence a quotient set. This mechanism does not directly generalize to the costed case, since any two boundary states need not be simply equivalent or inequivalent. They may instead be related by different finite transition costs. Consequently, the enriched stable object should not generally be expected to consist of ordinary equivalence classes of boundary states. Its elements may instead be understood as stable \(\mathscr{V}\)-valued profiles, enriched presheaves, or generalized boundary sectors that retain both reachability and dissipation information. This directly motivates us to propose the following open problem.

Let \(\mathscr{V}\) be a commutative unital quantale, let \(B\) be a small ordered-compatible \(\mathscr{V}\)-enriched boundary system, and let \(T:B\nrightarrow B\) be a \(\mathscr{V}\)-module satisfying the following conditions:
\begin{equation}
\operatorname{Id}_B\leq_{\mathscr{V}}T,\quad T\circ T=T,
\end{equation}
where \(T\) is Cauchy, meaning that it admits a right adjoint as a \(\mathscr{V}\)-module. 

\begin{tcolorbox}
\textbf{\underline{Splitting conjecture:}} Does there exist an order-compatible \(\mathscr{V}\)-enriched object \(B_T\) and \(\mathscr{V}\)-modules \(r:B\nrightarrow B_T\), \(i:B_T\nrightarrow B\) such that
\begin{equation}
r\dashv i,\quad i\circ r=T\quad\textit{and}\quad r\circ i=\operatorname{Id}_{B_T}\hspace{0.1cm}?
\end{equation}
\end{tcolorbox}

Moreover, we should ask whether for every order-compatible \(\mathscr{V}\)-enriched system \(X\), postcomposition induces a natural isomorphism of hom-posets
\begin{equation}
\mathscr{V}-\operatorname{OEnt}\left(X,B_T\right)\cong\operatorname{Alg}_T\left(X\right),
\end{equation}
where \(\operatorname{Alg}_T\left(X\right):=\{x:X\nrightarrow B\hspace{0.1cm}\vert\hspace{0.1cm}T\circ x\leq_{\mathscr{V}}x\}\).

Assuming that this open problem is true, we could equivalently claim that the splitting object \(B_T\) represents all \(T\)-stable costed boundary interfaces. 

Furthermore, \(B_T\) may be chosen within the enriched Cauchy completion of \(B\), and its underlying Boolean support should recover the stable sector determined by the Boolean feasibility interface associated with \(T\).

Let us now analyze the meaning of the assumptions above, since each one has a distinct mathematical and physical role. First, the unit inequality \(\operatorname{Id}_B\leq_{\mathscr{V}}T\) represents the enriched analogy of extensivity. In particular, it says that the native boundary comparison is contained in the comparison generated by the round trip. In the Boolean case, this means
\begin{equation}
b\leq_B b^{\prime}\implies T\left(b,b^{\prime}\right)=\top.
\end{equation}
However, in the strict cost quantale sense, this condition must be interpreted carefully. Since 
\begin{equation}
\alpha\leq_{\mathscr{V}_{cost}}\beta\implies \alpha\geq\beta,
\end{equation}
numerically, a strict unit may require the relevant identity transitions to have zero additional cost. Thus, the strict enriched-monad condition describes a normalized or zero-defect regime. 

From a physical point of view, this is quite restrictive. A realistic thermodynamic round trip may possess an unavoidable positive baseline dissipation, even when no additional information is lost. Therefore, the strict quantale conjecture should be understood as the first normalized form of the problem. Further generalization should replace the strict identity baseline with a prescribed dissipation profile or by a lax/defective unit. 

Proceeding to the idempotence, we presume the equality \(T\circ T=T\) has a direct cost interpretation. For \(\mathscr{V}_{cost}\), it becomes
\begin{equation}\label{eq.235}
T\left(b,b^{\prime}\right)=\inf_{c\in B}\left(T\left(b,c\right)+T\left(c,b^{\prime}\right)\right).
\end{equation}
The right-hand side is the least cost obtained by decomposing the process into two round trips through an intermediate boundary description \(c\). In particular, the effective cost \(T\left(b,b^{\prime}\right)\) is already saturated under decomposition through arbitrary intermediate boundary states. Allowing one further compositional stage does not reduce the infimal effective cost. This is precisely what the idempotence of \(T\) states, i.e. when allowing a second bulk-mediated round trip, it does not produce a cheaper effective transition than the cost already recorded by \(T\). 

In summary, this corresponds to the quantitative analogy of the Boolean statement that a second closure step introduces no new feasible comparisons. Hence, \(T\) already satisfies its own optimality equation. In thermodynamic language, the cost interface has already been saturated under all admissible round-trip concatenations. 

At last, we shall focus on the Cauchy condition. The idempotence alone does not ensure that a \(\mathscr{V}\)-module is represented by an actual enriched object. The Cauchy condition supplies precisely the representability needed for such a splitting. More concretely, a \(\mathscr{V}\)-module is Cauchy when it has a right adjoint in the bicategory of \(\mathscr{V}\)-modules. This is the enriched analogy of an absolute or representable weight. This is precisely what we want because the Cauchy completion is designed to adjoint objects representing such modules. Therefore, in the open problem/conjecture presented above, we do not merely ask whether an abstract idempotent exists. We ask whether the thermodynamic idempotent is sufficiently representable to determine a genuine enriched stable sector.

Summing up all the above considerations, we expect the following mechanism:
\begin{tcolorbox}
\begin{center}
Cauchy idempotent module \(\implies\) splitting in the enriched Cauchy completion \(\implies\) universal thermodynamic stable object
\end{center}
\end{tcolorbox}
However, the subtle point is that the splitting should remain compatible with the entropy/information orders and with the physical interpretation of the costs. This actually implies that a splitting in the unrestricted bicategory of \(\mathscr{V}\)-modules is not by itself enough, so it may lie in the thermodynamically admissible sub-bicategory. De facto, it is possible to formulate a central open problem as a unification of the open problems 2 and 3:
\begin{tcolorbox}
\underline{\textbf{Thermodynamically admissible enriched splitting:}} Let \(T:B\nrightarrow B\) be an idempotent Cauchy \(\mathscr{V}\)-module associated with a costed Landauer structure. Determine conditions under which its enriched Cauchy splitting can be realized inside the order-compatible thermodynamic bicategory \(\mathscr{V}-\operatorname{OEnt}\), while preserving the intended entropy/cost semantics and recovering the Boolean stable quotient under passage to Boolean support.
\end{tcolorbox}

\vspace{0.3cm}
\hspace{-0.7cm}\textbf{Why is this interesting?}

\vspace{0.1cm}
In conclusion, let us now summarize the exact reasons why one should consider the third open problem interesting. First of all, it represents the desired quantitative analogy of Theorem~\ref{thm.6}. The Boolean construction remembers only whether the transition \(b\longrightarrow b^{\prime}\) is feasible. However, the enriched construction should remember how expensive it is, taking into account the entropy/thermodynamic cost. The desired object \(B_T\) would therefore refine the Boolean stable quotient from the round-trip indistinguishability to the round-trip stability together with optimal dissipation. 

The Boolean object \(B_{T_0}\) identifies states that become mutually reachable. On the other hand, the enriched object should distinguish states or profiles when the transition costs between them differ, even if both directions are feasible. In this sense, it would consequently encode strictly more physical information than the ordinary quotient.   

Next, the solution to this open problem would provide a universal meaning to thermodynamic stability. Without a universal object, one can define a stable interface by the inequality \(T\circ x\leq_{\mathscr{V}}x\), but this remains only a condition imposed separately on every probe \(x\). Thereupon, the Eilenberg-Moore property would say that all such probes are represented by one object \(\mathscr{V}-\operatorname{Oent}\left(X,B_T\right)\cong\operatorname{Alg}_T\left(X\right)\). Thus, every \(T\)-stable observation of the boundary would factor uniquely through \(B_T\). Physically, this means that \(B_T\) is not merely one convenient reduction of the boundary,
\begin{tcolorbox}
\begin{center}
\(B_T\) is the universal carrier of all boundary information whose cost profile is already stable under bulk implementation and readout.
\end{center}
\end{tcolorbox}

Returning back to the formula~(\ref{eq.235}), it is an optimization problem for \(T\left(b,b^{\prime}\right)\). By itself, it determines the cheapest bulk mediator for one transition. In the case of the splitting conjecture, we ask whether the entire family of such optimization problems can be organized into a universal compositional object. In other words, the third open problem transforms the optimization into a compositional structure. If the answer to the splitting conjecture is positive, then minimizing the dissipation is not merely an external numerical calculation performed after the categorical construction. It unavoidably becomes part of the universal categorical semantics. From a physical perspective, this is particularly relevant for multistage thermodynamic processes \cite{Aurell2011,Seifert2012}. Technically, the infimal convolution composes costs, while idempotence expresses the stability of the resulting optimum under further composition.  

Another interesting fact about the consequences of the splitting conjecture is that the object \(B_T\) would depend only on \(T\), not on the particular factorization \(T=\Psi\circ\Phi\). It would therefore represent an invariant of the observable bulk-boundary coupling rather than of a chosen microscopic implementation. This mirrors the Boolean conclusion of Theorem~\ref{thm.6}, i.e. the stable sector is not the complete bulk, but the part of the coupling visible at the boundary. The enrichment construction would refine that conclusion by including minimal dissipation data. 

Lastly, the splitting conjecture would provide a mathematically precise bridge to the holographic reconstruction. In a holographic interpretation, several boundary descriptions may encode the same bulk logical information, but their reconstruction costs need not coincide. A Boolean quotient identifies equivalent reconstructions and completely forgets this energetic distinction. However, the enriched stable object could simultaneously record the following:
\begin{enumerate}
    \item which boundary descriptions reconstruct the same sector,
    \item the minimal thermodynamic cost of moving between their reconstructions.
\end{enumerate}
This would eventually yield a substantially richer form of bulk-visible reconstruction, and it suggests that the relevant holographic object need not be only a quotient of boundary states. It may be an enriched space of stable reconstruction profiles.

\section{Conclusion and outlook}\label{sec.4}

We develop a compositional formulation of the order-theoretic Landauer connection using the locally posetal bicategory \(\operatorname{OEnt}\). Its objects are entropy-ordered state spaces, its \(1\)-cells are Boolean-valued profunctors encoding open, many-to-many implementation interfaces, and its \(2\)-cells correspond to pointwise refinements. Composition is given by profunctorial composition and therefore models the gluing of implementation stages while existentially hiding intermediate thermodynamic degrees of freedom. 

Within this bicategorical framework, we formulate the Landauer connection as an exact adjunction
\begin{equation}
\Phi\dashv\Psi,\quad \Phi:B\nrightarrow D,\quad
\Psi: D\nrightarrow B,
\end{equation}
with 
\begin{equation}
\operatorname{Id}_B\leq\Psi\circ\Phi,\quad \Phi\circ\Psi\leq\operatorname{Id}_D.
\end{equation}
The induced composites \(T:=\Psi\circ\Phi\) and \(U:=\Phi\circ\Psi\) define, respectively, the boundary Landauer monad and the dual bulk comonad. In particular, 
\begin{equation}
\operatorname{Id}_B\leq T,\quad T\circ T=T,
\end{equation}
so the boundary round trip defines an idempotent enlargement of the native boundary accessibility relation. 

Our second result concerns representability. Although \(\operatorname{OEnt}\) contains genuinely non-representable open interfaces, every exact Landauer adjunction lies in the representable sector. Specifically, there exists a unique monotone map \(F:B\longrightarrow D\) such that \(\Phi=F_*\) and \(\Psi=F^*\). Hence, the passage from monotone maps to arbitrary profunctors enlarges the ambient space of admissible implementations, whereas the exact adjunction condition selects a distinguished deterministic subclass recovering the classical order-theoretic Landauer connection.

From here, we obtain our main reconstruction theorem (see Theorem~\ref{thm.6}). Using the boundary monad \(T=\Psi\circ\Phi\) and defining 
\begin{equation}
b\preceq_T b^{\prime}\iff T\left(b,b^{\prime}\right)=\top
\end{equation}
and
\begin{equation}
b\sim_T b^{\prime}\iff b\preceq_T b^{\prime}\land b^{\prime}\preceq_T b,
\end{equation}
we construct the quotient poset \(B_T:=B/{\sim_{T}}\). Then we prove that this quotient explicitly splits the idempotent
\begin{equation}
i\circ r=T,\quad r\circ i=\operatorname{Id}_{B_T},
\end{equation}
for canonical interfaces
\begin{equation}
r:B\nrightarrow B_T,\quad i:B_T\nrightarrow B.
\end{equation}
Moreover, we prove that
\begin{equation}
\operatorname{OEnt}\left(X,B_T\right)\cong\operatorname{Alg}_T\left(X\right)
\end{equation}
naturally in \(X\). Therefore, \(B_T\) is simultaneously the quotient by round-trip indistinguishability, the splitting object of \(T\), and the Eilenberg-Moore object of the Landauer monad. Therefore, \(B_T\) represents the boundary information that remains stable under bulk-mediated implementation and readout. 

However, it is necessary to be careful here, since this reconstruction should not be interpreted as a recovery of the complete microscopic bulk \(D\). Rather, \(B_T\) is the bulk-visible boundary sector determined by the chosen Landauer coupling. When considering the representable case, the interpretation becomes especially clear because
\begin{equation}
b\sim_T b^{\prime}\iff F\left(b\right)=F\left(b^{\prime}\right).
\end{equation}
Thus \(B_T\cong\operatorname{Im}\left(F\right)\). Subsequently, the reconstructed sector is obtained by identifying exactly those boundary distinctions that are invisible to the implementation map. 

Next, through finite examples, we illustrate both aspects of the present theory. One exhibits a genuinely non-representable open interface with incomparable minimal feasible bulk realizations, while the other explicitly computes an exact Landauer adjunction, its boundary monad, the induced equivalence classes, and the reconstruction interfaces \(r\), \(i\).

Considering our future aims, the principal direction is the quantitative enrichment. Replacing \(\mathbf{Bool}\) with a suitable quantale, in particular, the Lawvere cost quantale \(\mathscr{V}_{cost}=\left([0,\infty],\geq,+,0\right)\), allows an interface to encode minimal dissipation rather than mere feasibility. The composition operation then becomes an infimal convolution. For example,
\begin{equation}
T_{cost}\left(b,b^{\prime}\right)=\inf_{d\in D}\bigl(\Phi_{cost}\left(b,d\right)+\Psi_{cost}\left(d,b^{\prime}\right)\bigr).
\end{equation}

It is quite evident that a strict enriched adjunction is generally too restrictive physically, because the reversed Lawvere order may force zero-cost round-trips. A more realistic theory should therefore separate feasibility from cost or replace exact adjunctions by lax, defective or cost-bounded variants.

The second open direction that we shall explore is the enriched analogue of the reconstruction theorem. One would like to determine when an idempotent costed endo-interface admits an Eilenberg-Moore object or a Cauchy splitting inside an order-compatible, thermodynamically admissible \(\mathscr{V}\)-enriched bicategory. Unlike the Boolean case, the corresponding stable object need not be an ordinary quotient of states. Instead, it may retain nontrivial profiles with transition costs. Ideally, its Boolean support should completely recover the stable quotient \(B_T\).

Summarizing our present work, we establish the order-theoretic layer of a possible broader compositional Landauer theory. Structurally, the Boolean profunctors encode feasibility, exact adjunctions determine the representable Landauer sector, and the induced idempotent round trip selects a universal bulk-visible boundary object. Future extension of the structure to quantitative dissipation and thermodynamically admissible enriched reconstruction shall provide a natural direction for further research development.

\newpage
\textbf{Acknowledgments}

\vspace{0.1cm}

I would like to mention several sources of support for my research. The first of them is the grant project Cartan supergeometries and Higher Cartan geometries No. 24-10887S, provided by the Czech Science Foundation. Additionally, I acknowledge local support from the project for Specific Research in Mathematics MUNI/A/1457/2023 for doctoral students, provided by the Department of Mathematics and Statistics at Masaryk University.

Finally, I would also like to express my gratitude to Radosław Antoni Kycia, PhD, who introduced me to the basics of this topic and provided useful recommendations for this manuscript, including potential applications of the main results.

\bibliographystyle{plainnat}
\bibliography{cas-refs}

\end{document}